\documentclass[envcountsame,envcountsect,runningheads]{llncs}

\usepackage{graphicx}
\usepackage{multirow}
\usepackage{latexsym}
\usepackage{amsmath,amssymb,amstext}
\usepackage{graphicx}
\usepackage{wasysym}
\usepackage{float}
\usepackage{comment}
\usepackage{enumerate}
\usepackage[arrow, matrix, curve]{xy}
\usepackage[justification=centering,singlelinecheck=off]{caption}
\usepackage{color}
\usepackage{url}
\usepackage{longtable}
\usepackage{wrapfig}
\usepackage{stmaryrd}
\usepackage{framed}
\definecolor{shadecolor}{rgb}{.95,.95,.7}
\usepackage{fixfoot}

\input{macros.def}

\begin{document} 

\title{Data-Oblivious Stream Productivity}
   
\author{J\"{o}rg Endrullis\inst{1} \and Clemens Grabmayer\inst{2} \and Dimitri Hendriks\inst{1}}
\institute{
    Vrije Universiteit Amsterdam,
    Department of Computer Science\\
    De Boelelaan 1081a,
    1081 HV Amsterdam,
    The Netherlands\\
    \email{joerg@few.vu.nl}
      \mbox{ }
    \email{diem@cs.vu.nl}
  \and
    Universiteit Utrecht,
    Department of Philosophy\\
    Heidelberglaan 8,
    3584 CS Utrecht,
    The Netherlands\\
    \email{clemens@phil.uu.nl}
}
\maketitle

\begin{abstract}
  We are concerned with demonstrating productivity of specifications of
  infinite streams of data, based on orthogonal rewrite rules. 
  In general, this property is undecidable, but for restricted formats
  computable sufficient conditions can be obtained. 
  The usual analysis, also adopted here, disregards the identity
  of data, thus leading to approaches that we call data-oblivious. 
  We present a method 
  that is provably optimal among
  all such data-oblivious approaches. This means that in order to improve on our algorithm 
  one has to proceed in a data-aware
  fashion.\footnote{%
  This research has been partially funded by 
  the Netherlands Organisation for Scientific Research (NWO) 
  under FOCUS/BRICKS grant number~642.000.502.}
\end{abstract}

\section{Introduction}\label{sec:intro}
For programming with infinite structures, productivity is
what termination is for programming with finite structures.
Productivity captures the intuitive notion of unlimited progress,
of `working' programs producing defined values indefinitely.
In functional languages, usage of infinite structures
is common practice.
For the correctness of programs dealing with such structures
one must guarantee that every finite part of the infinite structure can be evaluated,
that is, the specification of the infinite structure must be productive.

\begin{wrapfigure}[11]{r}{60mm}
  \vspace{-4.5ex}
  \scalebox{.8}{\includegraphics{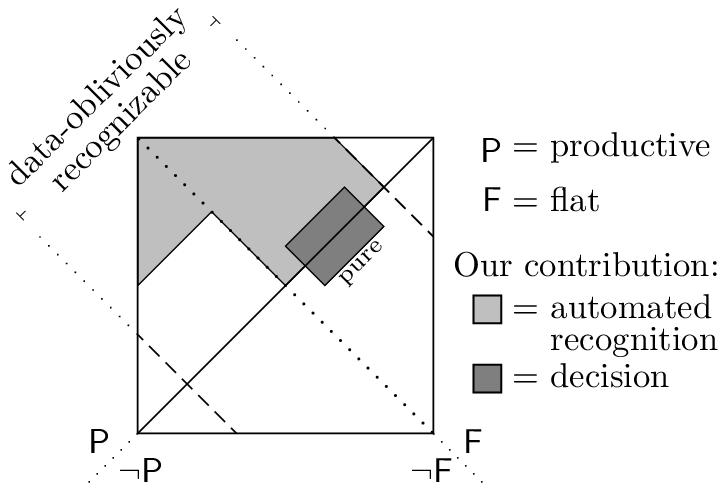}}
  \vspace{-2.5ex}
  \caption{Map of stream specifications}
  \label{fig:pnp}
\end{wrapfigure}
We investigate this notion for stream specifications,
formalized as orthogonal term rewriting systems.
Common to all previous approaches for recognizing productivity
is a quantitative analysis that abstracts away from the concrete values of stream elements.
We formalize this by a notion of `data-oblivious' rewriting, %
and introduce the concept of data-oblivious productivity.
Data-oblivious
(non-)\discretionary{}{}{}productivity implies
(non-)\discretionary{}{}{}productivity,
but neither of the converse implications holds.
Fig.~\ref{fig:pnp}
shows a Venn diagram of stream specifications,
highlighting the subset of `data-obliviously recognizable'
specifications
\pagebreak
where (non-)\discretionary{}{}{}productivity
can be recognized by a data-oblivious analysis.

We identify two syntactical classes of stream specifications:
`flat' and `pure' specifications, see the description below.
For the first we devise a decision algorithm for data-oblivious (\daob) productivity.
This gives rise to a computable, \daobly{} optimal,
criterion for productivity:
every flat stream specification that can be established
to be productive by whatever \daob{} argument
is recognized as productive by this criterion (see Fig.~\ref{fig:pnp}).
For the subclass of pure specifications,
we establish that \daob{} productivity coincides with productivity,
and thereby obtain a decision algorithm for productivity of this class.
Additionally, we extend our criterion
beyond the class of flat stream specifications,
allowing for `friendly nesting' in the specification of stream functions;
here \daob{} optimality is not preserved.

In defining the different formats of stream specifications,
we distinguish between rules for stream constants,
and rules for stream functions.
Only the latter are subjected to syntactic restrictions.
In flat stream specifications the defining rules for the stream functions
do not have nesting of stream function symbols;
however, in defining rules for stream constants
nesting of stream function symbols \emph{is} allowed.
This format makes use of exhaustive pattern matching on data
to define stream functions, allowing for
multiple defining rules for an individual stream function symbol.
Since the quantitative consumption/production behaviour of
a symbol $\astrfun$ might differ among its defining rules,
in a \daob{} analysis one has to settle for the use of lower bounds
when trying to recognize productivity.
If for all stream function symbols $\astrfun$ in a flat specification~$\atrs$
the defining rules for $\astrfun$ coincide,
disregarding the identity of data-elements, then $\atrs$ is called pure.

Our decision algorithm for \daob{} productivity
determines the tight \daob{} lower bound
on the production behaviour of every stream function,
and uses these bounds to calculate the \daob{} production of stream constants.
We briefly explain both aspects.
Consider the stream specification
$\strcf{A} \red \strcns{\datf{0}}{\funap{\strff{f}}{\strcf{A}}}$
together with the rules
$\funap{\strff{f}}{\strcns{\datf{0}}{\astr}} \to \strcns{\datf{1}}{\strcns{\datf{0}}{\strcns{\datf{1}}{\funap{\strff{f}}{\astr}}}}$, and
$\funap{\strff{f}}{\strcns{\datf{1}}{\astr}} \to \strcns{\datf{0}}{\funap{\strff{f}}{\astr}}$,
defining the stream
$\strcns{\datf{0}}{\strcns{\datf{1}}{\strcns{\datf{0}}{\strcns{\datf{1}}{\ldots}}}}$ of alternating bits.
The tight \daob{} lower bound for $\strff{f}$ is the function $\mit{id}$: $n \mapsto n$.
Further note that $\mit{suc}$:
$n\mapsto n+1$ %
captures the quantitative behaviour of the function prepending a data element to a stream term.
Therefore the \daob{} production of $\strcf{A}$
can be computed as $\lfp{\mit{suc}\circ\mit{id}} = \conattop$,
where $\lfp{f}$ is the least fixed point of $f\funin\conat\to\conat$ and
$\conat\defdby\nat\cup\{\infty\}$;
hence $\strcf{A}$ is productive.
As a comparison, only a `data-aware' approach
is able to establish productivity of %
$\strcf{B} \red \strcns{\datf{0}}{\funap{\strff{g}}{\strcf{B}}}$
with
$\funap{\strff{g}}{\strcns{\datf{0}}{\astr}} \to \strcns{\datf{1}}{\strcns{\datf{0}}{\funap{\strff{g}}{\astr}}}$, and
$\funap{\strff{g}}{\strcns{\datf{1}}{\astr}} \to \funap{\strff{g}}{\astr}$.
The \daob{} lower bound of $\strff{g}$ is $n \mapsto 0$, due to the latter rule.
This makes it impossible for any conceivable \daob{} approach
to recognize productivity of $\strcf{B}$.

We obtain the following results:
\begin{enumerate}
  \item \label{item:summary:flat}
    For the class of flat stream specifications
    we give a computable, \daobly{} optimal,
    sufficient condition for productivity.
  \item\label{item:summary:pure}
    We show decidability of productivity for
    the class of pure stream specifications,
    an extension of the format in~\cite{endr:grab:hend:isih:klop:2007}.
  \item\label{item:summary:friendly}
    Disregarding \daob{} optimality,
    we extend~\ref{item:summary:flat}
    to the bigger class of friendly nesting stream specifications.
  \item
    A tool automating~\ref{item:summary:flat}, \ref{item:summary:pure},
    and~\ref{item:summary:friendly},
    which can be downloaded from, and used via a web interface at:
    \small{\url{http://infinity.few.vu.nl/productivity}}. %
\end{enumerate}

\paragraph{Related work.}
Previous approaches~\cite{sijt:1989,hugh:pare:sabr:1996,telf:turn:1997,buch:2005}
employed \daob{} reasoning (without using this name for it)
to find sufficient criteria ensuring productivity,
but did not aim at optimality.
The \daob{} production behaviour of a stream function $\astrfun$
is bounded from below by a `modulus of production' $\sprdmod{\astrfun} \funin \nat^k\to\nat$
with the property that the first $\prdmod{\astrfun}{n_1,\ldots,n_k}$ elements of
$\funap{\astrfun}{t_1,\dots,t_k}$ can be computed whenever the first $n_i$ elements
of $t_i$ are defined.
\mbox{Sijtsma} develops an approach
allowing arbitrary production moduli
$\ssprdmod\funin\nat^k\to\nat$,
which, while providing an adequate mathematical description,
are less amenable to automation.
Telford and Turner~\cite{telf:turn:1997} employ production moduli of the form
$\funap{\ssprdmod}{n} = n + a$ with $a \in \mathbb{Z}$.
Hughes, Pareto and Sabry~\cite{hugh:pare:sabr:1996} use
$\funap{\ssprdmod}{n} = \max \{c \cdot x + d \where x \in \nat,\, n \ge a \cdot x + b\} \cup \{0\}$
with $a,b,c,d \in \nat$.
Both classes of production moduli are strictly contained in
the class of `periodically increasing' functions
which we employ in our analysis.
We show that the set of \daob{} lower bounds of flat stream function specifications
is exactly the set of periodically increasing functions.
Buchholz~\cite{buch:2005} presents a type system for productivity,
using unrestricted production moduli.
To obtain an automatable method for a restricted class of stream specifications,
he devises a syntactic criterion to ensure productivity.
This criterion easily handles all the examples of~\cite{telf:turn:1997},
but fails to deal with functions that have a negative effect like~$\sstrod$
defined by $\strod{\strcns{x}{\strcns{y}{\astr}}}\to\strcns{y}{\strod{\astr}}$
with a (periodically increasing)
modulus $\prdmod{\!\sstrod}{n} = \lfloor\frac{n}{2}\rfloor$.

\paragraph{Overview.}
In Sec.~\ref{sec:class} we define the notion of stream specification,
and the syntactic format of flat and pure specifications.
In Sec.~\ref{sec:quantitative} we formalize the notion of \daob{} rewriting.
In Sec.~\ref{sec:nets} we introduce a production calculus
as a means to compute the production of the data-abstracted
stream specifications,
based on the set of periodically increasing functions.
A translation of stream specifications into production terms
is defined in Sec.~\ref{sec:translation}.
Our main results, mentioned above, are collected in Sec.~\ref{sec:results}.
We conclude and discuss some future research topics in Sec.~\ref{sec:conclusion}.

\section{Stream Specifications}\label{sec:class}

We introduce the notion of stream specification.
An example is given in Fig.~\ref{fig:pascal},
\begin{figure}[htb]
  \begin{center}
  \scalebox{1}{\includegraphics{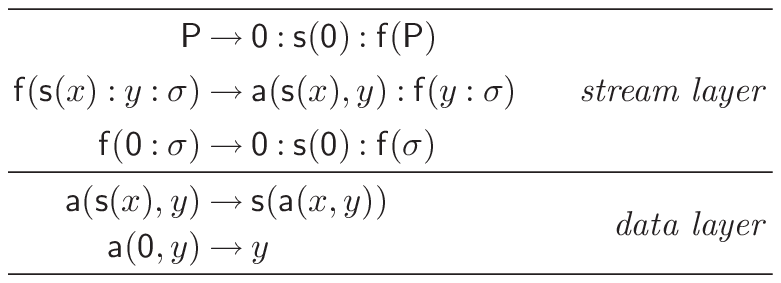}}
  \end{center}
  \vspace{-2ex}
  \caption{Example of a flat stream specification.}
  \label{fig:pascal}
\end{figure}
a productive specification of
Pascal's triangle where the rows are separated by zeros.
Indeed, evaluating this specification, we get:
\[%
  \strcf{P}
  \infred
  \strcns{\trmrep{0}}{\strcns{\trmrep{1}}{\strcns{\trmrep{0}}{
  \strcns{\trmrep{1}}{\strcns{\trmrep{1}}{\strcns{\trmrep{0}}{
  \strcns{\trmrep{1}}{\strcns{\trmrep{2}}{\strcns{\trmrep{1}}{
  \strcns{\trmrep{0}}{\strcns{\trmrep{1}}{\strcns{\trmrep{3}}{
  \strcns{\trmrep{3}}{\strcns{\trmrep{1}}{\ldots}}}}}}}}}}}}}}
  \punc,
\]
where, for $n\in\nat$, $\trmrep{n}$ is the \emph{numeral} for $n$,
defined by $\trmrep{n} \defdby \numsucn{n}{\numzer}$.

Stream specifications consist of
a \emph{stream layer} (top) where stream constants and functions are specified
and a \emph{data layer} (bottom).

The hierarchical setup of stream specifications is motivated as follows.
In order to abstract from the termination problem when investigating productivity,
the data layer is a term rewriting system on its own,
and is required to be strongly normalizing.
Moreover, an isolated data layer prevents the use of stream symbols
by data rules.
The stream layer may use symbols of the data layer, but not vice-versa.
Stream dependent data functions, like
$\strhd{\strcns{x}{\astr}} \to x$,
might cause the output of undefined data terms.
Consider the following examples\label{page:example:sijtsma} from~\cite{sijt:1989}:
\begin{align*}
   \strcf{S} &\to \strcns{\datf{0}}{\strcns{\strnth{\strcf{S}}{6}}{\strcf{S}}}
   &
   \strcf{T} &\to \strcns{\datf{0}}{\strcns{\strnth{\strcf{T}}{7}}{\strcf{T}}}
   \punc,
\end{align*}
where $\strnth{\strcf{\astr}}{n} \defdby \strhd{\strtln{n}{\strcf{\astr}}}$;
here $\strcf{S}$ is well-defined, whereas $\strcf{T}$ is not.
Note that this is not a severe restriction,
since stream dependent data functions usually
can be replaced by pattern matching:
$\bfunap{\astrfun}{\astr}{\bstr} \red \strcns{(\strhd{\astr} + \strhd{\bstr})}{\bfunap{\astrfun}{\strtl{\astr}}{\strtl{\bstr}}}$,
for example,
can be replaced by the better readable
$\bfunap{\astrfun}{\strcns{x}{\astr}}{\strcns{y}{\bstr}} \red \strcns{(x + y)}{\bfunap{\astrfun}{\astr}{\bstr}}$.

Stream specifications are formalized as many-sorted,
orthogonal, constructor term rewriting systems~\cite{terese:2003}.
We distinguish between \emph{stream terms} and \emph{data terms}.
For the sake of simplicity we consider only one sort $\sortS$ for stream terms
and one sort $\sortD$ for data terms.
Without any complication, our results
extend to stream specifications with multiple sorts for data terms and for stream terms.

Let $U$ be a finite set of \emph{sorts}.
A \emph{\mbox{$U\!$-sorted} set~$A$} is a family of sets $\{A_u\}_{u\in U}\,$;
for $V \subseteq U$ we define $A_V \defdby \bigcup_{v \in V} A_v$.
A \emph{\mbox{$U\!$-sorted} signature~$\asig$} is a \mbox{$U\!$-sorted} set of function symbols $f$,
each equipped with an arity
$\arity{f} = \pair{u_1\cdots u_n}{u} \in {U^{\ast}\times U}$
where $u$ is the sort of $f$;
we write $u_1\times\ldots\times u_n \to u$ for $\pair{u_1\cdots u_n}{u}$.
Let $X$ be a \mbox{$U\!$-sorted} set of variables.
The \emph{\mbox{$U\!$-sorted} set of terms $\ter{\asig,X}$}
is inductively defined by:
for all $u \in U$, $X_u \subseteq \ter{\asig,X}_u$, and
$\funap{f}{t_1,\ldots,t_n} \in \ter{\asig,X}_u$
if $f \in \asig$, $\arity{f} = u_1\times\ldots\times u_n \to u$, and $t_i \in \ter{\asig,X}_{u_i}$.
$\iter{\asig,X}$ denotes the set of
\emph{(possibly) infinite terms over $\asig$ and $X$} (see \cite{terese:2003}).
Usually we keep the set of variables implicit and write
$\ter{\asig}$ and $\iter{\asig}$.
A \emph{\mbox{$U\!$-sorted} term rewriting system (TRS)} is a pair
$\pair{\asig}{R}$ consisting of a \mbox{$U$-sorted} signature $\asig$
and a \mbox{$U\!$-sorted} set $R$ of \emph{rules}
that satisfy well-sortedness, for all $u\in U$:
$R_u \subseteq \ter{\asig,X}_u \times \ter{\asig,X}_u$,
as well as the standard TRS requirements.

Let $\atrs = \pair{\asig}{R}$ be a $U\!$-sorted TRS.
For a term $t\in\ter{\asig}_u$ where $u\in U$
we denote the root symbol of $t$ by $\rootsymb{t}$.
We say that two occurrences of symbols in a term are \emph{nested} if
the position \cite[p.29]{terese:2003} of one is a prefix of the position of the other.
We define $\defsymb{\asig} \defdby \{\rootsymb{l} \where l \to r \in R\}$,
the set of \emph{defined symbols},
and $\cnssymb{\asig} \defdby \asig \setminus \defsymb{\asig}$,
the set of \emph{constructor symbols}.
Then $\atrs$ is called a \emph{constructor TRS}
if for every rewrite rule $\rho \in R$,
the left-hand side is of the form $f(t_1,\ldots,t_n)$
with $t_i\in\ter{\cnssymb{\asig}}$;
then $\rho$ is called a \emph{defining rule for $f$}.
We call $\atrs$ \emph{exhaustive} for $f \in \asig$ if every term
$\funap{f}{t_1,\ldots,t_{n}}$
with closed constructor terms $t_i$ is a redex.

A \emph{stream TRS} is a
finite $\{\sortS,\sortD\}$-sorted, orthogonal, constructor TRS $\pair{\asig}{R}$
such that $\text{`$\sstrcns$'} \in \Ss$, the \emph{stream constructor symbol},
with arity $\sortD\times\sortS \to \sortS$
is the single constructor symbol in $\Ss$.
Elements of $\Sd$ and $\Ss$ are called the
\emph{data symbols} and the \emph{stream symbols}, respectively.
We let $\Ssmincns \defdby \Ss \setminus \{\text{`$\sstrcns$'}\}$,
and, for all $\astrfun \in \Ssmincns$,
we assume w.l.o.g.\ %
that the stream arguments are in front:
$\arity{\astrfun} \in \sortS^{\arityS{\astrfun}}\times\sortD^{\arityD{\astrfun}} \to \sortS$,
where $\arityS{\astrfun}$ and $\arityD{\astrfun} \in \nat$
are called the \emph{stream arity} and the \emph{data arity} of $\astrfun$, respectively.
By $\Ssc$ we denote the set of symbols in $\Ssmincns$ with stream arity $0$,
called the \emph{stream constant symbols},
and $\Ssf \defdby \Ssmincns\setminus\Ssc$
the set of symbols in $\Ssmincns$ with stream arity unequal to $0$,
called the \emph{stream function symbols}.
By $\Rsc$ we mean the defining rules for the symbols in $\Ssc$.

We repeat that the restriction to a single data sort $\sortD$ is solely for keeping the presentation simple;
all of the definitions and results in the sequel generalize to multiple data and stream sorts.
For stream TRSs we assume non-emptyness of data sorts, that is,
for every data sort there exists a finite, closed, contructor term of this sort.
In case there is only one data sort $\sortD$, 
then the requirement boils down to the existence of a nullary constructor symbol of this sort.

We come to the definition of stream specifications.
\begin{definition}\normalfont\label{def:scs}
  A \emph{stream specification} $\atrs$ is a stream TRS $\atrs = \pair{\asig}{R}$
  such that the following conditions hold:
  \begin{enumerate}
    \item There is a designated symbol $\rootsc\in\Ssc$
      with $\arityD{\rootsc} = 0$,
      the \emph{root of\/ $\atrs$}.
    \item\label{def:scs:data}
      $\pair{\Sd}{\Rd}$ is a terminating, $\sortD$-sorted TRS;
      $\Rd$ is called the \emph{data-layer} of\/ $\atrs$.
    \item\label{def:scs:exhaustive}
      $\atrs$ is exhaustive (for all defined symbols in $\asig = \Ss \uplus \Sd$).
  \end{enumerate}
\end{definition}
In the sequel we restrict to stream specifications
in which all stream constants in $\Ssc$
are reachable from the root:
$\astrcon\in\Ssc$ is \emph{reachable} if there is a term $\astrtrm$
such that $\rootsc\mred \astrtrm$ and $\astrcon$ occurs in~$\astrtrm$.
Note that reachability of stream constants is decidable,
and that unreachable symbols may be neglected for investigating
(non-)productivity.

Note that Def.~\ref{def:scs} indeed imposes a hierarchical setup;
in particular, stream dependent data functions are excluded
by item~\ref{def:scs:data}.
Exhaustivity for $\Sd$ together with strong normalization of $\Rd$
guarantees that closed data terms rewrite to constructor normal forms,
a property known as sufficient completeness~\cite{kapu:nare:rose:zhan:1991}.

We are interested in productivity of recursive stream specifications
that make use of a library of `manageable' stream functions.
By this we mean a class of stream functions
defined by a syntactic format
with the property that their \daob{} lower bounds are computable
and contained in a set of production moduli
that is effectively closed under composition,
pointwise infimum and where least fixed points can be computed.
As such a format we define the class of flat stream specifications
(Def.~\ref{def:flat}) for which \daob{} lower bounds are precisely
the set of `periodically increasing' functions (see Sec.~\ref{sec:nets}).
Thus only the stream function rules are subject to syntactic restrictions.
No condition other than well-sortedness is imposed on
the defining rules of stream constant symbols.

In the sequel let $\atrs = \pair{\asig}{R}$ be a stream specification.
We define the relation $\sdependson$ on rules in $\Rs$:
for all $\rho_1,\rho_2\in\Rs$,
$ \rho_1  \dependson  \rho_2 $
($\rho_1$ \emph{depends on} $\rho_2$)
holds if and only if $\rho_2$ is the defining rule of a stream function
symbol on the right-hand side of $\rho_1$.
Furthermore,
for a binary relation $\sred\subseteq A\times A$ on a set $A$
we define $\relsucc{\red}{a} \defdby \{ b\in A \where a\red b\}$ for all $a\in A$,
and we denote by $\stcred$ and $\srtcred$ the \emph{transitive closure} and
the \emph{reflexive and transitive closure} of $\red$, respectively.

\begin{definition}\normalfont\label{def:flat}
  A rule $\rho \in \Rs$ is called \emph{nesting}
  if its right-hand side contains nested occurrences of stream symbols from $\Ssmincns$.
  We use $\Rnest$ to denote the subset of nesting rules of $R$
  and define $\Rnnest \defdby \Rs \setminus \Rnest$, the set of \emph{non-nesting rules}.

  A rule $\rho \in \Rs$ is called \emph{flat} if all rules
  in $\relsucc{\mdependson}{\rho}$ are non-nesting.
  A symbol $\astrfun \in \Ssmincns$ is called \emph{flat} if all defining rules of $\astrfun$ are flat;
  the set of flat symbols is denoted $\Sflat$.
  A stream specification $\atrs$ is called \emph{flat} if $\Ssmincns \subseteq \Sflat \cup \Ssc$, that is,
  all symbols in $\Ssmincns$ are either flat or stream constant symbols.
\end{definition}
The specification given in Fig.~\ref{fig:pascal} is an example of a flat specification.
A second example %
is given in Fig.~\ref{fig:ternary_morse_flat}.
\begin{figure}[htb]
  \begin{center}
  \scalebox{1}{\includegraphics{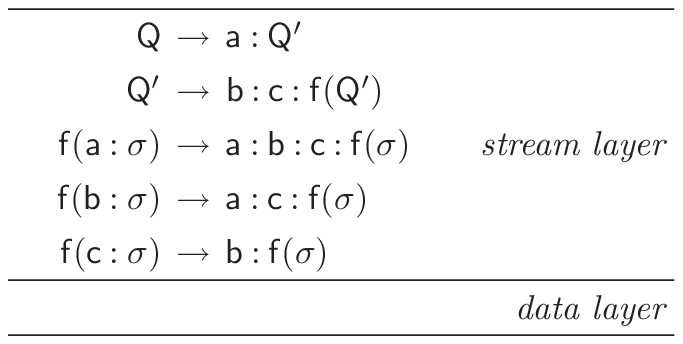}}
  \end{center}
  \vspace{-2ex}
  \caption{Example of a flat stream specification.}
  \label{fig:ternary_morse_flat}
\end{figure}
This is a productive specification that defines
the ternary Thue--Morse sequence, a square-free word over $\{\datf{a},\datf{b},\datf{c}\}$;
see, e.g.,~\cite{salo:1981}.
Indeed, evaluating this specification, we get:
\(
  \strcf{Q}
  \infred
  \strcns{\datf{a}}{\strcns{\datf{b}}{\strcns{\datf{c}}{
  \strcns{\datf{a}}{\strcns{\datf{c}}{\strcns{\datf{b}}{
  \strcns{\datf{a}}{\strcns{\datf{b}}{\strcns{\datf{c}}{
  \strcns{\datf{b}}{\strcns{\datf{a}}{\strcns{\datf{c}}{
  \ldots}}}}}}}}}}}}
\).

As the basis of \daob{} rewriting (see Def.~\ref{def:d-o-rewriting})
we define the data abstraction of terms as the results of replacing
all data-subterms by the symbol~$\trspeb$.
\begin{definition}\normalfont
  Let $\databstr{\asig} \defdby \{ \trspeb \} \uplus \Sigma_{\sortS}$.
  For stream terms $s \in \ter{\asig}_{\sortS}$,
  the \emph{data abstraction~$\databstr{s} \in \ter{\databstr{\asig}}_{\sortS}$}
  is defined by:
  \begin{align*}
    \databstr{\astr} & = \astr \punc, \\
    \databstr{\strcns{\adattrm}{\astrtrm}} &= \strcns{\trspeb}{\databstr{\astrtrm}} \punc, \\
    \databstr{\funap{\strff{f}}{\astrtrm_1,\ldots,\astrtrm_n,\adattrm_1,\ldots,\adattrm_m}}
      &= \funap{\strff{f}}{\databstr{\astrtrm_1},\ldots,\databstr{\astrtrm_n},\trspeb,\ldots,\trspeb}
    \punc.
  \end{align*}
\end{definition}
Based on this definition of data abstracted terms,
we define the class of pure stream specifications, an extension of
the equally named class in~\cite{endr:grab:hend:isih:klop:2007}.
\begin{definition}\normalfont\label{def:pure}
  A stream specification $\atrs$ is called \emph{pure} if it is flat and
  if for every %
  symbol $\astrfun \in \Ssmincns$
  the data abstractions $\databstr{\ell} \to \databstr{r}$
  of the defining rules ${\ell \to r}$ of $\astrfun$ coincide (modulo renaming of variables).
\end{definition}
\begin{figure}[htb]
  \begin{center}
  \scalebox{1}{\includegraphics{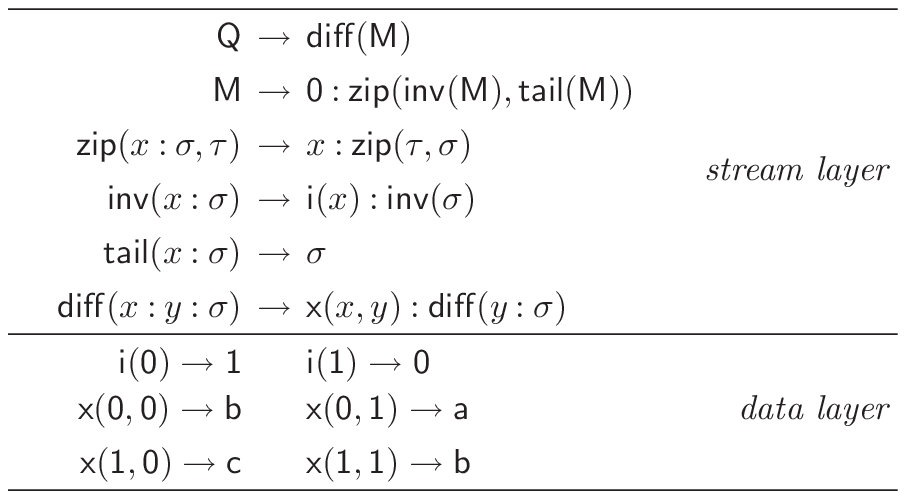}}
  \end{center}
  \vspace{-2ex}
  \caption{Example of a pure stream specification.}
  \label{fig:ternary_morse_pure}
\end{figure}
Fig.~\ref{fig:ternary_morse_pure} shows an alternative specification of
the ternary Thue--Morse sequence, this time constructed from the binary
Thue--Morse sequence specified by $\strcf{M}$.
This specification belongs to the subclass of pure specifications,
which is easily inferred by the shape of the stream layer:
for each symbol in $\Ssf = \{\sstrzip, \sstrinv, \sstrtl, \sstrdiff\}$
there is precisely one defining rule.

Another example of a pure stream specification is given in Fig.~\ref{fig:morse_dol}.
\begin{figure}[htb]
  \begin{center}
  \scalebox{1}{\includegraphics{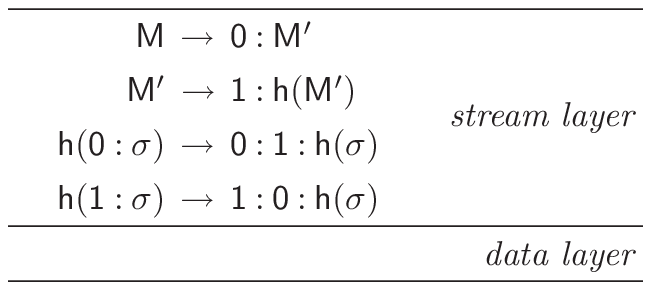}}
  \end{center}
  \vspace{-2ex}
  \caption{Example of a pure stream specification.}
  \label{fig:morse_dol}
\end{figure}
Both defining rules for $\strff{h}$ consume one, and produce two stream elements,
that is, their data abstractions coincide.

Def.~\ref{def:pure} generalizes the specifications called `pure'
in~\cite{endr:grab:hend:isih:klop:2007} in four ways concerning
the defining rules of stream functions:
First, the requirement of right-linearity of stream variables is dropped,
allowing for rules like $\strfunap{f}{\astr} \to \strfunap{g}{\astr,\astr}$.
Second, `additional supply' to the stream arguments is allowed.
For instance, in the defining rule for the function $\sstrdiff$ in Fig.~\ref{fig:ternary_morse_pure},
the variable $y$ is `supplied' to the recursive call of $\sstrdiff$.
Third, the use of non-productive stream functions is allowed now,
relaxing an earlier requirement on stream function symbols to be `weakly guarded'.
Finally, defining rules for stream function symbols may use a
restricted form of pattern matching as long as, for every stream function
$\strff{f}$, the \daob{} consumption/production behaviour (see Sec.~\ref{sec:quantitative})
of all defining rules for $\strff{f}$ is the same.

Next we define friendly nesting stream specifications,
an extension of the class of flat stream specifications.
\begin{definition}\normalfont\label{def:sfs:friendly}
  A rule $\rho \in \Rs$ is called \emph{friendly} if for all rules 
  $\gamma \in \relsucc{\mdependson}{\rho}$ we have:
  (1)~$\gamma$ consumes in each argument at most one stream element, and
  (2)~it produces at least one.
  The set of \emph{friendly nesting rules $\Rfnest$}
  is the largest extension of the set of friendly rules
  by non-nesting rules from $\Rs$ that is closed under $\sdependson$.
  A symbol $\astrfun \in \Ssmincns$ is \emph{friendly nesting}
  if all defining rules of $\astrfun$ are friendly nesting.
  A stream specification $\atrs$ is called \emph{friendly nesting} 
  if $\Ssmincns \subseteq \Sfnest \cup \Ssc$, that is,
  all symbols in $\Ssmincns$ are either friendly nesting or stream constant symbols.
\end{definition}
An example of a friendly nesting stream specification
is given in Fig.~\ref{fig:convolution}.
\begin{figure}[htb]
  \begin{center}
  \scalebox{1}{\includegraphics{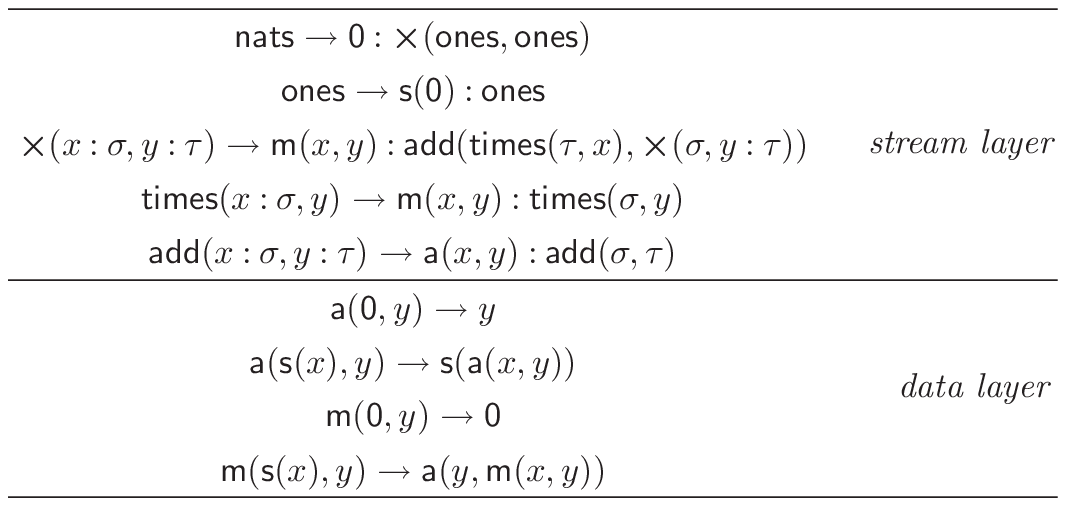}}
  \end{center}
  \vspace{-2ex}
  \caption{Example of a friendly nesting stream specification.}
  \label{fig:convolution}
\end{figure}
The root of this specification, the constant $\nats$, evaluates to the stream
$\strcns{\trmrep{0}}{\strcns{\trmrep{1}}{\strcns{\trmrep{2}}{\ldots}}}$.
The stream layer specifies a parameterized stream function $\sstrtimes$,
which multiplies every element of a stream with the parameter,
and the binary stream functions $\sstradd$ for componentwise addition,
and $\sconvprod$, the \emph{convolution product} of streams (see~\cite{rutt:2003}),
mathematically defined as an operation $\pair{\astr}{\bstr} \mapsto \astr\times\bstr$\punc:
\begin{align*}
  \strnth{({\astr}\times{\bstr})}{i}
  & = \sum_{j=0}^i \strnth{\astr}{j}\cdot\strnth{\bstr}{i-j}
  & \text{(for all $i\in\nat$).}
\end{align*}
The rewrite rule for $\sconvprod$ is nesting,
because its right-hand has nested stream function symbols,
both $\sstrtimes$ and $\sconvprod$ are nested within $\sstradd$.
Because of the presence of a nesting rule in the stream layer,
which is not a defining rule of a stream constant symbol,
this stream specification is not flat.
However, the defining rules for $\sstradd$, $\sstrtimes$, and $\sconvprod$ are friendly.
For instance for the rule for $\sconvprod$ we check:
in both arguments it consumes (at most) one stream element ($x$ and $y$),
it produces (at least) one stream element ($\nummult{x}{y}$),
and the defining rules of the stream function symbols in the right-hand side
($\sstradd$, $\sstrtimes$, and $\sconvprod$) are again friendly.
Thus the function stream layer consists of
one friendly nesting rule, two flat (and friendly) rules,
and one defining rule for a stream constant.
Therefore this stream specification is friendly nesting.

\begin{definition}\normalfont
  Let $\aars = \pair{\ter{\asig}_S}{{\to}}$
  be an abstract reduction system~(ARS) on the set
  of terms over a stream TRS signature $\asig$.
  The \emph{production function $\sterprd{\aars} \funin \ter{\asig}_S \to \conat$}
  of $\aars$ is defined for all
  $\astrtrm \in \ter{\asig}_S $~by:
  \[
    \terprd{\aars}{\astrtrm}
    \defdby
    \sup\, \{\, n\in\nat \where \astrtrm \irtcred{\aars} \strcns{\adattrm_1}{\strcns{\ldots}{\strcns{\adattrm_n}{\bstrtrm}}} \,\}
    \punc.
  \]
  We call \emph{$\aars$ productive for a stream term $\astrtrm$}
  if $\terprd{\aars}{\astrtrm} = \infty$.
  A stream specification $\atrs$ is called \emph{productive}
  if $\atrs$ is productive for its root $\rootsc$.
\end{definition}

The following proposition justifies this definition of productivity
by stating an easy consequence:
the root of a productive stream specification can be evaluated
continually in such a way that a uniquely determined stream in constructor
normal form is obtained as the limit.
This follows easily from the fact that a stream specification is an
orthogonal TRS, and hence has a confluent rewrite relation
that enjoys the property $\UNinf$ (uniqueness of infinite normal form)
\cite{klop:vrij:2005}.

\begin{proposition}
  A stream specification is productive if and only if\/
  its root has an infinite constructor term of the form
  $ \strcns{\adattrm_1}{\strcns{\adattrm_2}{\strcns{\adattrm_3}{\ldots}}} $
  as its unique infinite normal form.
\end{proposition}

\section{Data-Oblivious Analysis}\label{sec:quantitative}
We formalize the notion of \daob{} rewriting
and introduce the concept of \daob{} productivity.
The idea is a quantitative reasoning
where all knowledge about the concrete values of data elements
during an evaluation sequence is ignored.
For example, consider the following stream specification:
\begin{align}
  \strcf{M} &\red \funap{\strff{f}}{\strcns{\datf{0}}{\strcns{\datf{1}}{\strcf{M}}}}
  \notag
  &
  (1)\;\,
  \strfunap{f}{ \strcns{\datf{0}}{\strcns{x}{\astr}} }
  &\red \strcns{\datf{0}}{\strcns{\datf{1}}{\strfunap{f}{\astr}}}
  &
  (2)\;\,
  \strfunap{f}{\strcns{\datf{1}}{\strcns{x}{\astr}}}
  &\red \strcns{x}{\strfunap{f}{\astr}}
\end{align}
\jlabel{ex:do1}{1}%
\jlabel{ex:do2}{2}%
The specification of $\strcf{M}$ is productive:
$\strcf{M}
 \red^2 \strcns{\datf{0}}{\strcns{\datf{1}}{\funap{\strff{f}}{\strcf{M}}}}
 \red^3 \strcns{\datf{0}}{\strcns{\datf{1}}{\strcns{\datf{0}}{\strcns{\datf{1}}{\funap{\strff{f}}{\funap{\strff{f}}{\strcf{M}}}}}}}
 \mred \ldots\punc.$
During the rewrite sequence \eqref{ex:do2} is never applied.
Disregarding the identity of data, however,
\eqref{ex:do2} becomes applicable and allows for the rewrite sequence:
$$\strcf{M}
 \red \funap{\strff{f}}{\strcns{\trspeb}{\strcns{\trspeb}{\strcf{M}}}}
 \red^{\eqref{ex:do2}} \strcns{\trspeb}{\funap{\strff{f}}{\strcf{M}}}
 \mred \strcns{\trspeb}{\funap{\strff{f}}{\strcf{\strcns{\trspeb}{\funap{\strff{f}}{\strcf{\strcns{\trspeb}{\funap{\strff{f}}{\strcf{\ldots}}}}}}}}}
  \punc,$$
producing only one element.
Hence $\strcf{M}$ is not \daobly{} productive.

\begin{figure}[b!]
  \begin{center}
  \scalebox{1}{\includegraphics{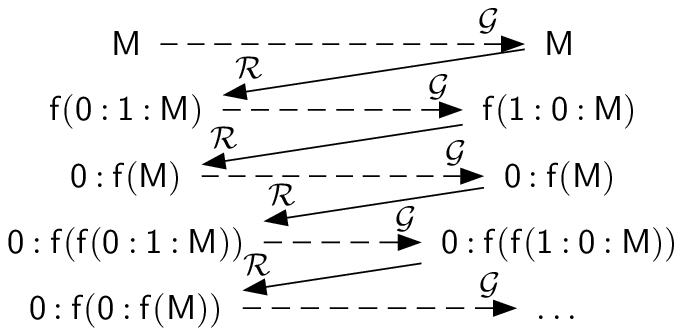}}
  \end{center}
  \vspace{-2ex}
  \caption{Data-oblivious rewriting}
  \label{fig:dorewriting}
\end{figure}
\Daob{} term rewriting can be thought of as a two-player game
between a \emph{rewrite player} $\mathcal{R}$ which performs the usual term rewriting,
and an \emph{opponent} $\sdg$
which before every rewrite step
is allowed to arbitrarily exchange data elements
for (sort-respecting) data terms in constructor normal form.
The opponent can either handicap or support the rewrite player.
Respectively,
the \daob{} lower (upper) bound on the production of a stream term $\astrtrm$
is the infimum (supremum) of the production of $\astrtrm$ with respect to
all possible strategies for the opponent $\sdg$.
Fig.~\ref{fig:dorewriting}
depicts \daob{} rewriting of the above stream specification $\strcf{M}$;
by exchanging data elements, the opponent $\sdg$ enforces
the application of \eqref{ex:do2}.

\begin{remark}
  At first glance it might appear natural to model the opponent
  using a function $\sdg \funin \ter{\Sd} \to \ter{\Sd}$ from data terms to data terms.
  However, such a per-element deterministic exchange strategy
  preserves equality of data elements.
  Then the following specification of $\strcf{W}$\,:
  \begin{align*}
  \strcf{W}&\to \funap{\strff{g}}{\strcf{Z},\strcf{Z}} %
  &
  \strcf{Z}&\to\strcns{\datf{0}}{\strcf{Z}} %
  \\
  \funap{\strff{g}}{\strcns{\datf{0}}{\astr},\strcns{\datf{0}}{\bstr}}
  &\to\strcns{\datf{0}}{\funap{\strff{g}}{\astr,\bstr}} %
  &
  \funap{\strff{g}}{\strcns{\datf{0}}{\astr},\strcns{\datf{1}}{\bstr}}
  &\to\funap{\strff{g}}{\astr,\bstr} %
  \\
  \funap{\strff{g}}{\strcns{\datf{1}}{\astr},\strcns{\datf{1}}{\bstr}}
  &\to\strcns{\datf{0}}{\funap{\strff{g}}{\astr,\bstr}} %
  &
  \funap{\strff{g}}{\strcns{\datf{1}}{\astr},\strcns{\datf{0}}{\bstr}}
  &\to\funap{\strff{g}}{\astr,\bstr} %
  \end{align*}
  would be productive for every such $\sdg$,
  which is clearly not what one would expect of a \daob{} analysis.
\end{remark}

The opponent can be modelled by an operation on stream terms: %
$\ter{\asig}_\sortS \to \ter{\asig}_\sortS$.
However, it can be shown that it is sufficient to quantify over strategies for $\sdg$
for which $\dg{\astrtrm}$ is invariant under exchange of data elements in $\astrtrm$ for all terms $\astrtrm$.
Therefore we first abstract from the data elements in favour of symbols $\trspeb$
and then define the opponent $\sdg$ on the abstracted terms, $\sdg \funin \ter{\databstr{\asig}}_{\sortS} \to \ter{\asig}_{\sortS}$.

\begin{definition}\normalfont
  Let $\atrs = \pair{\asig}{R}$ be a stream specification.
  A \emph{data-exchange function on $\atrs$}
  is a function $\sdg \funin \ter{\databstr{\asig}}_{\sortS} \to \ter{\asig}_{\sortS}$
  such that $\forall \cstrtrm \in \ter{\databstr{\asig}}_{\sortS}$ it holds:
  $\databstr{\funap{\sdg}{\cstrtrm}} = \cstrtrm$,
  and $\funap{\sdg}{\cstrtrm}$ is in data-constructor normal form.
\end{definition}

\begin{definition}\normalfont\label{def:d-o-rewriting}
  We define the ARS
  $\dgars{\atrs}{\sdg}
   \subseteq \ter{\databstr{\asig}}_{\sortS} \times \ter{\databstr{\asig}}_{\sortS}$
  for every data-exchange function $\sdg$,
  as follows:
  \begin{align*}
    \dgars{\atrs}{\sdg}
    \defdby \{ s \to \databstr{t} \mid \: s \in \ter{\databstr{\asig}},\;
                  t \in \ter{\asig} \text{ with } \dg{s} \to_\atrs t\} \; .
  \end{align*}
  The \emph{\daob{} production range $\doRng{\atrs}{\astrtrm}$}
  of a data abstracted stream term $\astrtrm \in \ter{\databstr{\asig}}_{\sortS}$ is defined as follows:
  \begin{align*}
  \doRng{\atrs}{\astrtrm} \defdby  \{ \terprd{\dgars{\atrs}{\sdg}}{\astrtrm} \where {}
     &\sdg \text{ a data-exchange function on } \atrs \} \punc .
  \end{align*}
  The \emph{\daob{} lower} and \emph{upper bound on the production of $s$}
  are defined by
  \begin{align*}
  \doLow{\atrs}{\astrtrm} &\defdby \inf(\doRng{\atrs}{\astrtrm})
  \punc,
  &\text{ and}
  &&
  \doUp{\atrs}{\astrtrm} &\defdby \sup(\doRng{\atrs}{\astrtrm})
  \punc,
  \end{align*}
  respectively.
  These notions carry over to stream terms $s \in \ter{\asig}_{\sortS}$
  by taking their data abstraction $\databstr{s}$.

  A stream specification $\atrs$ is
  \emph{\daobly{} productive} (\emph{\daobly{} non-productive})
  if $\doLow{\atrs}{\rootsc} = \infty$  (if $\doUp{\atrs}{\rootsc} < \infty$)
  holds.
\end{definition}
\begin{proposition}\label{prop:doprod}
  For $\atrs = \pair{\asig}{R}$
  a stream specification and $\astrtrm \in \ter{\asig}_{\sortS}$:
  \[
    \doLow{\atrs}{\astrtrm} \le \terprd{\atrs}{\astrtrm} \le \doUp{\atrs}{\astrtrm}\punc.
  \]
  Hence \daob{} productivity implies productivity
  and \daob{} non-productivity implies non-productivity.
\end{proposition}

We define lower and upper bounds
on the \daob{} consumption/production behaviour of stream functions.
These bounds are used to reason about \daob{}
(non-) productivity of stream constants, see Sec.~\ref{sec:results}.
\begin{definition}\normalfont\label{def:dorange}
  Let $\atrs = \pair{\asig}{R}$ be a stream specification,
  $\strff{g} \in \Ssmincns$, $k = \arityS{\strff{g}}$, and $\ell = \arityD{\strff{g}}$.
  The \emph{\daob{} production range
  $\doRng{\atrs}{\strff{g}} \funin \nat^{k} \to \powerset{\conat}$
  of $\strff{g}$} is:
  \[
    \funap{\doRng{\atrs}{\strff{g}}}{n_1,\ldots,n_{k}}
    \defdby
    \doRng{\atrs}{\,\funap{\strff{g}}{
      (\strcns{\trspeb^{n_1}}{\astr}),\ldots,(\strcns{\trspeb^{n_{k}}}{\astr}),
      \underbrace{\trspeb,\ldots,\trspeb}_{\text{$\ell$-times}}}\,}
    \punc,
    \]
    \\[-3ex]
  where $\strcns{\trspeb^{m}}{\astr} \defdby \overbrace{\trspeb\xstrcns\ldots\xstrcns\trspeb\,\sstrcns}^{\text{$m$ times}}\,\astr$.
  The \emph{\daob{} lower and upper bounds on the production of $\strff{g}$}
  are defined by $\doLow{\atrs}{\strff{g}} \defdby \inf(\doRng{\atrs}{\strff{g}})$ and
  $\doUp{\atrs}{\strff{g}} \defdby \sup(\doRng{\atrs}{\strff{g}})$, respectively.
\end{definition}

Even simple stream function specifications
can exhibit a complex \daob{} behaviour.
For example, consider the flat function specification:
\begin{align*}
  \label{example:traces}
\funap{\strff{f}}{\astr}
&\red \bfunap{\strff{g}}{\astr}{\astr}
\\
\bfunap{\strff{g}}{\strcns{\datf{0}}{\strcns{y}{\astr}}}{\strcns{x}{\bstr}}
&\red \strcns{\datf{0}}{\strcns{\datf{0}}{\bfunap{\strff{g}}{\astr}{\bstr}}}
\\
\bfunap{\strff{g}}{\strcns{\datf{1}}{\astr}}{\strcns{x_1}{\strcns{x_2}{\strcns{x_3}{\strcns{x_4}{\bstr}}}}}
&\red \strcns{\datf{0}}{\strcns{\datf{0}}{\strcns{\datf{0}}{\strcns{\datf{0}}{\strcns{\datf{0}}{\bfunap{\strff{g}}{\astr}{\bstr}}}}}}
\end{align*}
Fig.~\ref{fig:dobounds} (left)
\begin{figure}[htb]
  \begin{center}
  \scalebox{1}{\includegraphics{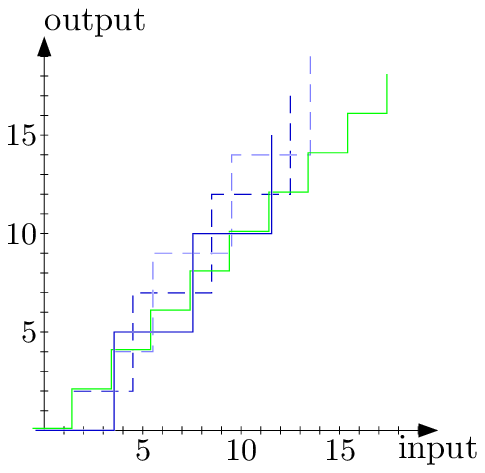}}
  \hspace{5em}
  \scalebox{1}{\includegraphics{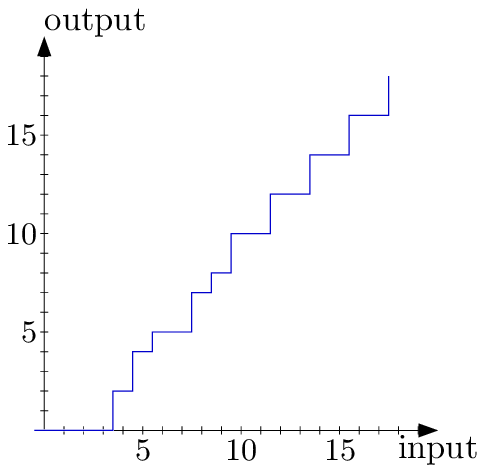}}
  \end{center}
\vspace{-2ex}
	\caption{Traces}
\label{fig:dobounds}
\end{figure}
shows a (small) selection of the possible function-call traces for $\strff{f}$.
In particular, it depicts the traces that contribute
to the \daob{} lower bound $\doLow{\atrs}{\strff{f}}$.
The lower bound $\doLow{\atrs}{\strff{f}}$, shown on the right,
is a superposition of multiple traces of $\strff{f}$.
In general $\doLow{\atrs}{\strff{f}}$ can even be a superposition of infinitely many traces.

\subsection*{First Observations on Data-Oblivious Rewriting}

For a stream function symbol $\bstrfun$,
we define its optimal production modulus $\sprdmod{\strff{g}}$,
the data-aware, quantitative lower bound on the production of $\strff{g}$,
and compare it to its \daob{} lower bound $\doLow{\atrs}{\strff{g}}$.
\begin{definition}\normalfont\label{def:dorange:1}
  Let $\atrs = \pair{\asig}{R}$ be a stream specification.
  We define the set $S_n$ of \emph{$n$-defined} stream terms,
  i.e.\ finite stream terms with a constructor-prefix of length $n\in\nat$:
  \begin{gather*}
    S_n \defdby
      \{\strcns{\adattrm_1}{\strcns{\ldots}{\strcns{\adattrm_n}{\astr}}}
        \where \adattrm_1,\ldots,\adattrm_n \in \ter{\cnssymb{\Sd},\setemp}\}\punc,
  \end{gather*}
  Moreover, let $\strff{g} \in \Ssf$ with %
  $k = \arityS{\bstrfun}$ and $\ell = \arityD{\bstrfun}$,
  and define, for all $\vec{n} = n_1,\ldots,n_k$,
  the set $G_{\vec{n}}$
  of applications of $\strff{g}$ to $n_i$-defined arguments:
  \[
    G_{\vec{n}} %
    \defdby
    \{\,
      \strfunap{g}{\vec{\astrtrm},\vec{\bdattrm}}
      \where
      \myall{i}{\astrtrm_i \in S_{n_i}} ,\: \myall{j}{\bdattrm_j \in \ter{\cnssymb{\Sd},\setemp}}
    \,\}
    \punc,
  \]
  where $\vec{\astrtrm} = \astrtrm_1,\ldots,\astrtrm_k$
  and $\vec{\bdattrm} = \bdattrm_1,\ldots,\bdattrm_\ell$.
  Then, the
  \emph{optimal production modulus $\sprdmod{\strff{g}}\funin\conat^k\to\conat$ of $\strff{g}$}\/
  is defined by:
  \[
    \prdmod{\strff{g}}{\vec{n}}
    \defdby
    \inf\, \{ \terprd{\atrs}{\bstrtrm} \where \bstrtrm \in G_{\vec{n}} \}
    \punc.
  \]
\end{definition}

To illustrate the difference between the optimal production modulus $\sprdmod{\strff{h}}$
and the \daob{} lower bound $\doLow{\atrs}{\strff{h}}$,
consider the following stream function specification:
\begin{align}
  \strfunap{h}{ \strcns{\datf{0}}{\strcns{x}{\astr}} }
  &\red \strcns{x}{\strcns{x}{\strfunap{h}{\strcns{\datf{0}}{\astr}}}}
  \tag{$\rho_{\strff{h}_\datf{0}}$}\label{ex:quant1}
  \\
  \strfunap{h}{\strcns{\datf{1}}{\strcns{x}{\astr}}}
  &\red \strcns{x}{\strfunap{h}{\strcns{\datf{0}}{\astr}}}
  \tag{$\rho_{\strff{h}_\datf{1}}$}\label{ex:quant2}
\end{align}
with $\Sd = \{\datf{0},\datf{1}\}$.
Then $\sprdmod{\strff{h}}(n) = \cosubtr{2n}{3}$ is
the optimal production modulus of the stream function $\strff{h}$.
To obtain this bound one has to take into account
that the data element $\datf{0}$ is supplied to the recursive call
and conclude that \eqref{ex:quant2}
is only applicable in the first step
of a rewrite sequence
$\funap{\strff{h}}{\strcns{\adattrm_1}{\strcns{\ldots}{\strcns{\adattrm_n}{\sigma}}}} \to \ldots$.
However, the \daob{} lower bound is $\doLow{\atrs}{\strff{h}}(n) = \cosubtr{n}{1}$,
derived from rule~\eqref{ex:quant2}.

The following lemmas state an observation about the role of the opponent and the rewrite player.
Basically, the opponent $\sdg$ can select the rule which is applicable
for each position in the term; the rewrite player $\mathcal{R}$ can choose which position to rewrite.
We use subscripts for pebbles $\trspeb$, for example $\trspeb_w$,
to introduce `names' for referring to these pebbles.

\begin{definition}[Instantiation with respect to a rule $\rho$]\normalfont\label{def:inst:rho}
  For $s \in \ter{\databstr{\asig}}_{\sortS}$ with
  $s \equiv
  \funap{\astrfun}{s_1,\ldots,s_{\arityS{\astrfun}},\pebble_1,\ldots,\pebble_{\arityD{\astrfun}}}$
  and $\rho \in R$ a defining rule of $\astrfun$,
  we define a data-exchange function $\sdgrho{s}{\rho}$ as follows.
  Note that $\rho$ is of the form:
  \[\rho :
   \funap{\astrfun}{
      \strcns{\vec{\adattrm}_1}{\astr_1},
      \ldots,
      \strcns{\vec{\adattrm}_n}{\astr_n},
      \bdattrm_1,\ldots,\bdattrm_{\arityD{\strff{f}}}
    }
   \red r \in R \punc.\]
  For $i = 1,\ldots,\arityS{\astrfun}$ let $n_i \in \nat$ be maximal
  such that $s_i \equiv \strcns{\trspeb_{i,1}}{\ldots \strcns{\trspeb_{i,n_i}}{s'_i}}$.

  Let $\sdgrho{s}{\rho}$ 
  for all $1 \le i \le \arityD{\astrfun}$ instantiate the pebbles $\pebble_i$
  with closed instances  of the data patters $\bdattrm_i$,
  and for all $1 \le i \le \arityS{f}$, $1 \le j \le \min (n_i, \lstlength{\vec{\adattrm}_i})$ instantiate the pebble $\trspeb_{i,j}$
  with closed instances of the data pattern $\vec{u_i}_j$, respectively.
\end{definition}

\begin{lemma}
  Let $s \in \ter{\databstr{\asig}}_{\sortS}$, 
  $s \equiv \funap{\astrfun}{s_1,\ldots,s_{\arityS{\astrfun}},\pebble_1,\ldots,\pebble_{\arityD{\astrfun}}}$
  and $\rho \in R$ a defining rule of $\astrfun$; we use the notation from Def.~\ref{def:inst:rho}.
  Then $\dgrho{s}{\rho}{s}$ is a redex if and only if for all $1 \le i \le \arityS{f}$ we have: $\vec{\adattrm}_i \le n_i$,
  that is, there is `enough supply for the application of $\rho$'.
  Furthermore if $\dgrho{s}{\rho}{s}$ is a redex, then it is a redex with respect to $\rho$
  (and no other rule from $R$ by orthogonality).
\end{lemma}

\begin{proof}
  A consequence of the fact that $\atrs$ is an orthogonal constructor TRS. \qed
\end{proof}

\begin{lemma}\label{lem:gforce}
  Let $\sdg$ be a data-exchange function,
  $t \in \ter{\databstr{\asig}}$ a term
  and define $\mathcal{P} = \{ p \where p \in \pos{t},\; \rootsymb{t|_p} \in \Sigma\}$.
  For every $\varsigma : \mathcal{P} \to R$ such that
  $\funap{\varsigma}{p}$ is a defining rule for $\rootsymb{t|_p}$ for all $p \in \mathcal{P}$,
  there is a data-exchange function $\sdgalt{t}{\varsigma}$
  such that for all $p \in \mathcal{P}$
  if the term $\dgalt{t}{\varsigma}{t}|p$ is a redex, then it is an $\funap{\varsigma}{p}$-redex,
  and $\dgalt{t}{\varsigma}{s} \equiv \dg{s}$ for all $s \not \equiv t$.
\end{lemma}

\begin{proof}
  For all $p \in \mathcal{P}$ and $\funap{\varsigma}{p} \equiv \ell \to r$
  we alter $\sdg$ to obtain a data-guess function $\sdgalt{t}{\varsigma}$ as follows.
  If $q \in \pos{\ell}$ such that $\ell|_q$ is a data term
  and none of the $t(q\cdot p')$ with $p'$ a non empty prefix of $p$ is a symbol from $\Ssf$,
  then instantiate the data term at position $q\cdot p$ in $t$
  with an instance of the data pattern $\ell|_q$.
  Then if $\dgalt{t}{\varsigma}{t}$ is a redex, then by orthogonality of $\atrs$ it can only be a $\ell \to r$-redex.
  \qed
\end{proof}

\subsection*{History Aware Data-Exchange}

Above we have modelled the opponent using history free
data-exchange strategies and the rewrite player was omniscient,
that is, she always chooses the best possible rewrite sequence,
which produces the maximum possible number of elements.

Now we investigate the robustness of our definition.
We strengthen the opponent, allowing for history aware data-exchange strategies,
and we weaken the rewrite player, dropping omniscience,
assuming only an outermost-fair rewrite strategy.
However it turns out that
these changes do not affect the \daob{} production range $\doRng{\atrs}{s}$,
in this way
providing evidence for the robustness of our first definition.

\begin{definition}\normalfont\label{def:history}
  Let $\atrs = \pair{\asig}{R}$ be a stream specification.
  A \emph{history in $\atrs$} is a finite list in
  $ %
    (\ter{\asig}_{\sortS} \times R \times \nat^* \times \ter{\asig}_{\sortS})^*$
  of the form:
  \begin{align*}
    \quadruple{s_0}{\ell_0 \to r_0}{p_0}{t_0}\quadruple{s_1}{\ell_1 \to r_1}{p_1}{t_1}\ldots\quadruple{s_n}{\ell_n \to r_n}{p_n}{t_n}
  \end{align*}
  such that $\forall 0 \le i \le n$: $s_i \red t_i$
  by an application of rule $\ell_i \to r_i$ at position $p_i$.
  We use $\hist{\atrs}$ to denote the set of all histories of $\atrs$.

  A function $\sdg \funin \hist{\atrs} \to
                   (\ter{\databstr{\asig}}_{\sortS} \to \ter{\asig}_{\sortS})$
  from histories in $\atrs$ to data-exchange functions on $\atrs$
  is called a \emph{history-aware data-exchange function on $\atrs$}.
  For such a function we write
  $\sdgh{h}$ as shorthand for $\funap{\sdg}{h}$
  for all $h\in\hist{\atrs}$.
\end{definition}

\begin{definition}\normalfont\label{def:d-o-rewriting-history}
  Let $\atrs = \pair{\asig}{R}$ be a stream specification.
  Let $\sdg$ be a history-aware data-exchange function.
  We define the ARS
  $\dgarsh{\atrs}{\sdg} \subseteq
  (\ter{\databstr{\asig}}_{\sortS} \times \hist{\atrs}) \times (\ter{\databstr{\asig}}_{\sortS} \times \hist{\atrs})$
  as follows:
  \begin{align*}
    \dgarsh{\atrs}{\sdg} \defdby
    \{\pair{\databstr{s}}{h} &\to \pair{\databstr{t}}{h\quadruple{s}{\ell \red r}{p}{t}} \mid\\
      & s,t \in \ter{\asig}\text,\; h \in \hist{\atrs} \text{, with }
                \dgh{h}{\databstr{s}} = s
                \text{ and } s \to_\atrs t\} \punc .
  \end{align*}
  For $\bars \subseteq \dgarsh{\atrs}{\sdg}$ we define the \emph{production function $\sterprd{\bars} \funin \ter{\asig}_S \to \conat$} by:
  \[
    \terprd{\bars}{\astrtrm}
    \defdby
    \sup\, \{\, n\in\nat \where \pair{\astrtrm}{\lstemp}
             \mred_{\bars} \pair{\strcns{\adattrm_1}{\strcns{\ldots}{\strcns{\adattrm_n}{\bstrtrm}}}}{h} \,\}
    \punc,
  \]
  that is, the production of $\astrtrm$ starting with empty history $\lstemp$.
\end{definition}

In an ARS~$\dgarsh{\atrs}{\sdg}$ we allow to write
$s_0 \red_{\dgarsh{\atrs}{\sdg}} s_1 \red_{\dgarsh{\atrs}{\sdg}}
 s_2 \red_{\dgarsh{\atrs}{\sdg}} \ldots \;$
when we will actually mean a rewrite sequence in $\dgarsh{\atrs}{\sdg}$
of the form
$ \pair{s_0}{h_0} \red_{\dgarsh{\atrs}{\sdg}}
  \pair{s_1}{h_1} \red_{\dgarsh{\atrs}{\sdg}}
  \pair{s_2}{h_2} \red_{\dgarsh{\atrs}{\sdg}} \ldots \;$
with some histories $h_0,h_1,h_2,\ldots\in\hist{\atrs}$.

The notion of \emph{positions of rewrite steps and redexes}
extends to the ARS $\dgarsh{\atrs}{\sdg}$ in the obvious way,
since the steps in $\dgarsh{\atrs}{\sdg}$ are projections of steps in $\atrs$.
Note that the opponent can change the data elements in $\dgarsh{\atrs}{\sdg}$ at any time,
therefore the rules might change with respect to which a certain position is a redex.
For this reason we adapt the notion of outermost-fairness for our purposes.
Intiutively, a rewrite sequence is outermost-fair
if every position which is always eventually an outermost redex position, is rewritten infinitely often.
Formally, we define a rewrite sequence $\pair{s_1}{h_1} \to \pair{s_2}{h_2} \to \ldots$ in $\dgarsh{\atrs}{\sdg}$
to be \emph{outermost-fair}
if it is finite, or it is infinite and
whenever $p$ is an outermost redex position for infinitely many $s_{n}$, $n \in \nat$,
then there exist infinitely many rewrite steps $s_m \to s_{m+1}$, $m \in \nat$, at position $p$.
Moreover, a strategy $\bars \subseteq \dgarsh{\atrs}{\sdg}$ is \emph{outermost-fair}
if every rewrite sequence with respect to this strategy is outermost-fair.

\begin{proposition}\label{prop:omfair:strat}
  For every stream specification $\atrs = \pair{\asig}{R}$
  there exists a computable, deterministic om-fair strategy
  on $\dgarsh{\atrs}{\sdg}$.
\end{proposition}

\begin{definition}\normalfont
  The \emph{history aware \daob{} production range $\doRngh{\atrs}{\astrtrm}$}
  of $s \in \ter{\databstr{\asig}}_{\sortS}$ is defined as follows:
  \begin{align*}
  \doRngh{\atrs}{\astrtrm} \defdby  \bigl\{ \terprd{\bars}{\astrtrm} \mid{}
     &\sdg \text{ a history-aware data-exchange function on } \atrs,\\
     &\bars \subseteq \dgarsh{\atrs}{\sdg} \text{ an outermost-fair strategy} \bigr\} \punc .
  \end{align*}
  For $s \in \ter{\asig}_{\sortS}$ define $\doRngh{\atrs}{\astrtrm} \defdby \doRngh{\atrs}{\databstr{\astrtrm}}$.
\end{definition}

Note that also the strategy of the rewrite player $\bars \subseteq \dgarsh{\atrs}{\sdg}$ is history aware
since the elements of the ARS $\dgarsh{\atrs}{\sdg}$ have history annotations.

\begin{lemma}\label{lem:do:lowh:ind}
  Let $\atrs = \pair{\asig}{R}$ be a stream specification.
  Let $\astrfun\in\Ssf$ with
  $\arityS{\astrfun} = n$ and $\arity{\astrfun} = n'$,
  and let $s_1,\ldots,s_n\in\ter{\asig}_{\sortS}$,
  $u_1,\ldots,u_{n'}\in\ter{\asig}_{\sortD}$.
  Then it holds:
  \begin{align}\label{eq1:lem:do:lowh:id}
    \doLowh{\atrs}{\funap{\astrfun}{s_1,\ldots,s_n, u_1,\ldots,u_{n'}}} &=
    \doLowh{\atrs}{\funap{\astrfun}{\strcns{\trspeb^{m_1}}{\astr},\ldots,\strcns{\trspeb^{m_n}}{\astr},u_1,\ldots,u_{n'}}}
    \punc,
  \end{align}
  where $m_i \defdby \doLowh{\atrs}{s_i}$ for all $i\in\{1,\ldots,n\}$.
\end{lemma}

\begin{proof}
  Let $\atrs = \pair{\asig}{R}$ be a stream specification.
  Let $\astrfun\in\Ssf$ with $\arityS{\astrfun}=n$,
  and $s_1,\ldots,s_n\in\ter{\asig}_{\sortS}$.
  We assume that the data arity of $\astrfun$
  is zero; in the presence of data arguments the proof
  proceeds analogously. %
  Furthermore we let
  $m_i \defdby \doLowh{\atrs}{s_i}$ for $i\in\{1,\ldots,n\}$.
  We show \eqref{eq1:lem:do:lowh:id} by demonstrating the
  inequalities ``$\le$'' and ``$\ge$'' in this equation.
  \begin{description}
    \item[$\ge$:]
      Suppose that
      $ \doLowh{\atrs}{\funap{\astrfun}{s_1,\ldots,s_n}}
        = \terprd{\bars}{\databstr{\funap{\astrfun}{{s_1},\ldots,{s_n}}}} = N\in\nat $
      for some om-fair strategy $\bars\subseteq\dgarsh{\atrs}{\sdg}$
      and some history-aware data-exchange function $\sdg$ on $\atrs$.
      (In case that
       $\doLowh{\atrs}{\funap{\astrfun}{s_1,\ldots,s_n}} = \infty$,
       nothing remains to be shown.)
      Then there exists an om-fair rewrite sequence
      \begin{multline*}
        \xi \funin
        \databstr{\funap{\astrfun}{{s_1},\ldots,{s_n}}}
          \red_{\dgarsh{\atrs}{\sdg}}  u_1
          \red_{\dgarsh{\atrs}{\sdg}}  u_2
          \red_{\dgarsh{\atrs}{\sdg}}  \ldots
          \\
          \ldots
          \red_{\dgarsh{\atrs}{\sdg}}  u_l
          \red_{\dgarsh{\atrs}{\sdg}}  u_{l+1}
          \red_{\dgarsh{\atrs}{\sdg}}  \ldots
      \end{multline*}
      in $\dgarsh{\atrs}{\sdg}$
      with $\lim_{l\to\infty} \nguards{u_l} = N$.
      Due to the definition of the numbers $m_i$,
      ah history-aware data-exchange function $\sdg'$ from $\hist{\atrs}$ 
      to data-exchange functions
      can be defined that enables a rewrite sequence
      \begin{multline*}
        \xi' \funin
        \funap{\astrfun}{\strcns{\trspeb^{m_1}}{\astr},\ldots,\strcns{\trspeb^{m_n}}{\astr}}
          \red^=_{\dgarsh{\atrs}{\sdg'}}  u'_1
          \red^=_{\dgarsh{\atrs}{\sdg'}}  u'_2
          \red^=_{\dgarsh{\atrs}{\sdg'}}  \ldots
          \\
          \ldots
          \red^=_{\dgarsh{\atrs}{\sdg'}}  u'_l
          \red^=_{\dgarsh{\atrs}{\sdg'}}  u'_{l+1}
          \red^=_{\dgarsh{\atrs}{\sdg'}}  \ldots 
      \end{multline*}
      with the properties that
      $\nguards{u_l} = \nguards{u'_l}$ holds for all $l\in\nat$,
      and that $\xi'$ is the projection of $\xi$ to
      a rewrite sequence with source
      $\funap{\astrfun}{\strcns{\trspeb^{m_1}}{\astr},\ldots,\strcns{\trspeb^{m_n}}{\astr}}$.
      More precisely, $\xi'$ arises as follows:
      Steps $ u_l \red_{\dgarsh{\atrs}{\sdg}}  u_{l+1} $ in $\xi$
      that contract redexes in descendants of the outermost symbol
      $\astrfun$ in the source of $\xi$ give rise
      to steps $ u'_l \red_{\dgarsh{\atrs}{\sdg'}}  u'_{l+1} $
      that contract redexes at the same position and apply the same rule.
      This is possible because, due to the definition of the numbers
      $m_i$, always enough pebble-supply is guaranteed to carry out
      steps in $\xi'$ corresponding to such steps in $\xi$.
      Steps $ u_l \red_{\dgarsh{\atrs}{\sdg}}  u_{l+1} $
      that contract redexes in descendants of one of the subterms
      $s_1$, \ldots, $s_n$ in the source of $\xi$ project to
      empty steps $ u'_l = u'_{l+1} $
      in $\xi'$.
      (On histories $h\in\hist{\atrs}$ 
       and terms $r\in\ter{\databstr{\asig}}_{\sortS}$ occurring in $\xi$, 
       $\sdg'$ is defined in such a way as to make this projection possible.
       On histories $h$ and terms $r$ that do not occur in $\xi'$,
       $\sdg'$ is defined arbitrarily with the only restriction
       $\databstr{\funap{\sdgh{h}}{r}}$
       ensuring that $\sdg'$ behaves as a history-aware data-exchange function
       on these terms.)

      By its construction, $\xi'$ is again om-fair.
      Now let $\bars'$ be the extension
      of the sub-ARS of $\dgarsh{\atrs}{\sdg}$ induced by $\xi'$
      to a deterministic om-fair strategy for $\dgarsh{\atrs}{\sdg'}$.
      (Choose an arbitrary deterministic om-fair strategy~$\tilde{\bars}$
       on $\dgarsh{\atrs}{\sdg}$, which is possible by
       Proposition~\ref{prop:omfair:strat}.
       On term-history pairs that do not occur in $\xi'$,
       define $\bars'$ according to $\tilde{\bars}$.)
      Then $\xi'$ is also a rewrite sequence in $\bars'$
      that witnesses
      $\terprd{\bars}{\funap{\astrfun}{\strcns{\trspeb^{m_1}}{\astr},\ldots,\strcns{\trspeb^{m_n}}{\astr}}}
        = \lim_{l\to\infty} \nguards{u'_l}
        = \lim_{l\to\infty} \nguards{u_l}
        = N
        = \terprd{\bars}{\databstr{\funap{\astrfun}{{s_1},\ldots,{s_n}}}}
        $.
      Now it follows:
      \begin{multline*}
        \doLowh{\atrs}{\funap{\astrfun}{\strcns{\trspeb^{m_1}}{\astr},\ldots,\strcns{\trspeb^{m_n}}{\astr}}}
        \le
          \terprd{\bars'}{
            \funap{\astrfun}{\strcns{\trspeb^{m_1}}{\astr},\ldots,\strcns{\trspeb^{m_n}}{\astr}}
                          }
        \\
        = \terprd{\bars}{\databstr{\funap{\astrfun}{{s_1},\ldots,{s_n}}}}
        = \doLowh{\atrs}{\funap{\astrfun}{s_1,\ldots,s_n}} \; ,
      \end{multline*}
      establishing ``$\ge$'' in \eqref{eq1:lem:do:lowh:id}.
      \vspace*{1ex}
    \item[$\le$:]
      Suppose that
      $ \doLowh{\atrs}{\funap{\astrfun}{\strcns{\trspeb^{m_1}}{\astr},\ldots,\strcns{\trspeb^{m_n}}{\astr}}}
        = \terprd{\bars}{\funap{\astrfun}{\strcns{\trspeb^{m_1}}{\astr},\ldots,\strcns{\trspeb^{m_n}}{\astr}}} = N$
      for some $N \in\nat $, 
      and an om-fair strategy $\bars\subseteq\dgarsh{\atrs}{\sdg}$
      and a history-aware data-exchange function $\sdg$ on $\atrs$.
      (In case that
       $\doLowh{\atrs}{\funap{\astrfun}{\strcns{\trspeb^{m_1}}{\astr},\ldots,\strcns{\trspeb^{m_n}}{\astr}}} = \infty$,
       nothing remains to be shown.)
      Then there exists an om-fair rewrite sequence
      \begin{multline*}
        \xi : \;
        \funap{\astrfun}{\strcns{\trspeb^{m_1}}{\astr},\ldots,\strcns{\trspeb^{m_n}}{\astr}}
          \\
          \red_{\dgarsh{\atrs}{\sdg}}  u_1
          \red_{\dgarsh{\atrs}{\sdg}}  u_2
          \red_{\dgarsh{\atrs}{\sdg}}  \ldots
          \red_{\dgarsh{\atrs}{\sdg}}  u_l
          \red_{\dgarsh{\atrs}{\sdg}}  u_{l+1}
          \red_{\dgarsh{\atrs}{\sdg}}  \ldots
      \end{multline*}
      in $\dgarsh{\atrs}{\sdg}$
      with $\lim_{l\to\infty} \nguards{u_l} = N$.
      Let $l_0\in\nat$ be minimal such that $\nguards{u_{l_0}} = N$.
      Furthermore,
      let, for all $i\in\{1,\ldots,n\}$,
      $m'_i \le m_i$ be minimal such that each of the topmost
      $m'_i$ symbols ``$\pebble$''
      of the subterm $\strcns{\trspeb^{m_i}}{\astr}$ of
      the source
      $ \funap{\astrfun}{\strcns{\trspeb^{m_1}}{\astr},\ldots,\strcns{\trspeb^{m_n}}{\astr}} $
      of $\xi$ is part of a redex pattern at some step among the
      first $l_0$ steps of $\xi$.

    Since by assumption, we have
    $m_i \defdby \doLowh{\atrs}{s_i}$ for $i\in\{1,\ldots,n\}$,
    there exist, for each $i\in\{1,\ldots,n\}$,
    history-aware data-guess functions  $\sdg_i$,
    and om-fair strategies $\bars_i\subseteq\dgarsh{\atrs}{\sdg_i}$
    such that $\terprd{\bars_i}{\databstr{s_i}} = m_i $;
    let all $\bars_i$ and $\sdg_i$ be chosen as such.
    Then there exist,
    for each $i$, om-fair rewrite sequences of the form:
    \begin{equation*}
      \xi_i : \;
        \databstr{s_i}
          \red_{\dgarsh{\atrs}{\sdg_i}}  u_{i,1}
          \red_{\dgarsh{\atrs}{\sdg_i}}  \ldots
          \red_{\dgarsh{\atrs}{\sdg_i}}  u_{i,l_i}
          = \strcns{\trspeb^{m'_i}}{\tilde{u}_{i,l_i}}
          \red_{\dgarsh{\atrs}{\sdg_i}} \ldots
    \end{equation*}
    in $\dgarsh{\atrs}{\sdg}$
    with $\lim_{l\to\infty}{\nguards{u_{i,l}}} = m_i$.

    Now it is possible to combine the history-aware data-exchange functions
    $\sdg$, $\sdg_1$, \ldots, $\sdg_n$ into a function $\sdg'$
    that makes it possible to combine
    the rewrite sequences $\xi$ and $\xi_1$, \ldots, $\xi_n$
    into a rewrite sequence of the form:
    \begin{align*}
      \xi' : \;
      &
      \databstr{\funap{\astrfun}{{s_1},\ldots,{s_n}}}
        \\
      & \hspace*{1.5ex}
      \mred_{\dgarsh{\atrs}{\sdg'}} \;
        \funap{\astrfun}{\strcns{\trspeb^{m'_1}}{u_{1,l_1}},\ldots,\strcns{\trspeb^{m'_n}}{u_{n,l_n}}}
        \\
      & \hspace*{5ex}
        \text{(carry out the first $l_i$ steps in $\xi_i$ in the context
               $\funap{\astrfun}{\porti{1},\ldots,\porti{n}}$)}
        \displaybreak[0]\\
      & \hspace*{1.5ex}
      \red_{\dgarsh{\atrs}{\sdg'}}  u'_1
      \red_{\dgarsh{\atrs}{\sdg'}}  u'_2
      \red_{\dgarsh{\atrs}{\sdg'}}  \ldots
      \red_{\dgarsh{\atrs}{\sdg'}}  u'_{l_0}
      = \strcns{\trspeb^{N}}{\tilde{u}'_{l_0}}
        \\
      & \hspace*{5ex}
        \text{(parallel to the first $l_0$ steps of $\xi$)}
        \displaybreak[0]\\
      & \hspace*{1.5ex}
      \red_{\dgarsh{\atrs}{\sdg'}} w_1
      \red_{\dgarsh{\atrs}{\sdg'}} w_2
      \red_{\dgarsh{\atrs}{\sdg'}} \ldots
        \\
      & \hspace*{5ex}
       \text{(fair interleaving of the rests of the $\xi_i$ put in context}
        \\[-0.5ex]
      & \hspace*{6ex}
       \text{and of steps parallel to the rest of $\xi'$)}
    \end{align*}
    By its construction, $\xi'$ is again om-fair
    and it holds:
    $ \lim_{l\to\infty} w_l = N$.
    Now let $\bars'$ be the extension
    of the sub-ARS of $\dgarsh{\atrs}{\sdg}$ induced by $\xi'$
    to a deterministic om-fair strategy for $\dgarsh{\atrs}{\sdg'}$.
    (Again choose an arbitrary deterministic
     om-fair strategy~$\tilde{\bars}$
     on $\dgarsh{\atrs}{\sdg}$, which is possible by
     Proposition~\ref{prop:omfair:strat}.
     On term-history pairs that do not occur in $\xi'$,
     define $\bars'$ according to $\tilde{\bars}$.)
    Then $\xi'$ is also a rewrite sequence in $\bars'$
    that witnesses
    $\terprd{\bars}{\databstr{\funap{\astrfun}{{s_1},\ldots,{s_n}}}}
        = \lim_{l\to\infty} \nguards{w_l}
        = N
        = \terprd{\bars}{\funap{\astrfun}{\strcns{\trspeb^{m_1}}{\astr},\ldots,\strcns{\trspeb^{m_n}}{\astr}}}
        $.
      Now it follows:
      \begin{multline*}
        \doLowh{\atrs}{\funap{\astrfun}{s_1,\ldots,s_n}}
        \le
        \terprd{\bars}{\databstr{\funap{\astrfun}{{s_1},\ldots,{s_n}}}}
        \\
        =
        \terprd{\bars'}{
            \funap{\astrfun}{\strcns{\trspeb^{m_1}}{\astr},\ldots,\strcns{\trspeb^{m_n}}{\astr}}
                        }
        =
        \doLowh{\atrs}{\funap{\astrfun}{\strcns{\trspeb^{m_1}}{\astr},\ldots,\strcns{\trspeb^{m_n}}{\astr}}} \; ,
      \end{multline*}
      establishing ``$\le$'' in \eqref{eq1:lem:do:lowh:id}.
      \qed
  \end{description}
\end{proof}

\begin{lemma}\label{lem:do:lowh:ind:contexts}
  Let $s \in \ter{\databstr{\asig}}_{\sortS}$ be a stream term
  and $\asubst \funin \var_{\sortS} \to \ter{\databstr{\asig}}_{\sortS}$
  a substitution of stream variables. Then
  \[\doLowh{\atrs}{\subst{s}{\asubst}} = \doLowh{\atrs}{\subst{s}{\bsubst}}\]
  where $\bsubst \funin \var_{\sortS} \to \ter{\databstr{\asig}}_{\sortS}$
  is defined by $\funap{\bsubst}{\bstr} = \strcns{\trspeb^{m_\bstr}}{\astr}$,
  and $m_{\bstr} \defdby \doLowh{\atrs}{\funap{\asubst}{\bstr}}$.
\end{lemma}
\begin{proof}
  Induction using Lem.~\ref{lem:do:lowh:ind}.
  \qed
\end{proof}

\begin{corollary}\label{lem:do:lowh:subst}
  Let $s \in \ter{\databstr{\asig}}_{\sortS}$
  and $\asubst, \bsubst \funin \var_{\sortS} \to \ter{\databstr{\asig}}_{\sortS}$
  substitutions of stream variables 
  such that $\doLowh{\atrs}{\funap{\asubst}{\bstr}} = \doLowh{\atrs}{\funap{\bsubst}{\bstr}}$ for all $\bstr \in \var_{\sortS}$.
   Then
  \[\doLowh{\atrs}{\subst{s}{\asubst}} = \doLowh{\atrs}{\subst{s}{\bsubst}}\punc.\]
\end{corollary}
\begin{proof}
  Two applications of Lem.~\ref{lem:do:lowh:ind:contexts}.
  \qed
\end{proof}

We define a history free data-exchange function $\sdgw{\atrs}$
such that no reduction with respect to $\sdgw{\atrs}$
produces more than $\doLowh{\atrs}{s}$ elements.
\begin{definition}\normalfont\label{def:do:worst}
  Let a stream specification $\atrs = \pair{\asig}{R}$ be given, and let $>$ be a well-founded order on $R$.
  We define the \emph{lower bound data-exchange function
  $\sdgw{\atrs} \funin \ter{\databstr{\asig}}_{\sortS} \to \ter{\asig}_{\sortS}$} as follows.

  Let $s \in \ter{\databstr{\asig}}_{\sortS}$ and
  $\funap{\astrfun}{s_1,\ldots,s_{\arityS{\astrfun}},\pebble_1,\ldots,\pebble_{\arityD{\astrfun}}}$
  a subterm occurrence of $s$.
  For $i = 1,\ldots,\arityS{\astrfun}$ let $n_i \in \nat$ be maximal
  such that $s_i \equiv \strcns{\trspeb_{i,1}}{\ldots \strcns{\trspeb_{i,n_i}}{s'_i}}$.
  We define $m_i \defdby \doLowh{\atrs}{s_i}$ for $1 \le i \le n$,
  then we have by Lem.~\ref{lem:do:lowh:ind}:
  \begin{gather*}
    \doLowh{\atrs}{t \defdby \funap{\astrfun}{s_1,\ldots,s_{\arityS{\astrfun}}, \pebble_1,\ldots,\pebble_{\arityD{\astrfun}}}} = \hspace{3cm}\\
    \hspace{3cm}
    \doLowh{\atrs}{t' \defdby \funap{\astrfun}{\strcns{\trspeb^{m_1}}{\astr},\ldots,\strcns{\trspeb^{m_{\arityS{\astrfun}}}}{\astr},\pebble_1,\ldots,\pebble_{\arityD{\astrfun}}}}
    \punc.
  \end{gather*}
  We choose a defining rule $\rho$ of $\astrfun$ as follows.
  In case there exists a rule
  \[\gamma :
   \funap{\astrfun}{
      \strcns{\vec{\adattrm}_1}{\astr_1},
      \ldots,
      \strcns{\vec{\adattrm}_n}{\astr_n},
      \bdattrm_1,\ldots,\bdattrm_{\arityD{\strff{f}}}
    }
   \red t \in R\]
  such that $m_i < \lstlength{\vec{\adattrm}_i}$ for some $i \in \{1,\ldots,n\}$, then let $\rho \defdby \gamma$.
  In case there are multiple possible choices for $\gamma$, we pick the minimal $\gamma$ with respect to $>$.

  Otherwise
  in exists a history-aware data-exchange function $\sdg$, and an outermost-fair rewrite sequence 
  $\aseq : t' \red_{\dgarsh{\atrs}{\sdg}} \ldots$ in $\dgarsh{\atrs}{\sdg}$ producing only $\doLowh{\atrs}{t'}$ elements.
  From exhaustivity of $\atrs$ we get $\dg{t'}$ is not a normal form, since all rules have enough supply.
  Moreover, by orthogonality exactly one defining rule $\gamma$ of $\astrfun$ is applicable, let $\rho \defdby \gamma$.
  Again, in case there are multiple possible choices for $\gamma$, due to freedom in the choice of the data-exchange function $\sdg$,
  we take the minimal $\gamma$ with respect to $>$.

  We define $\sdgw{\atrs}$ to instantiate:
  \begin{itemize}
    \item for $1 \le i \le \arityD{\astrfun}$ the occurrences of pebbles $\trspeb_i$ in $s$, and
    \item for $1 \le i \le \arityS{\astrfun}$, $1 \le j \le n_i$ the occurrences of pebbles $\trspeb_{i,j}$ in $s$
  \end{itemize}
  with respect to the rule $\rho$ (Def.~\ref{def:inst:rho}).
\end{definition}

Rewrite steps with respect to the lower bound data-exchange function $\sdgw{\atrs}$
do not change the \daob{} lower bound of the production of a term.
\begin{lemma}\label{lem:lowdg}
  For all $\astrtrm,\astrtrm \in \ter{\databstr{\asig}}_{\sortS}$: 
  $\astrtrm \red_{\dgars{\atrs}{\sdgw{\atrs}}} \bstrtrm$
  implies $\doLowh{\atrs}{\astrtrm} = \doLowh{\btrs}{\bstrtrm}$.
\end{lemma}
\begin{proof}
  We use the notions introduced in Def.~\ref{def:do:worst}.
  Note that the lower bound data-exchange function $\sdgw{\atrs}$
  is defined in a way that the pebbles are instantiated independent of 
  the context above the stream function symbol.
  Therefore it is sufficient to consider rewrite steps 
  $\astrtrm \red_{\dgars{\atrs}{\sdgw{\atrs}}} \bstrtrm$ at the root,
  closure under contexts follows from Cor.~\ref{lem:do:lowh:subst}. 
  \qed
\end{proof}

\begin{corollary}\label{cor:lowdg:prod}
  For all $\astrtrm \in \ter{\databstr{\asig}}_{\sortS}$ we have: 
  $\terprd{\dgars{\atrs}{\sdgw{\atrs}}}{\astrtrm} = \doLowh{\atrs}{\astrtrm}$.
\end{corollary}
\begin{proof}
  Direct consequence of Lem.~\ref{lem:lowdg}.
  \qed
\end{proof}

Hence, as a consequence of Cor.~\ref{cor:lowdg:prod}, the lower bound of the history-free \daob{} production
conincides with the lower bound of the history-aware \daob{} production.
\begin{lemma}\label{do:low:le:lowh}
  For all $\astrtrm \in \ter{\asig}_{\sortS}$
  we have: $\doLow{\atrs}{\astrtrm} = \doLowh{\atrs}{\astrtrm}$.
\end{lemma}
\begin{proof}
  The direction $\doLow{\atrs}{\astrtrm} \ge \doLowh{\atrs}{\astrtrm}$ is trivial, and
  $\doLow{\atrs}{\astrtrm} \le \doLowh{\atrs}{\astrtrm}$ follows from Cor.~\ref{cor:lowdg:prod},
  that is, using the history-free data-exchange function $\sdgw{\atrs}$
  no reduction produces more than $\doLowh{\atrs}{\astrtrm}$ elements.
  \qed
\end{proof}

\begin{lemma}
  For all $s \in \ter{\asig}_{\sortS}$
  we have: $\doRng{\atrs}{\astrtrm} = \doRngh{\atrs}{\astrtrm}$.
\end{lemma}
\begin{proof}
\begin{itemize}
\item [$\subseteq$]
  Let $n \in \doRng{\atrs}{\astrtrm}$. Then there exists a data-exchange function $\sdg$
  such that $\terprd{\dgars{\atrs}{\sdg}}{\astrtrm} = n$, and
  a rewrite sequence $\sigma : s_0 \to_{\dgars{\atrs}{\sdg}} \ldots \to_{\dgars{\atrs}{\sdg}} s_m$
  with $s_m \equiv \strcns{\adattrm_1}{\strcns{\ldots}{\strcns{\adattrm_{n}}{\bstrtrm}}}$.
  We define for all $h \in \mathcal{H}_\atrs$ a data-exchange function $\sdgh{h}$ by
  setting $\dgh{h}{s} \defdby \dg{s}$ for all $s \in \ter{\databstr{\asig}}_{\sortS}$.
  Moreover, we define the strategy
  $\bars \subseteq \dgarsh{\atrs}{\sdg}$
  to execute $\sigma$ and continue outermost-fair afterwards.
  Since $s_m$ does not produce more than $n$ elements,
  every maximal rewrite sequence with respect to $\bars$ will produce exactly $n$ elements.
  Hence $n \in \doRngh{\atrs}{\astrtrm}$.
\item [$\supseteq$]
  Let $n \in \doRngh{\atrs}{\astrtrm}$. Then there exist
  data-exchange functions $\sdgh{h}$ for every $h \in \mathcal{H}_\atrs$,
  and an outermost-fair strategy $\bars \subseteq \dgarsh{\atrs}{\sdg}$
  such that $\terprd{\bars}{\astrtrm} = n$.
  There exists a (possibly infinite) maximal, outermost-fair rewrite sequence
  \[\sigma : \pair{s}{\lstemp} \equiv \pair{s_0}{h_0} \to_\bars \ldots \to_\bars \pair{s_m}{h_m} \to_\bars \ldots\]
  such that $s_m \equiv \strcns{\adattrm_1}{\strcns{\ldots}{\strcns{\adattrm_{n}}{\bstrtrm}}}$.
  Since the elements of $\dgarsh{\atrs}{\sdg}$ are annotated with their history,
  we do not encounter repetitions during the rewrite sequence.
  Hence, w.l.o.g.\ we can assume that $\bars$ is a deterministic strategy,
  admitting only the maximal rewrite sequence $\sigma$.
  Disregarding the history annotations there might exist repetitions,
  that is, $i < j \in \nat$ with $s_i \equiv s_j$.
  Nevertheless, from $\sigma$ we can construct a history free data-exchange function $\sdg$ as follows.
  For all $s \in \ter{\databstr{\asig}}_{\sortS}$ let $S \defdby \{i \in \nat \mid s_i \equiv s\}$ and define:
  \begin{align*}
    \funap{\sdg'}{s} \defdby
    \begin{cases}
      \dgh{h_j}{s_j} &\text{if } j \defdby \sup S  < \conattop\\
      \dgh{h_k}{s_k} &\text{if } \sup S = \conattop,\; k \defdby \inf S\\
      \text{arbitrary} &\text{if } \neg \myex{i \in \nat}{s_i \equiv s}
    \end{cases}
  \end{align*}
  \qed
  
\end{itemize}
\end{proof}

\section{The Production Calculus}\label{sec:nets}
As a means to compute the \daob{} 
production behaviour of stream specifications,
we introduce a `production calculus'
with periodically increasing functions as its
central ingredient.
This calculus is the term equivalent of the calculus
of pebbleflow nets that was introduced in~\cite{endr:grab:hend:isih:klop:2007},
see Sec.~\ref{ssec:pebbleflow}.

\subsection{Periodically Increasing Functions}

We use $\conat \defdby \nat \uplus \{\conattop\}$,
the \emph{extended natural numbers}, with
the usual $\le$, $+$, and we define 
$\conattop - \conattop \defdby 0$.
An infinite sequence $\aseq \in \str{X}$ is \emph{eventually periodic}
if $\aseq = \alst\blst\blst\blst\ldots$ for 
some $\alst\in\lst{X}$ and $\blst\in\nelst{X}$.
A function $f\funin\setfun{\nat}{\conat}$ is
\emph{eventually periodic} if the sequence
$\langle \funap{f}{0}, \funap{f}{1}, \funap{f}{2}, \ldots \rangle$
is eventually periodic.
\begin{definition}\normalfont
  A function  $g \funin \setfun{\nat}{\conat}$ is called
  \emph{periodically increasing} if it is increasing, 
  i.e.\ $\myall{n}{\funap{g}{n}\le\funap{g}{n+1}}$,
  and the \emph{derivative of $g$},
  $n\mapsto \funap{g}{n+1} - \funap{g}{n}$, 
  is eventually periodic.
  A function  $h \funin \setfun{\conat}{\conat}$ is called
  \emph{periodically increasing} if its restriction to $\nat$ is
  periodically increasing and if 
  $ h(\conattop) = \lim_{n\to\infty} h(n) $.
  Finally, a $k$-ary function $i \funin \setfun{(\conat)^k}{\conat}$ 
  is called \emph{periodically increasing} if 
  $\funap{i}{n_1,...,n_k} = \funap{\min}{\funap{i_1}{n_1}, \ldots ,\funap{i_k}{n_k}}$
  for some unary periodically increasing functions
  $i_1,\ldots,i_k$.
\end{definition}

Periodically increasing (\pein) functions can be denoted by their value at 0 
followed by a representation of their derivative.
For example, $03\overline{12}$ denotes 
the \pein{} function $f \funin \setfun{\nat}{\nat}$ with values
$ 0, 3, 4, 6, 7, 9, \ldots $.
However, we use a finer and more flexible notation over the alphabet $\io$
that will be useful for computing several operations on \pein~functions. 
For instance, we represent $f$ as above by the `\ioseq' 
\begin{equation*}
  \ios{-+++}{-+-++}
  =
  \iosi{-+++}{\,\iosi{-+-++}{\,\iosi{-+-++}{\,\iosi{-+-++}{\,\ldots}}}}
  \punc,
\end{equation*}
in turn abbreviated by the `\ioterm'
$\pair{\iosi{-+++}{\nix}}{\iosi{-+-++}{\nix}}$.

\begin{definition}\normalfont\label{def:iosq}
  Let $\pm \defdby \io$.
  An \emph{\ioseq} is a finite ($\in\iolst$) 
  or infinite sequence ($\in\iostr$) over $\pm$,
  the set of \ioseq{s} is denoted by 
  $\iosq \defdby \setunion{\iolst}{\iostr}$.
  We let $\iosqemp$ denote the empty \ioseq,
  and use it as an explicit end marker.
  An \emph{\ioterm} is a pair $\pair{\aiolst}{\biolst}$
  of finite \ioseq{s} $\aiolst,\biolst\in\iolst$
  with $\biolst\ne\iosqemp$.
  The set of \ioterm{s} is denoted by $\iotrm$, 
  and we use $\aiotrm,\biotrm$ to range over \ioterm{s}.
  For an \ioseq~$\aiosq\in\iostr$
  and an \ioterm~$\aiotrm = \pair{\aiolst}{\biolst}\in\iotrm$
  we say that \emph{$\aiotrm$ denotes $\aiosq$} %
  if $\aiosq = \aiolst\lstcyc{\biolst}$
  where $\lstcyc{\biolst}$ stands for the infinite sequence 
  $\biolst\biolst\biolst\ldots$.
  A sequence $\aiosq\in\iosq$ is \emph{rational}
  if it is finite or if it is denoted %
  by an \ioterm.
  The set of rational \ioseq{s} is denoted by $\iosqrat$. 
  An infinite sequence $\aiosq\in\iostr$ is \emph{productive} 
  if it contains infinitely many $\ioout$'s:
  \[
    \text{$\aiosq\in\iostr$ is \emph{productive}}
    \,\,\,\Longleftrightarrow\,\,\,
    \myall{n}{\myex{m\geq n}{\strnth{\aiosq}{m}=\ioout}}
    \punc.
  \]
  We let $\prodiostr$ denote the set of productive sequences 
  and define $\prodiosq\defdby\setunion{\iolst}{\prodiostr}$.
\end{definition}
The reason for having both \ioterm{s} and \ioseq{s}
is that, depending on the purpose at hand, 
one or the other is more convenient to work with. 
Operations are often easier to understand on \ioseq{s} 
than on \ioterm{s}, whereas we need finite representations 
to be able to compute these operations.

\begin{definition}\normalfont\label{def:iosq:prd}
  We define $\siosqprd{\aiosq} \funin \setfun{\nat}{\conat}$,
  the \emph{interpretation} of $\aiosq\in\iosq$, by:
  \begin{align}
    \iosqprd{\iosqemp}{n} &= 0
    \label{eq:iosqemp:def:iosq:prd}
    \\
    \iosqprd{\iosqcns{\ioout}{\aiosq}}{n} 
    &= 1 + \iosqprd{\aiosq}{n}
    \label{eq:ioout:def:iosq:prd}
    \\
    \iosqprd{\iosqcns{\ioin}{\aiosq}}{0} &= 0
    \label{eq:ioin:zero:def:iosq:prd}
    \\
    \iosqprd{\iosqcns{\ioin}{\aiosq}}{n+1}
    &= \iosqprd{\aiosq}{n} 
    \label{eq:ioin:succ:def:iosq:prd}
  \end{align}
  for all $n\in\nat$, and extend it to $\conat\to\conat$
  by defining $\iosqprd{\aiosq}{\conattop} = \lim_{n\to\infty}{\iosqprd{\aiosq}{n}}$.
  We say that $\siosqprd{\aiosq}$ \emph{interprets} $\aiosq$, and,
  conversely, that $\aiosq$ \emph{represents} $\siosqprd{\aiosq}$.
  We overload notation and define $\ssiotrmprd\funin\iotrm\to(\conat\to\conat)$ by
  $\siotrmprd{\pair{\aiolst}{\biolst}}\defdby\siosqprd{\iosqcat{\aiolst}{\iosqcyc{\biolst}}}$.
\end{definition}
It is easy to verify that, for every $\aiosq\in\iosq$,
the function $\siosqprd{\aiosq}$ is increasing,
and that, for every $\aiotrm\in\iotrm$, the function %
$\siotrmprd{\aiotrm}$ is periodically increasing.
Furthermore, every increasing function is represented by an \ioseq,
and every \pein{} function is denoted by an \ioterm.

Subsequently, for an eventually constant function $f\funin\conat\to\conat$,
we write $\iosqrep{f}$ for the shortest finite \ioseq\ representing $f$
(trailing $\ioin$'s can be removed). 
For an always eventually strictly increasing function $f\funin\conat\to\conat$,
i.e.\ $\myall{n}{\myex{m>n}{\funap{f}{m} > \funap{f}{n}}}$,
we write $\iosqrep{f}$ for the unique \ioseq\ representing $f$;
note that then $\iosqrep{f}\in\prodiostr$ follows.
For a periodically increasing function $f\funin\conat\to\conat$, 
we write $\iotrmrep{f}$ for the \ioterm\ 
$\aiotrm = \pair{\aiolst}{\biolst}$
such that $\siotrmprd{\aiotrm} = f$ 
and $\lstlength{\aiolst} + \lstlength{\biolst}$ is minimal; 
uniqueness follows from the following lemma:
\begin{lemma}\label{lem:shortest:ioterm}
  For all \pein\ functions $f\funin\conat\to\conat$,
  the term $\iotrmrep{f}\in\iotrm$ is unique. 
\end{lemma}
\begin{proof}
  Consider the `compression' TRS consisting of the rules:
  \begin{align}
    \pair{\aiolst x}{\biolst x} &\red \pair{\aiolst}{x\biolst} 
    & (x\in\pm)
    \label{rule:roll}
    \\
    \pair{\aiolst}{\biolst^p} &\red \pair{\aiolst}{\biolst} 
    & (n > 1)
    \label{rule:fold}
  \end{align}
  Clearly, $\aiotrm\red\biotrm$ implies 
  $\iotrmlength{\aiotrm} > \iotrmlength{\biotrm}$
  where, for $\aiotrm = \pair{\aiolst}{\biolst}$,
  $\iotrmlength{\aiotrm} \defdby \lstlength{\aiolst}+\lstlength{\biolst}$.
  Hence, compression is terminating.
  
  There are two critical pairs due to an overlap 
  of rule~\eqref{rule:fold} with itself, and
  in rules~\eqref{rule:roll} and~\eqref{rule:fold}:
  $\pair{\pair{\aiolst}{\biolst^n}}{\pair{\aiolst}{\biolst^m}}$
  originating from a term of the form $\pair{\aiolst}{\biolst^{n\cdot m}}$,
  and
  $\pair{\pair{\aiolst}{(x \biolst)^n}}{\pair{\aiolst x}{\biolst x}}$
  originating from a term of the form 
  $\pair{\aiolst x}{(\biolst x)^n}$,
  respectively.
  Both pairs are easily joinable in one step:
  the first to $\pair{\aiolst}{\biolst}$,
  and the second to $\pair{\aiolst}{x\biolst}$.
  Hence, the system is locally confluent. 
  Therefore, by Newman's Lemma, it is also confluent, and normal forms are unique.

  Finally, assume that two \ioterm{s} 
  $\pair{\iaiolst{1}}{\ibiolst{1}},\pair{\iaiolst{2}}{\ibiolst{2}}\in\iotrm$ 
  each of minimal length represent the same \pein\ function $f$.
  It is easy to  see that then 
  $\iaiolst{1}\iosqcyc{\ibiolst{1}} = \iaiolst{2}\iosqcyc{\ibiolst{2}}$.
  It follows that there exist $\pair{\aiolst}{\biolst}$ such that 
  $\pair{\aiolst}{\biolst}\red^\ast \pair{\iaiolst{1}}{\ibiolst{1}}$
  and
  $\pair{\aiolst}{\biolst}\red^\ast \pair{\iaiolst{2}}{\ibiolst{2}}$.
  By confluence and termination of compression 
  it follows that they have the same normal form 
  $\pair{\iaiolst{0}}{\ibiolst{0}}$.
  Since compression reduces the length 
  of \ioterm{s} it must follow that
  $\pair{\iaiolst{1}}{\ibiolst{1}} = \pair{\iaiolst{0}}{\ibiolst{0}} = \pair{\iaiolst{2}}{\ibiolst{2}}$.
  \qed
\end{proof}

Note that we have that $\siosqprd{\iosqrep{f}} = f$ for all increasing functions $f$, 
and $\siotrmprd{\iotrmrep{f}} = f$ for all \pein\ functions $f$.
As an example one can check that the interpretation of the aforementioned \ioterm\ 
$\trmrep{f} = \pair{\iosi{-+++}{\nix}}{\iosi{-+-++}{\nix}}$
indeed has the sequence $0,3,4,6,7,9,\ldots$ as its graph.

\begin{remark}\label{rem:iosq:prd:wd}
  Note that Def.~\ref{def:iosq:prd} is well-defined due to 
  the productivity requirement on infinite \ioseq{s}. 
  This can be seen as follows: starting on $\iosqprd{\aiosq}{n}$, 
  after finitely many, say $m\in\nat$, $\ioin$'s, 
  we either arrive at rule~\eqref{eq:ioout:def:iosq:prd} (if $m\leq n$) 
  or at rule~\eqref{eq:ioin:zero:def:iosq:prd} (if $m > n$).
  Had we not required infinite sequences $\aiosq\in\iostr$ 
  to contain infinitely many $\ioout$'s, then, e.g., the computation of
  $\iosqprd{\iosqcns{\ioin}{\iosqcns{\ioin}{\iosqcns{\ioin}{\ldots}}}}{\conattop}$ 
  would consist of applications of rule~\eqref{eq:ioin:succ:def:iosq:prd} only,
  and hence would not lead to an infinite normal form.
\end{remark}

\begin{proposition}
  Unary periodically increasing functions are closed under
  composition and pointwise infimum. 
\end{proposition}
We want to constructively define the operations of composition, pointwise infimum, 
and least fixed point calculation of increasing functions (\pein\ functions) 
on their \ioseq\ (\ioterm) representations.
We first define these operations on \ioseq{s} by means of coinductive clauses. %
The operations on \ioterm{s} are more involved, 
because they require `loop checking' on top.
We will proceed by showing that the operations 
of composition and pointwise infimum on \ioseq{s}  %
preserve rationality. This %
will then give rise to the needed algorithms 
for computing the operations on \ioterm{s}.

\begin{definition}\normalfont\label{def:iosq:cmp}
  The operation 
  \emph{composition
    $\siosqcmp\funin\setfun{\prodiosq\times\prodiosq}{\prodiosq}$,
    $\pair{\aiosq}{\biosq}\mapsto\iosqcmp{\aiosq}{\biosq}$
    of (finite or productive) \ioseq{s}%
  }
  is defined coinductively as follows:
  \begin{align}
    \iosqcmp{\iosqemp}{\biosq}
    &= \iosqemp
    \label{def:iosqcmp:clause:leftemp}
    \\
    \iosqcmp{\iosqcns{\ioout}{\aiosq}}{\biosq}
    &= \iosqcns{\ioout}{(\iosqcmp{\aiosq}{\biosq})} 
    \label{def:iosqcmp:clause:output}
    \\
    \iosqcmp{\iosqcns{\ioin}{\aiosq}}{\iosqemp}
    &= \iosqemp 
    \label{def:iosqcmp:clause:rightemp}
    \\
    \iosqcmp{\iosqcns{\ioin}{\aiosq}}{\iosqcns{\ioout}{\biosq}}
    &= \iosqcmp{\aiosq}{\biosq} 
    \label{def:iosqcmp:clause:internal}
    \\
    \iosqcmp{\iosqcns{\ioin}{\aiosq}}{\iosqcns{\ioin}{\biosq}} 
    &= \iosqcns{\ioin}{(\iosqcmp{\iosqcns{\ioin}{\aiosq}}{\biosq})}
    \label{def:iosqcmp:clause:input}
  \end{align}
\end{definition}
An argument similar to Rem.~\ref{rem:iosq:prd:wd} concerning well-definedness 
applies for Def.~\ref{def:iosq:cmp}.
\begin{remark}\label{rem:iosqcmp:well-defined}
  Note that the defining 
  rules~\eqref{def:iosqcmp:clause:leftemp}--\eqref{def:iosqcmp:clause:input}
  are orthogonal and exhaust all cases.
  As it stands, it is not immediately clear that this definition 
  is well-defined (productive!), 
  the problematic case being \eqref{def:iosqcmp:clause:internal}.
  Rule~\eqref{def:iosqcmp:clause:internal} is to be thought of as 
  an internal communication between components $\aiosq$ and $\biosq$, 
  as a silent step in the black box~$\iosqcmp{\aiosq}{\biosq}$.
  How to guarantee that always,
  after finitely many internal steps,
  either the process ends or 
  there will be an external step, in the form of
  output or a requirement for input?
  The recursive call in~\eqref{def:iosqcmp:clause:internal}
  is not guarded. However, well-definedness of composition 
  can be justified as follows:
  
  \noindent
  Consider arbitrary sequences $\aiosq,\biosq\in\prodiosq$.
  By definition of $\prodiosq$, there exists $m\in\nat$ such that
  $\aiosq = \iosqcat{\ioin^m}{\iosqcns{\ioout}{\aiosq'}}$ 
  for some $\aiosq'\in\prodiosq$
  or $\aiosq = \iosqcat{\ioin^m}{\iosqemp}$; 
  likewise there exists $n\in\nat$
  such that $\biosq = \iosqcat{\ioin^n}{\iosqcns{\ioout}{\biosq'}}$ 
  for some $\biosq'\in\prodiosq$,
  or $\biosq = \iosqcns{\ioin^n}{\iosqemp}$.
  Rules~\eqref{def:iosqcmp:clause:internal} and~\eqref{def:iosqcmp:clause:input}
  are decreasing with respect to the lexicographic order on $(m,n)$.
  After finitely many applications of~\eqref{def:iosqcmp:clause:internal} 
  and~\eqref{def:iosqcmp:clause:input}, one of the 
  rules~\eqref{def:iosqcmp:clause:leftemp}--\eqref{def:iosqcmp:clause:rightemp}
  must be applied.
  Hence composition is well-defined and the sequence produced is an element of $\prodiosq$,
  i.e.\ either it is a finite sequence %
  or it contains an infinite number of $\ioout$'s.
\end{remark}

\begin{lemma}\label{lem:iosq:cmp:assoc}
  Composition of \ioseqs\ is associative.
\end{lemma}
 \begin{proof}
  Let
  \(
    R =
    \{
      \pair
        {\iosqcmp{\aiosq}{(\iosqcmp{\biosq}{\ciosq})}}
        {\iosqcmp{(\iosqcmp{\aiosq}{\biosq})}{\ciosq}}
      \where
      \aiosq,\biosq,\ciosq\in\prodiosq
    \}
  \).
  To show that $R$ is a bisimulation,
  we prove that, for all $\aiosq,\biosq,\ciosq,\diosq_1,\diosq_2\in\prodiosq$,
  if $\diosq_1 = \iosqcmp{\aiosq}{(\iosqcmp{\biosq}{\ciosq})}$ and 
  $\diosq_2 = \iosqcmp{(\iosqcmp{\aiosq}{\biosq})}{\ciosq}$,
  then either $\diosq_1 = \iosqemp = \diosq_2$
  or $\seqhd{\diosq_1} = \seqhd{\diosq_2}$
  and $\pair{\seqtl{\diosq_1}}{\seqtl{\diosq_2}}\in R$,
  by induction on the number $n\in\nat$ of 
  leading $\ioin$'s of $\aiosq$\footnote{%
    By definition of $\prodiosq$ the number of leading 
    $\ioin$'s of any sequence in $\prodiosq$ is finite:
    for all $\aiosq\in\prodiosq$ there exists $n\in\nat$ such that
    either $\aiosq = \iosqcat{\ioin^{n}}{\iosqemp}$ 
    or $\aiosq = \iosqcat{\ioin^{n}}{\iosqcns{\ioout}{\aiosq'}}$.
  }, and a sub-induction on the number $m\in\nat$ 
  of leading $\ioin$'s of $\biosq$.

  If $n = 0$ and $\aiosq = \iosqemp$,
  then $\diosq_1 = \iosqemp = \diosq_2$.
  If %
  $\aiosq = \iosqcns{\ioout}{\aiosq'}$, we have 
  $\diosq_1 = \iosqcns{\ioout}{(\iosqcmp{\aiosq'}{(\iosqcmp{\biosq}{\ciosq})})}$, 
  and
  $\diosq_2 = \iosqcns{\ioout}{(\iosqcmp{(\iosqcmp{\aiosq'}{\biosq})}{\ciosq})}$,
  and so $\pair{\seqtl{\diosq_1}}{\seqtl{\diosq_2}}\in R$.

  If $n > 0$, then $\aiosq = \iosqcns{\ioin}{\aiosq'}$,
  and we proceed by sub-induction on $m$.
  If $m = 0$ and $\biosq = \iosqemp$,
  then $\diosq_1 = \iosqemp = \diosq_2$.
  If %
  $\biosq = \iosqcns{\ioout}{\biosq'}$,
  we compute
  $\diosq_1 = \iosqcmp{\aiosq'}{(\iosqcmp{\biosq'}{\ciosq})}$, and
  $\diosq_2 = \iosqcmp{(\iosqcmp{\aiosq'}{\biosq'})}{\ciosq}$,
  and conclude by IH.

  If $m > 0$, then $\biosq = \iosqcns{\ioin}{\biosq'}$,
  and we proceed by case distinction on $\ciosq$.
  If $\ciosq = \iosqemp$, 
  then $\diosq_1 = \iosqemp = \diosq_2$.
  If $\ciosq = \iosqcns{\ioout}{\ciosq'}$, 
  we compute 
  $\diosq_1 = \iosqcmp{\aiosq}{(\iosqcmp{\biosq'}{\ciosq'})}$, and
  $\diosq_2 = \iosqcmp{(\iosqcmp{\aiosq}{\biosq'})}{\ciosq'}$,
  and conclude by sub-IH.
  Finally, if $\ciosq = \iosqcns{\ioin}{\ciosq'}$, then 
  $\diosq_1 = \iosqcns{\ioin}{(\iosqcmp{\aiosq}{(\iosqcmp{\biosq}{\ciosq'})})}$, and
  $\diosq_2 = \iosqcns{\ioin}{(\iosqcmp{(\iosqcmp{\aiosq}{\biosq})}{\ciosq'})}$.
  Clearly $\pair{\seqtl{\diosq_1}}{\seqtl{\diosq_2}}\in R$.
  \qed
\end{proof}

\begin{remark}\label{rem:iosq:prd:cmp}
  If we allow to use extended natural numbers at places
  where an \ioseq{} is expected, using a coercion
  $n \mapsto \iosqcns{\ioout}{}^n$, with 
  $\iosqcns{\ioout}{}^\numzer = \iosqemp$,
  one observes that the interpretation of an \ioseq\
  (Def.~\ref{def:iosq:prd}) is just a special case of composition:
  \[
    \iosqprd{\aiosq}{n} = \iosqcmp{\aiosq}{n}
    \punc.
  \]
\end{remark}

\begin{proposition}\label{prop:iosq:cmp:interpret}
  For all increasing functions $f,g\funin\conat\to\conat$:\:
  \(
    \siosqprd{\iosqcmp{\trmrep{f}}{\trmrep{g}}} 
    = \funcmp{\siosqprd{\trmrep{f}}}{\siosqprd{\trmrep{g}}}
  \).
\end{proposition}
\begin{proof}
  Immediate from Rem.~\ref{rem:iosq:prd:cmp},
  and Lem.~\ref{lem:iosq:cmp:assoc}.
  \qed
\end{proof}

Next, we show that composition of \ioseq{s} preserves rationality.
\begin{lemma}%
  \label{lem:iosq:cmp:rat}
  If $\aiosq,\biosq\in\iosqrat$, then $\iosqcmp{\aiosq}{\biosq}\in\iosqrat$.
\end{lemma}
\begin{proof}
  Let $\aiosq,\biosq\in\iosqrat$ %
  and set $\aiosq_0 \defdby \aiosq$ and $\biosq_0 \defdby \biosq$.
  In the rewrite sequence starting with $\iosqcmp{\aiosq_0}{\biosq_0}$ 
  each of the steps is either of the form:
  \begin{equation}\label{thm:rat:composition:eq:a}
    \iosqcmp{\aiosq_{n}}{\biosq_{n}} 
    \to C_{n}[\iosqcmp{\aiosq_{n+1}}{\biosq_{n+1}}]\punc,
  \end{equation}
  where
  $\aiosq_{n+1} \in \{\aiosq_n, \iosqtl{\aiosq_n}\}$, 
  $\biosq_{n+1} \in \{\biosq_n, \iosqtl{\biosq_n}\}$
  and $C_n \in \{\square, \iosqexplcns{\ioout}{\square}, \iosqexplcns{\ioin}{\square}\}$,
  or the rewrite sequence ends with a step of the form:
  \begin{equation*}
    \iosqcmp{\aiosq_{n}}{\biosq_{n}} \to \iosqemp \punc.
  \end{equation*}
  In the latter case, 
  $\iosqcmp{\aiosq}{\biosq}$ results in a finite and hence rational \ioseq.
  
  \renewcommand{\blst}{\theta}
  Otherwise the rewrite sequence is infinite, 
  consisting of steps of form~\eqref{thm:rat:composition:eq:a} only.
  For $\blst \in \iolst$ we define the context $C_{\blst}$ inductively: 
  $ C_{\blst} \equiv \square $, for $ \blst \equiv \bot $;
  and
  $ C_{+\blst} \equiv + C_{\blst} $ 
  as well as
  $ C_{-\blst} \equiv - C_{\blst} $, 
  for all $ \blst\in\iolst$.
  Because the sequences $\aiosq,\biosq$ are rational,
  the sets $\{\aiosq_n \where n \in \nat\}$ 
  and $\{\biosq_n \where n \in \nat\}$ are finite.
  Then, the pigeonhole principle implies the existence of 
  $i,j \in \nat$ such that $i < j$, $\aiosq_i = \aiosq_j$ and $\biosq_i = \biosq_j$.
  Now let $\alst, \clst \in \iolst$ with $\clst \neq \iosqemp$
  such that
  $C_{\alst} \equiv C_0[ \ldots C_{i-1}[C_i] \ldots ]$,
  $C_{\clst} \equiv C_{i+1}[ \ldots C_{j-1}[C_j] \ldots ]$,
  and:
  \begin{align*}
    \iosqcmp{\aiosq_0}{\biosq_0} 
    &\mred 
    C_{\alst} [ \iosqcmp{\aiosq_i}{\biosq_i} ] \:\text{, and}
    \\
    \iosqcmp{\aiosq_i}{\biosq_i}  
    &\mred
    C_{\clst} [ \iosqcmp{\aiosq_j}{\biosq_j} ]
    =  %
    C_{\clst} [ \iosqcmp{\aiosq_i}{\biosq_i} ] 
    \:\text{.}
  \end{align*}
  Then we find an eventually `spiralling' rewrite sequence:
  \begin{align*}
    \iosqcmp{\aiosq_0}{\biosq_0}   
    & \mred  
    C_{\alst} [ \iosqcmp{\aiosq_i}{\biosq_i} ] 
    = %
    \iosqcat{\alst}{(\iosqcmp{\aiosq_i}{\biosq_i})}
    \\
    &\mred 
    C_{\alst} [ C_{\clst}[ \iosqcmp{\aiosq_i}{\biosq_i} ] ] 
    = %
    \iosqcat{\iosqcat{\alst}{\clst}}{(\iosqcmp{\aiosq_i}{\biosq_i})}
    \\
    &\mred 
    C_{\alst} [ C_{\clst}[ C_{\clst} [ \iosqcmp{\aiosq_i}{\biosq_i} ] ] ]
    = %
    \iosqcat{\iosqcat{\iosqcat{\alst}{\clst}}{\clst}}{(\iosqcmp{\aiosq_i}{\biosq_i})} 
    \:\text{,}
  \end{align*}
  and therefore $\iosqcmp{\aiosq_0}{\biosq_0} \infred \iosqcat{\alst}{\lstcyc{\clst}}$.
  This shows that $\iosqcmp{\aiosq}{\biosq}$ is rational.
  \qed
\end{proof}

Next, we define the operation of pointwise infimum of increasing functions on \ioseqs.
\begin{definition}\normalfont\label{def:iosq:inf}
  The operation 
  \emph{pointwise infimum
    $\siosqinf\funin\iosq\times\iosq\to\iosq$,
    $\pair{\aiosq}{\biosq}\mapsto\iosqinf{\aiosq}{\biosq}$
    of \ioseq{s}%
  } 
  is defined coinductively by the following equations:
  \begin{align*}
    \iosqinf{\iosqemp}{\biosq}
    &= \iosqemp
    &
    \iosqrmioin{\iosqemp}
    &= \iosqemp
    \\
    \iosqinf{\aiosq}{\iosqemp}
    &= \iosqemp
    &
    \iosqrmioin{\iosqcns{\ioout}{\aiosq}}
    &= \iosqcns{\ioout}{\iosqrmioin{\aiosq}}
    \\
    \iosqinf{\iosqcns{\ioout}{\aiosq}}{\iosqcns{\ioout}{\biosq}}
    &= \iosqcns{\ioout}{(\iosqinf{\aiosq}{\biosq})}
    &
    \iosqrmioin{\iosqcns{\ioin}{\aiosq}}
    &= \aiosq
    \\
    \iosqinf{\iosqcns{\ioin}{\aiosq}}{\biosq}
    &= \iosqcns{\ioin}{(\iosqinf{\aiosq}{\iosqrmioin{\biosq}})}
    \dquad\text{($\biosq \neq \iosqemp$)}
    \\
    \iosqinf{\aiosq}{\iosqcns{\ioin}{\biosq}} 
    &= \iosqcns{\ioin}{(\iosqinf{\iosqrmioin{\aiosq}}{\biosq})}
    \dquad\text{($\aiosq \neq \iosqemp$)}
  \end{align*}
  where $\iosqrmioin{\aiosq}$ removes the first requirement of $\aiosq\in\iosq$ (if any).
\end{definition}
We let $\siosqcmp$ bind stronger than $\siosqinf$.

Requirement removal distributes over pointwise infimum:
\begin{lemma}\label{lem:rmioin:infm:distr}
  For all $\aiosq,\biosq\in\prodiosq$, it holds that:\:
  \(
    \iosqrmioin{\iosqinf{\aiosq}{\biosq}} 
    = \iosqinf{\iosqrmioin{\aiosq}}{\iosqrmioin{\biosq}}
  \).
\end{lemma}
\begin{proof}
  Check that 
  \(
    \{
      \pair{
        \iosqrmioin{\iosqinf{\aiosq}{\biosq}}}{
        \iosqinf{\iosqrmioin{\aiosq}}{\iosqrmioin{\biosq}}
      }
      \where
      \aiosq,\biosq\in\iosq
    \}
    \cup
    \{ 
      \pair{\aiosq}{\aiosq} \where \aiosq\in\iosq 
    \}
  \)
  is a bisimulation.
  \qed
\end{proof}

\begin{lemma}\label{lem:infm:idemp:comm:assoc}
  Infimum is idempotent, commutative, and associative.
  \end{lemma}
\begin{proof}
  By coinduction, with the use of Lem.~\ref{lem:rmioin:infm:distr} in case of associativity.
  \qed
\end{proof}

In a composition requirements come from the second component:
\begin{lemma}\label{lem:iosq:rmioin:cmp}
  For all $\aiosq,\biosq\in\prodiosq$, it holds that:\:
  \(
    \iosqrmioin{\iosqcmp{\aiosq}{\biosq}}
    =
    \iosqcmp{\aiosq}{\iosqrmioin{\biosq}}
  \).
\end{lemma}

Composition distributes both left and right over pointwise infimum:
\begin{lemma}\label{lem:iosq:cmp:infm:distr:left}
  For all $\aiosq,\biosq,\ciosq\in\prodiosq$, it holds that:\:
    \(
      \iosqcmp{\aiosq}{(\iosqinf{\biosq}{\ciosq})}
      = \iosqinf{\iosqcmp{\aiosq}{\biosq}}{\iosqcmp{\aiosq}{\ciosq}}
    \).
\end{lemma}
\begin{proof}
  By coinduction; let 
  \(
    L =
    \{
      \pair{
        \iosqcmp{\aiosq}{(\iosqinf{\biosq}{\ciosq})}}{
        \iosqinf{\iosqcmp{\aiosq}{\biosq}}{\iosqcmp{\aiosq}{\ciosq}}
      }
      \where
      \aiosq,\biosq,\ciosq \in \prodiosq
    \}
  \).
  To show that $L$ is a bisimulation, we prove that,
  for all $\aiosq,\biosq,\ciosq,\diosq_1,\diosq_2\in\prodiosq$,
  if $\diosq_1 = \iosqcmp{\aiosq}{(\iosqinf{\biosq}{\ciosq})}$
  and $\diosq_2 = \iosqinf{\iosqcmp{\aiosq}{\biosq}}{\iosqcmp{\aiosq}{\ciosq}}$,
  then either $\diosq_1 = \iosqemp = \diosq_2$,
  or $\seqhd{\diosq_1} = \seqhd{\diosq_2}$ and
  $\pair{\seqtl{\diosq_1}}{\seqtl{\diosq_2}}\in L$,
  by induction on the number $n\in\nat$ of leading $\ioin$'s of $\aiosq$.

  If $n=0$ and $\aiosq=\iosqemp$, then $\diosq_1=\iosqemp=\diosq_2$.
  
  If $n=0$ and $\aiosq=\iosqcns{\ioout}{\aiosq'}$, then:
  $\diosq_1 = \iosqcns{\ioout}{(\iosqcmp{\aiosq'}{(\iosqinf{\biosq}{\ciosq})})}$, 
  and
  $\diosq_2 = \iosqcns{\ioout}{(\iosqinf{\iosqcmp{\aiosq'}{\biosq}}{\iosqcmp{\aiosq'}{\ciosq}})}$,
  and so $\pair{\seqtl{\diosq_1}}{\seqtl{\diosq_2}}\in L$.
  
  If $n>0$ and $\aiosq = \iosqcns{\ioin}{\aiosq'}$,
  we proceed by case analysis of $\biosq$ and $\ciosq$.
  
  If one of $\biosq$, $\ciosq$ is empty, then $\diosq_1 =\iosqemp =\diosq_2$.

  If $\biosq = \iosqcns{\ioout}{\biosq'}$
  and $\ciosq = \iosqcns{\ioout}{\ciosq'}$,
  then 
  $\diosq_1 = \iosqcmp{\aiosq'}{(\iosqinf{\biosq'}{\ciosq'})}$
  and
  $\diosq_2 = \iosqinf{\iosqcmp{\aiosq}{\biosq}}{\iosqcmp{\aiosq}{\ciosq}}$,
  and we conclude by the induction hypothesis.

  If $\biosq = \iosqcns{\ioout}{\biosq'}$
  and $\ciosq = \iosqcns{\ioin}{\ciosq'}$,
  then
  \( 
    \diosq_1 
    = \iosqcns{\ioin}{(
        \iosqcmp{\aiosq}{
          (\iosqinf{\iosqrmioin{\biosq}}{\ciosq'})
        }
      )}
  \)
  and
  \(
    \diosq_2 
    = \iosqcns{\ioin}{(
        \iosqinf{
          \iosqrmioin{\iosqcmp{\aiosq}{\biosq}}}{
          \iosqcmp{\aiosq}{\ciosq'}
        }
      )}
    = \iosqcns{\ioin}{(
        \iosqinf{
          \iosqcmp{\aiosq}{\iosqrmioin{\biosq}}}{
          \iosqcmp{\aiosq}{\ciosq'}
        }
      )}
  \)
  by Lem.~\ref{lem:iosq:rmioin:cmp}.
  Thus $\pair{\seqtl{\diosq_1}}{\seqtl{\diosq_2}}\in L$.

  The case $\biosq = \iosqcns{\ioin}{\biosq'}$,
  $\ciosq = \iosqcns{\ioout}{\ciosq'}$
  is proved similarly.

  Finally,
  if $\biosq = \iosqcns{\ioin}{\biosq'}$
  and $\ciosq = \iosqcns{\ioin}{\ciosq'}$,
  we compute 
  $\diosq_1 = \iosqcns{\ioin}{(\iosqcmp{\aiosq}{(\iosqinf{\biosq'}{\ciosq'})})}$
  and $\diosq_2 = \iosqcns{\ioin}{(\iosqinf{\iosqcmp{\aiosq}{\biosq'}}{\iosqcmp{\aiosq}{\ciosq'}})}$,
  and conclude $\pair{\seqtl{\diosq_1}}{\seqtl{\diosq_2}}\in L$.
  \qed
\end{proof}
  
\begin{lemma}\label{lem:iosq:cmp:infm:distr:right}
  For all $\aiosq,\biosq,\ciosq\in\prodiosq$, it holds that:\:
    \(
      \iosqcmp{(\iosqinf{\biosq}{\ciosq})}{\aiosq}
      = \iosqinf{\iosqcmp{\biosq}{\aiosq}}{\iosqcmp{\ciosq}{\aiosq}}
    \).
\end{lemma}
\begin{proof}
  Analogous to the proof of Lem.~\ref{lem:iosq:cmp:infm:distr:left}.
  \qed
\end{proof}

The operation of pointwise infimum of (periodically) increasing functions
is defined on their \ioseq\ (\ioterm) representations.
\begin{proposition}\label{prop:iosq:inf:interpret}
  For all increasing functions $f,g\funin\conat\to\conat$:\:
  \(
    \siosqprd{\iosqinf{\trmrep{f}}{\trmrep{g}}} 
    = \funinf{\siosqprd{\trmrep{f}}}{\siosqprd{\trmrep{g}}}
  \).
\end{proposition}
\begin{proof}
  Immediate from Rem.~\ref{rem:iosq:prd:cmp} and Lem.~\ref{lem:iosq:cmp:infm:distr:right}.
  \qed
\end{proof}

We give a coinductive definition of the calculation of the least fixed point of an \ioseq.
The operation $\siosqrmioin$ was defined in Def.~\ref{def:iosq:inf}.
\begin{definition}\normalfont\label{def:iosq:fix}
  The \emph{operation $\siosqfix\funin{\iosq\to\conat}$ 
  computing the least fixed point of a sequence $\aiosq\in\iosq$}, 
  is defined, for all $\aiosq\in\iosq$, by:
  \begin{align*}
  \iosqfix{\iosqemp} & = 0
  \\
  \iosqfix{\iosqcns{\ioout}{\aiosq}} & = 1 + \iosqfix{\iosqrmioin{\aiosq}}
  \\
  \iosqfix{\iosqcns{\ioin}{\aiosq}} & = 0
  \end{align*}
\end{definition}

Removal of a requirement and feeding an input have equal effect: 
\begin{lemma}\label{lem:iosq:cmp:rmioin}
  For all $\aiosq,\biosq\in\prodiosq$, it holds that:\:
  \(
    \iosqcmp{\iosqrmioin{\aiosq}}{\biosq}
    = \iosqcmp{\aiosq}{(\iosqcns{\ioout}{\biosq})}
  \).
\end{lemma}
\begin{proof}
  We show that 
  \(
    R = 
    \{ 
      \pair{
        \iosqcmp{\iosqrmioin{\aiosq}}{\biosq}}{
        \iosqcmp{\aiosq}{(\iosqcns{\ioout}{\biosq})}
      }   
    \}
    \cup
    \{
      \pair{\aiosq}{\aiosq}
    \}
  \) 
  is a bisimulation, by case analysis of $\aiosq$.
  For $\aiosq = \iosqemp$, and $\aiosq = \iosqcns{\ioin}{\aiosq'}$,
  \(
    \pair{
      \iosqcmp{\iosqrmioin{\aiosq}}{\biosq}}{
      \iosqcmp{\aiosq}{(\iosqcns{\ioout}{\biosq})}
    }
    \in R
  \) follows by reflexivity. 
  If $\aiosq = \iosqcns{\ioout}{\aiosq'}$, 
  then $\iosqcmp{\iosqrmioin{\aiosq}}{\biosq} =
  \iosqcns{\ioout}{(\iosqcmp{\iosqrmioin{\aiosq'}}{\biosq})}$,
  and $\iosqcmp{\aiosq}{(\iosqcns{\ioout}{\biosq})} =
  \iosqcns{\ioout}{(\iosqcmp{\aiosq'}{(\iosqcns{\ioout}{\biosq})})}$.
  Hence, 
  \(
    \pair{
      \iosqtl{\iosqcmp{\iosqrmioin{\aiosq}}{\biosq}}
    }{
      \iosqtl{\iosqcmp{\aiosq}{(\iosqcns{\ioout}{\biosq})}}
    }\in R
  \).
  \qed
\end{proof}

The following proposition states that $\iosqfix{\aiosq}$ 
is a \fixp{} of $\siosqprd{\aiosq}$:
\begin{proposition}\label{prop:iosq:fix}
  For all $\aiosq\in\prodiosq$, it holds that:\:
  $\iosqprd{\aiosq}{\iosqfix{\aiosq}} = \iosqfix{\aiosq}$.
\end{proposition}
\begin{proof}
  We prove 
  \(
    \iosqfix{\aiosq}
    = \iosqcmp{\aiosq}{\iosqfix{\aiosq}}
  \)
  by case analysis and coinduction.

  If $\aiosq = \iosqemp$, then 
  $\iosqfix{\aiosq}  = \iosqemp = \iosqcmp{\aiosq}{\iosqfix{\aiosq}}$.
  
  If $\aiosq = \iosqcns{\ioin}{\aiosq'}$, then
  $\iosqfix{\aiosq} = \iosqemp$, and
  \(
    \iosqcmp{\aiosq}{\iosqfix{\aiosq}} 
    = \iosqcmp{(\iosqcns{\ioin}{\aiosq'})}{\iosqemp} 
    = \iosqemp
  \).
  
  If $\aiosq = \iosqcns{\ioout}{\aiosq'}$, then
  $\iosqfix{\aiosq} = \iosqcns{\ioout}{\iosqfix{\iosqrmioin{\aiosq'}}}$
  and
  \(
    \iosqcmp{\aiosq}{\iosqfix{\aiosq}} 
    = \iosqcns{\ioout}{
        (\iosqcmp{\aiosq'}{
        (\iosqcns{\ioout}{\iosqfix{\iosqrmioin{\aiosq'}}})})}
  \),
  and we have to prove that 
  \(
    \iosqfix{\iosqrmioin{\aiosq'}}
    = \iosqcmp{\aiosq'}{(\iosqcns{\ioout}{\iosqfix{\iosqrmioin{\aiosq'}}})}
  \)
  which follows by an instance of Lem.~\ref{lem:iosq:cmp:rmioin}:
  \(
    \iosqcmp{\aiosq'}{(\iosqcns{\ioout}{\iosqfix{\iosqrmioin{\aiosq'}}})}
    = \iosqcmp{\iosqrmioin{\aiosq'}}{\iosqfix{\iosqrmioin{\aiosq'}}}
  \).
  \qed
\end{proof}

\begin{lemma}\label{lem:iosq:fix:prd}
  For all $\aiosq\in\iosq$, it holds that:\: 
  $\lfp{\siosqprd{\aiosq}} = \iosqfix{\aiosq}$.
\end{lemma}
\begin{proof}
\end{proof}

\subsection{Computing Production}

We introduce a term syntax for the production calculus
and rewrite rules for evaluating closed terms.
\begin{definition}\normalfont\label{def:net}
  Let $\nam$ be a set of recursion variables.
  The set of \emph{production terms} $\net$ is generated by:
  \[
    \anet
    \BNFis
    \netsrc{k}
    \BNFor
    \netvar{\avar}
    \BNFor
    \netpeb{\anet}
    \BNFor
    \netbox{\trmrep{f}}{\anet} %
    \BNFor
    \netrec{\anam}{\anet}
    \BNFor
    \netmeet{\anet}{\anet}
  \]
  where $\anam\in\nam$, 
  for $k\in\conat$, the symbol $\netsrc{k}$ is a \emph{numeral}
  (a term representation) for $k$, and, for a unary \pein~function 
  $f\funin\conat\to\conat$, $\trmrep{f}\in\iotrm$, 
  the \ioterm{} representing $f$.
  For every finite set $P = \{ \anet_1,\ldots,\anet_n \} \subseteq \net$, 
  we use $\netmeetn{n}{\anet_1,\ldots,\anet_n}$ and $\snetmeet\; P$
  as shorthands for the production term
  $\netmeet{\anet_1}{\netmeet{\anet_2}{\ldots,\netmeet{\anet_{n-1}}{\anet_n}}}$.
\end{definition}

The `production' $\netprd{\anet} \in \conat$
of a closed production term $\anet \in \net$ is defined by 
induction on the term structure,
interpreting $\snetrec$ as the least fixed point operator,
$\trmrep{f}$ as $f$,
$\netsrc{k}$ as $k$, and $\snetmeet$ as $\min$, 
as follows.
\begin{definition}\normalfont\label{def:netprd}
  The \emph{production} $\netprdr{\anet}{\aenv}\in\conat$ of a term $\anet\in\net$
  with respect to an \emph{assignment} $\aenv\funin{\nam\to\conat}$
  is defined inductively by:
  \begin{align*}
    \netprdr{\netsrc{k}}{\aenv} & = k
    &
    \netprdr{\netbox{\aiotrm}{\anet}}{\aenv} & = \iosqprd{\aiotrm}{\netprdr{\anet}{\aenv}}
    \\
    \netprdr{\netvar{\anam}}{\aenv} & = \funap{\aenv}{\anam}
    &
    \netprdr{\netrec{\anam}{\anet}}{\aenv} & =
    \lfp{\mylam{n}{\netprdr{\anet}{\envupd{\aenv}{\anam}{n}}}}
    \\
    \netprdr{\netpeb{\anet}}{\aenv} & = 1 + \netprdr{\anet}{\aenv}
    &
    \netprdr{\netmeet{\ianet{1}}{\ianet{2}}}{\aenv}
    & = \min(\netprdr{\ianet{1}}{\aenv},\netprdr{\ianet{2}}{\aenv})
  \end{align*}
  where %
  $\envupd{\aenv}{\anam}{n}$ denotes an \emph{update} of $\aenv$,
  defined by $\funap{\envupd{\aenv}{\anam}{n}}{\bnam} = n$
  if $\bnam = \anam$,
  and $\funap{\envupd{\aenv}{x}{n}}{y} = \funap{\aenv}{y}$ otherwise.
  
  Finally, we let $\netprd{\anet}\defdby\netprdr{\anet}{\aenv_0}$ with
  $\aenv_0$ defined by $\funap{\aenv_0}{\anam} = 0$ for all $\anam\in\nam$.
\end{definition}
As becomes clear from Def.~\ref{def:netprd}, 
we could have done without the clause $\netpeb{\anet}$
in the BNF grammar for production terms in Def.~\ref{def:net},
as $\snetpeb$ can be abbreviated to $\ios{+}{-+}$,
an \ioterm\ that denotes the successor function.
However, we take it up as a primitive constructor here 
in order to match with pebbleflow nets,
see Sec.~\ref{ssec:pebbleflow}.

For faithfully modelling the \daob{} lower bounds of stream functions with stream arity $r$,
we employ $r$-ary \pein{} functions, which we represent by `$r$-ary gates'.
\begin{definition}\normalfont\label{def:gate}
  An \emph{$r$-ary gate $\egate{k}{\iaiosq{1},\ldots,\iaiosq{r}}$}
  is defined as a production term context of the form:
  \[
    \egate{k}{\iaiosq{1},\ldots,\iaiosq{r}}
    \defdby
    \netmeetn{r+1}{
      \netsrc{k},
      \netbox{\iaiosq{1}}{\porti{1}},
      \ldots,
      \netbox{\iaiosq{r}}{\porti{r}}
    }
    \punc,
  \]
  where $k\in\conat$ and $\iaiosq{1},\ldots,\iaiosq{r}\in\iosq$.
  We use $\agate$ as a syntactic variable for gates.
  The \emph{interpretation} of a gate $\agate = \egate{k}{\iaiosq{1},\ldots,\iaiosq{r}}$
  is defined by:
  \[
    \iosqprd{\agate}{n_1,\ldots,n_r} 
    \defdby 
    \min(k,\iosqprd{\iaiosq{1}}{n_1},\ldots,\iosqprd{\iaiosq{r}}{n_r})
    \punc.
  \]
  In case $k=\conattop$, 
  we simplify $\egate{k}{\iaiosq{1},\ldots,\iaiosq{r}}$ to 
  \[
    \gate{\iaiosq{1},\ldots,\iaiosq{r}}
    \defdby
    \netmeetn{r}{
      \netbox{\iaiosq{1}}{\porti{1}},
      \ldots,
      \netbox{\iaiosq{r}}{\porti{r}}
    }
    \punc.
  \]
\end{definition}
It is possible to choose unique gate representations $\trmrep{f}$ 
of p-i functions $f$ that are efficiently computable from other 
gate representations, see Section~\ref{ssec:pebbleflow}.

Owing to the restriction to (term representations of)
periodically increasing functions in Def.~\ref{def:netprd} 
it is possible to calculate the production $\netprd{\anet}$ of terms $\anet\in\net$.
For that purpose, we define a rewrite system 
which reduces any closed term to a numeral $\trmrep{k}$.
This system makes use of the computable operations 
$\siotrmcmp$ and $\siosqfix$ on \ioterm{s} mentioned above.
\begin{definition}\normalfont
  \label{def:coltrs}
  \setcounter{equation}{0}
  \renewcommand{\theequation}{$\colrul${\arabic{equation}}}
  The \emph{rewrite relation $\scolred$ on production terms}
  is defined as the compatible closure of the following rules:
  \begin{align}
    \netpeb{\anet} &\red \netbox{\ios{+}{-+}}{\anet}
    \label{coltrs:peb}
    \\
    \netbox{\iaiotrm{1}}{\netbox{\iaiotrm{2}}{\anet}}
    & \red \netbox{\iotrmcmp{\iaiotrm{1}}{\iaiotrm{2}}}{\anet}
    \label{coltrs:boxbox}
    \\
    \netbox{\aiotrm}{\netmeet{\ianet{1}}{\ianet{2}}}
    & \red \netmeet{\netbox{\aiotrm}{\ianet{1}}}{\netbox{\aiotrm}{\ianet{2}}}
    \label{coltrs:boxmeet}
    \\
    \netrec{\anam}{\netmeet{\ianet{1}}{\ianet{2}}}
    & \red \netmeet{\netrec{\anam}{\ianet{1}}}{\netrec{\anam}{\ianet{2}}}
    \label{coltrs:recmeet}
    \\
    \netrec{\anam}{\anet}
    & \red \anet
    \dquad \text{if $\anam\not\in\netfv{\anet}$}
    \label{coltrs:recrm}
    \\
    \netrec{\anam}{\netbox{\aiotrm}{\anam}}
    & \red \netsrc{\iosqfix{\aiotrm}}
    \label{coltrs:recbox}
    \\
    \netmeet{\netsrc{k_1}}{\netsrc{k_2}}
    & \red \netsrc{\comin{k_1,k_2}}  
    \label{coltrs:meetsrc}
    \\
    \netbox{\aiotrm}{\netsrc{k}}
    &\red
    \netsrc{\iosqprd{\aiotrm}{k}}
    \label{coltrs:boxsrc}
    \\
    \netrec{\anam}{\netvar{\anam}}
    & \red \netsrc{0}
    \label{coltrs:recvar}
  \end{align}
  \setcounter{equation}{0}
  \renewcommand{\theequation}{\arabic{equation}}
\end{definition}
The following two lemmas establish the usefulness of 
the rewrite relation $\scolred$.
In order to compute the production 
$\netprd{\anet}$ of a production term $\anet$ it suffices 
to obtain a $\scolred$-normal form of $\anet$. 

\begin{lemma}
  \label{lem:collapse:pop}\label{lem:coll:nets}
  The rewrite relation $\scolred$ is production preserving:
  \[\anet \colred \anet' \implies \netprd{\anet} = \netprd{\anet'}\punc.\]
\end{lemma}
\begin{proof}
  It suffices to prove:
  \(
    \cxtap{\sacxt}{\ell^{\asub}} \colred \cxtap{\sacxt}{r^{\asub}}
    \implies
    \myall{\aenv}{
      \netprdr{\cxtap{\sacxt}{\ell^{\asub}}}{\aenv} 
      = \netprdr{\cxtap{\sacxt}{r^{\asub}}}{\aenv}
    }
    \punc,
  \)
  where $\ell\red r$ is a rule of the TRS given in Def.~\ref{def:coltrs},
  and $\sacxt$ a unary context over $\net$. We proceed
  by induction on $\sacxt$. %
  For the base case, $\sacxt = \ahole$, 
  we give the essential proof steps only: %
  For rule~\eqref{coltrs:peb}, observe that $\siosqprd{\ios{}{-+}}$ 
  is the identity function on $\conat$.
  For rule~\eqref{coltrs:boxbox}, we apply Prop.~\ref{prop:iosq:cmp:interpret}.
  For rule~\eqref{coltrs:boxmeet} the desired equality follows from 
  $C_1$ on page~\pageref{page:mono:props}. 
  For rule~\eqref{coltrs:recmeet} we conclude by $C_2$ ibid.
  For rule~\eqref{coltrs:recbox} we use Lem.~\ref{lem:iosq:fix:prd}.
  For the remaining rules %
  the statement trivially holds.
  For the induction step, the statement easily follows from the induction hypotheses.
  \qed
\end{proof}

\begin{lemma}\label{lem:collapse:sn:cr}
  The rewrite relation $\scolred$ is terminating and confluent,
  and every closed $\anet \in \net$ has a numeral 
  $\netsrc{k}$ as its unique $\scolred$-normal form.
\end{lemma}
\begin{proof}
  \newcommand{\sweight}{\mrm{w}}
  \newcommand{\weight}{\funap{\sweight}} %
  To see that $\scolred$ is terminating, 
  let $\sweight\funin\net\to\nat$ be defined by:
  \begin{align*}
    \weight{\netvar{\anam}}&=1
    &\weight{\netpeb\anet}&=2\cdot\weight\anet+1
    &\weight{\netrec\anam\anet}&=2\cdot\weight\anet
    \\
    \weight{\netsrc{k}}&=1
    &\weight{\netbox\aiosq\anet}&=2\cdot\weight{\anet}
    &\weight{\netmeet{\ianet{1}}{\ianet{2}}}&=\weight{\ianet{1}}+\weight{\ianet{2}}+1
    \punc,
  \end{align*}
  and observe that $\anet\colred\bnet$ implies $\weight{\anet} > \weight{\bnet}$.
    
  Some of the rules of $\scolred$ overlap; e.g.\ rule~\eqref{coltrs:boxbox} with itself.
  For each of the five critical pairs we can find a common reduct  
  (the critical pair 
   $\pair{\iosqcmp{\aiosq}{(\iosqcmp{\biosq}{\ciosq})}}{\iosqcmp{(\iosqcmp{\aiosq}{\biosq})}{\ciosq}}$
   due to an \eqref{coltrs:boxbox}/\eqref{coltrs:boxbox}-overlap
   can be joined by Lem.~\ref{lem:iosq:cmp:assoc}),
  and hence $\scolred$ is locally confluent,
  by the Critical Pairs Lemma (cf.~\cite{terese:2003}).
  By Newman's Lemma, we obtain that $\scolred$ is confluent.
  Thus normal forms are unique.

  To show that every closed net normalises to a source, %
  let $\anet$ be an arbitrary normal form.
  Note that the set of free variables of a net is closed under $\scolred$,
  and hence $\anet$ is a closed net.
  Clearly, $\anet$ does not contain pebbles,
  otherwise \eqref{coltrs:peb} would be applicable. 
  To see that $\anet$ contains no subterms of the form $\netrec{\anam}{\bnet}$,
  suppose it does and consider the innermost such subterm, viz.\ $\bnet$ contains no $\snetrec$.
    If $\bnet\equiv\netsrc{k}$ or $\bnet\equiv\netvar{\anam}$, then \eqref{coltrs:recrm}, 
    resp.\ \eqref{coltrs:recvar} is applicable.
    If $\bnet\equiv\netbox{\aiosq}{\bnet'}$, we further distinguish four cases:
      if $\bnet'\equiv\netsrc{k}$ or $\bnet'\equiv\netvar{\anam}$,
      then \eqref{coltrs:boxsrc} resp.\ \eqref{coltrs:recbox} is applicable;
      if the root symbol of $\bnet'$ is one of $\snetbox,\snetmeet$,
      then $\bnet$ constitutes a redex w.r.t.\ \eqref{coltrs:boxbox}, \eqref{coltrs:boxmeet},
      respectively.
    If $\bnet\equiv{\netmeet{\ibnet{1}}{\ibnet{2}}}$, 
    we have a redex w.r.t.\ \eqref{coltrs:recmeet}.
  Thus, there are no subterms $\netrec{\anam}{\bnet}$ in $\anet$,
  and therefore, because $\anet$ is closed, also no variables $\netvar{x}$.
  To see that $\anet$ has no subterms of the form $\netbox{\aiosq}{\bnet}$,
  suppose it does and consider the innermost such subterm.
  Then, if $\bnet\equiv\netsrc{k}$ or $\bnet\equiv\netmeet{\ibnet{1}}{\ibnet{2}}$
  then \eqref{coltrs:boxsrc} resp.\ \eqref{coltrs:boxmeet} is applicable;
  other cases have been excluded above.
  Finally, $\anet$ does not contain subterms of the form $\netmeet{\ianet{1}}{\ianet{2}}$.
  For if it does, consider the innermost occurrence and note that, 
  since the other cases have been excluded already, 
  $\ianet{1}$ and $\ianet{2}$ have to be sources, 
  and so we have a redex w.r.t.\ \eqref{coltrs:meetsrc}.
  We conclude that $\anet\equiv\netsrc{k}$ for some $k\in\conat$.
  \qed
\end{proof}

\subsection{Pebbleflow Nets}\label{ssec:pebbleflow}
Production terms can be visualized by `pebbleflow nets'
introduced in~\cite{endr:grab:hend:isih:klop:2007},
and serve as a means to model the `data-oblivious'
consumption/production behaviour of stream specifications.
The idea is to abstract from the actual stream elements (data) in a stream term 
in favour of occurrences of the symbol $\pebble$, which we call `pebble'.
Thus, a stream term $\strcns{\adattrm}{\astrtrm}$ is translated to
$\trnsl{\strcns{\adattrm}{\astrtrm}} = \netpeb{\trnsl{\astrtrm}}$, 
see Section~\ref{sec:translation}.

We give an operational description of pebbleflow nets,
and define the production of a net as the number of pebbles 
a net is able to produce at its output port.
Then we prove that this definition of production 
coincides with Def.~\ref{def:netprd}.

Pebbleflow nets are networks built of pebble processing units
(fans, boxes, meets, sources) connected by wires.
We use the term syntax given in Def.~\ref{def:net} for nets 
and the rules governing the flow of pebbles through a net, 
and then give an operational meaning of the units a net is built of.

\begin{definition}\normalfont
  \label{def:pebbleflow}
  \setcounter{equation}{0}
  \renewcommand{\theequation}{$\pebrul${\arabic{equation}}}
  The \emph{pebbleflow rewrite relation $\spebred$}
  is defined as the compatible closure
  of the union of the following rules:
  \begin{align}
    \netmeet{\netpeb{\anet_1}}{\netpeb{\anet_2}}
    &\red \netpeb{\netmeet{\anet_1}{\anet_2}}
    \label{def:pebbleflow:rule:meet}
    \\
    \netrec{\anam}{\netpeb{\anet(\avar)}}
    &\red
    \netpeb{\netrec{\anam}{\anet(\netpeb{\netvar{\anam}})}}
    \label{def:pebbleflow:rule:feedback}
    \\
    \netbox{\iosqcns{\ioout}{\aiosq}}{\anet}
    &\red \netpeb{\netbox{\aiosq}{\anet}}
    \label{def:pebbleflow:rule:boxP}
    \\
    \netbox{\iosqcns{\ioin}{\aiosq}}{\netpeb{\anet}}
    &\red \netbox{\aiosq}{\anet}
    \label{def:pebbleflow:rule:boxM}
    \\
    \netsrc{1+k}
    &\red
    \netpeb{\netsrc{k}}
    \label{def:pebbleflow:rule:source}
  \end{align}
  \setcounter{equation}{0}
  \renewcommand{\theequation}{\arabic{equation}}
\end{definition}

Wires are unidirectional FIFO communication channels.
They are idealised in the sense that there is no upper bound on
the number of pebbles they can store; arbitrarily long queues are allowed.
Wires have no counterpart on the term level;
in this sense they are akin to the edges of a term tree.
Wires connect \emph{boxes}, \emph{meets}, \emph{fans},
and \emph{sources}, that we describe next.

A \emph{meet} is waiting for a pebble at each of its input ports
and only then produces one pebble at its output port, see Fig.~\ref{fig:meet}.
Put differently, the number of pebbles a meet produces equals the
minimum of the numbers of pebbles available at each of its input ports.
Meets enable explicit branching; they are used to model
stream functions of arity $>1$, as will be explained below.
A meet with an arbitrary number $n \geq 1$ of input ports
is implemented by using a single wire in case $n = 1$,
and if $n = k + 1$ with $k\geq 1$, by connecting the output port of
a `$k$-ary meet' to one of the input ports of a (binary) meet.
\begin{figure}[htb]
  \begin{minipage}[b]{.49\textwidth}
  \newcommand{\aaa}{$\snetmeetn{r}$}
  \newcommand{\bbb}{$\ianet{1}$}
  \newcommand{\ccc}{$\ianet{2}$}
  \begin{center}
  \fbox{\scalebox{.6}{\input{\figpath/meet.pstex_t}}}
  \end{center}
  \mycaption{
    Rule~\eqref{def:pebbleflow:rule:meet}.
  }
  \label{fig:meet}
  \end{minipage}
  \hfill
  \begin{minipage}[b]{.49\textwidth}
  \begin{center}
  \fbox{\scalebox{.6}{\input{\figpath/recfan.pstex_t}}}
  \end{center}
  \mycaption{
    Rule~\eqref{def:pebbleflow:rule:feedback}.
  }
  \label{fig:recfan}
  \end{minipage}
\end{figure}

The behaviour of a \emph{fan} is dual to that of a meet:
a pebble at its input port is duplicated along its output ports.
A fan can be seen as an explicit sharing device,
and thus enables the construction of cyclic nets.
More specifically, we use fans only to implement feedback when drawing nets;
there is no explicit term representation for the fan in our term calculus.
In Fig.~\ref{fig:recfan} a pebble is sent over the output wire of the net
and at the same time is fed back to the `recursion wire(s)'.
Turning a cyclic net into a term (tree) means to introduce
a notion of binding; certain nodes need to be labelled by a name ($\mu x$)
so that a wire pointing to that node is replaced by a name ($x$)
\emph{referring} to the labelled node.
In rule~\eqref{def:pebbleflow:rule:feedback} feedback is accomplished
by substituting $\netpeb{\avar}$ for all free occurrences $\avar$ of $\anet$.

A \emph{source} has an output port only,
contains a number $k\in\conat$ of pebbles,
and can fire if $k > 0$, see Fig.~\ref{fig:source}.
The rewrite relation $\scolred$ given in Def.~\ref{def:coltrs}
collapses any closed net (without input ports that is) into a source.

A \emph{box} consumes pebbles at its input port and produces pebbles
at its output port, controlled by an \ioseq{}
$\aiosq\in\iosq$ associated with the box.
For example, consider the unary stream function
$\strff{dup}$, defined as follows, and its corresponding \ioseq:
\[
  \funap{\strff{dup}}{\strcns{x}{\astr}}
  = \strcns{x}{\strcns{x}{\funap{\strff{dup}}{\astr}}}
  \quad\quad\quad\quad
  \ios{}{-++}
\]
which is to be thought of as: \textit{for $\strff{dup}$ to produce two outputs,
it first has to consume one input, and this process repeats indefinitely}.
Intuitively, the symbol~$\ioin$ represents a requirement for an input pebble,
and $\ioout$ represents a ready state for an output pebble.
Pebbleflow through boxes is visualised in
Figs.~\ref{fig:boxP} and~\ref{fig:boxM}.
\begin{figure}[htb]
  \begin{minipage}[b]{.3\textwidth}
  \newcommand{\aaa}{\iosqcns{\ioout}{\aiosq}}%
  \newcommand{\bbb}{\aiosq}%
    \begin{center}
    \fbox{\scalebox{.6}{\input{\figpath/box1.pstex_t}}}
    \end{center}
    \mycaption{Rule~\eqref{def:pebbleflow:rule:boxP}.}
    \label{fig:boxP}
  \end{minipage}
  \hfill
  \begin{minipage}[b]{.3\textwidth}
  \newcommand{\aaa}{\iosqcns{\ioin}{\aiosq}}%
  \newcommand{\bbb}{\aiosq}%
    \begin{center}
    \fbox{\scalebox{.6}{\input{\figpath/box2.pstex_t}}}
    \end{center}
    \mycaption{Rule~\eqref{def:pebbleflow:rule:boxM}.}
    \label{fig:boxM}
  \end{minipage}
  \hfill
  \begin{minipage}[b]{.3\textwidth}
  \newcommand{\aaa}{\trmrep{1+k}}
  \newcommand{\bbb}{\trmrep{k}}
    \begin{center}
    \fbox{\scalebox{.6}{\input{\figpath/source.pstex_t}}}
    \end{center}
    \mycaption{Rule~\eqref{def:pebbleflow:rule:source}.}
    \label{fig:source}
  \end{minipage}
\end{figure}

\begin{wrapfigure}[7]{r}{37mm}
  \newcommand{\aaa}{\snetmeetn{n}}
  \newcommand{\bbb}{\iaiosq{1}}
  \newcommand{\ccc}{\iaiosq{r}}
  \vspace{-5ex}
  \hspace{1.5em}
  \fbox{\scalebox{0.6}{\input{\figpath/gate.pstex_t}}}
  \vspace{-1ex}
  \caption{$\netgate{\iaiosq{1},\ldots,\iaiosq{r}}$}%
  \label{fig:gate}
\end{wrapfigure}
The data-oblivious production behaviour of stream functions $\astrfun$ 
with a stream arity $\arityS{\astrfun} = r$ are modelled by $r$-ary gates
(defined in Def.~\ref{def:gate})
that express the contribution of each individual stream argument
to the total production of $\astrfun$, see Fig.~\ref{fig:gate}.
The precise translation of stream functions into gates is given in Sec.~\ref{sec:translation},
in particular in Def.~\ref{def:transl:flat:fn:symbols:infioseqspec}.
\vspace{4ex}

\begin{lemma}\label{thm:pebred:confluent}
  The pebbleflow rewrite relation $\spebred$ is confluent.
\end{lemma}
\begin{proof}
  The rules of $\spebred$ can be viewed as a higher-order rewriting system (HRS)
  that is orthogonal.
  Applying Thm.~11.6.9 in~\cite{terese:2003}
  then establishes the lemma.
  \qed
\end{proof}

\begin{definition}\normalfont\label{def:pfnetprd}
  The \emph{production function $\spfnetprd \funin \net \to \conat$ of nets}
  is defined by:
  \[
    \pfnetprd{\anet} 
    \defdby \sup\,\{ n\in\nat \where \anet \mpebred \netpebn{n}{\anet'}\}
    \punc,
  \]
  for all $\anet\in\net$. $\pfnetprd{\anet}$ is called \emph{the production of $\anet$}.
  Moreover, for a net $\anet$ and an assignment $\aenv\in{\nam\to\conat}$,
  let $\pfnetprdr{\anet}{\aenv} \defdby \pfnetprd{\anet^{\aenv}}$
  where $\anet^\aenv$ denotes the net obtained by
  replacing each free variable $\netvar{\anam}$ of $\anet$ with
  $\netpebn{\funap{\aenv}{\anam}}{\netvar{\anam}}$.
\end{definition}
Note that for closed nets $\anet$ (where $\netfv{\anet} = \setemp$),
$\anet^{\aenv} = \anet$ and therefore
$\netprdr{\anet}{\aenv} = \netprd{\anet}$,
for all assignments $\aenv$.

An important property used in the following lemma
is that functions of the form
$\mylamin{n}{\conat}{\pfnetprdr{\anet}{\envupd{\aenv}{\anam}{n}}}$
are monotonic functions over $\conat$.
Every monotonic function $f\funin\conat\to\conat$
in the complete chain~$\conat$ has a least fixed point $\lfp{f}$
which can be computed by $\lfp{f} = \lim_{n\to\infty}{\funnap{f}{n}{0}}$.\label{page:mono:props}
In what follows we employ, for monotonic $f,g\funin\conat\to\conat$,
two basic properties:
\begin{align}
  \myall{n,m}{
    \big(
      \funap{f}{\min(n,m)} & = \min(\funap{f}{n},\funap{f}{m})
    \big)
  }
  \tag{$C_1$}
  \\
  \lfp{\mylam{n}{\min(\funap{f}{n},\funap{g}{n})}} 
  & = \min(\lfp{f},\lfp{g})
  \tag{$C_2$}
\end{align}

\begin{lemma}\label{lem:netprod:netrec:lfp}
  For all nets $\bnet\in\net$ and all assignments $\aenv$,
  we have that\/
  $\pfnetprdr{\netrec{\anam}{\bnet}}{\aenv}$
  is the least fixed point of
  $\mylamin{n}{\conat}{\pfnetprdr{\bnet}{\envupd{\aenv}{\anam}{n}}}$.
\end{lemma}

\begin{proof}
  Let $\aenv \funin \envtyp$ be an arbitrary assignment 
  and $\bnet_0 \defdby \bnet^{\envupd{\aenv}{x}{0}}$.
  Observe that $\pfnetprdr{\netrec{\anam}{\bnet}}{\aenv} = \pfnetprd{\netrec{x}{\bnet_0}}$
  and consider a rewrite sequence of the form
  \begin{gather*}
    \netrec{x}{\bnet_0} \mred \ldots \mred
    \netpebn{n_i}{\netrec{x}{\bnet_i}}
    \mred \netpebn{n_i}{\netrec{x}{\netpebn{\ell_i}{\bnet'_{i}}}}
    \mred \netpebn{n_i + \ell_i}{\netrec{x}{\bnet_{i+1}}}
    \mred \ldots
  \end{gather*}
  where $\ell_i = \pfnetprd{\bnet_i}$, $n_0 = 0$, $n_{i+1} = n_i + \ell_i$,
  and $\bnet_{i+1} \defdby \bnet'_{i}(\netpebn{\ell_i}{x})$.
  Note that
  $\lim_{m \to \infty}{n_m} = \pfnetprd{\netrec{\anam}{\bnet_0}}$;
  `$\le$' follows from
  $\myall{m}{\netrec{\anam}{\bnet_0} \mred \netpebn{n_m}{\netrec{x}{\bnet_{m}}}}$,
  and
  `$\ge$' since if
  $\lim_{m \to \infty}{n_m} < \infty$
  then $\exists m \in \nat$ such that $\ell_m \defdby \pfnetprd{\bnet_m} = 0$
  and therefore
  $\pfnetprd{\netrec{\anam}{\bnet_0}} = \pfnetprd{\netpebn{n_m}{\netrec{x}{\bnet_{m}}}} = n_m$
  by confluence.

  Let $f_i = \mylam{n}{\netprd{\bnet_i(\netpebn{n}{x})}}$,
  and $f'_i = \mylam{n}{\netprd{\bnet'_{i}(\netpebn{n}{x})}}$.
  We prove
  \begin{equation*}
    \myallin{k}{\nat}{\funap{f_0}{n_m + k} = n_m + \funap{f_m}{k}} \tag{$*$}\label{eq:lfp}
  \end{equation*}
  by induction over $m$.
  The base case $m = 0$ is trivial, we consider the induction step.
  We have
  $\bnet_m \mred \netpebn{\ell_m}{\bnet'_{m}}$
  and by substituting $\netpebn{k}{x}$ for $x$ we get
  \begin{equation*}
  \myallin{k}{\nat}{\funap{f_m}{k} = \ell_m + \funap{f'_m}{k}} \tag{$**$}\label{eq:lfp2}
  \end{equation*}
  Moreover, since $\funap{f_{m+1}}{k} = \funap{f'_{m}}{\ell_m+k}$,
  we get
  $n_{m+1} + \funap{f_{m+1}}{k}
   = n_{m+1} + \funap{f'_{m}}{\ell_m+k} %
   = n_{m} + \ell_m + \funap{f'_{m}}{\ell_m+k}
   \stackrel{\eqref{eq:lfp2}}{=} n_{m} + \funap{f_{m}}{\ell_m+k}
   \stackrel{\eqref{eq:lfp}}{=} \funap{f_0}{n_m + \ell_m +k}
   = \funap{f_0}{n_{m+1} +k}$.

  Let $f \defdby f_0$. We proceed with showing
  $\myall{m}{\funnap{f}{m}{0} = n_m}$
  by induction over $m \in \nat$.
  For the base case $m = 0$ we have $\funnap{f}{0}{0} = 0$ and $n_0 = 0$,
  and for the induction step
  we get
  $\funnap{f}{m+1}{0}
   = \funap{f}{\funnap{f}{m}{0}}
   \stackrel{\text{IH}}{=} \funap{f}{n_m}
   \stackrel{\eqref{eq:lfp}}{=} n_{m} + \funap{f_m}{0}
   = n_{m} + \ell_m
   = n_{m+1}$.

  Hence
  $\lfp{f} = \lim_{m \to \infty}{\funnap{f}{m}{0}}
   = \lim_{m \to \infty}{n_m} = \pfnetprd{\netrec{x}{\bnet_0}}
   = \pfnetprdr{\netrec{\anam}{\bnet}}{\aenv}$.
  \qed
\end{proof}

\begin{lemma}\label{lem:netprd:netbox}
  For $\anet\in\net$,  $\aiosq\in\prodiosq$, $\aenv\funin\envtyp$:
  $\pfnetprdr{\netbox{\aiosq}{\anet}}{\aenv} = \iosqprd{\aiosq}{\pfnetprdr{\anet}{\aenv}}$.
\end{lemma}
\begin{proof}
  We show that the relation $R\subseteq\conat\times\conat$ defined as follows is a bisimulation:
  \[
    R \defdby
    \big\{\,
      \pair{\pfnetprdr{\netbox{\aiosq}{\anet}}{\aenv}}{\iosqprd{\aiosq}{\pfnetprdr{\anet}{\aenv}}}
      \where \aiosq\in\prodiosq ,\; \anet\in\net ,\; \aenv\funin\envtyp
    \,\big\}
    \punc,
  \]
  that is, we prove that, for all $k_1,k_2\in\conat$, $\aiosq\in\prodiosq$, $\anet\in\net$,
  and $\aenv\funin\envtyp$, if $k_1 = \pfnetprdr{\netbox{\aiosq}{\anet}}{\aenv}$
  and $k_2 = \iosqprd{\aiosq}{\pfnetprdr{\anet}{\aenv}}$,
  then either $k_1 = k_2 = 0$ or
  $k_1 = 1 + k'_1$, $k_2 = 1 + k'_2$ and $\pair{k'_1}{k'_2}\in R$.
  
  Let $k_1,k_2\in\conat$, $\aiosq\in\prodiosq$, $\anet\in\net$, and $\aenv\funin\envtyp$,
  be such that $k_1 = \pfnetprdr{\netbox{\aiosq}{\anet}}{\aenv}$
  and $k_2 = \iosqprd{\aiosq}{\pfnetprdr{\anet}{\aenv}}$.
  By definition of $\iosq$, we have that
  $\aiosq \equiv \iosqcat{\ioin^n}{\iolstcns{\ioout}{\biosq}}$
  for some $n\in\nat$ and $\biosq\in\iosq$.
  We proceed by induction on $n$.
  If $n = 0$, then $k_1 = 1 + k'_1$
  with $k'_1 = \pfnetprdr{\netbox{\biosq}{\anet}}{\aenv}$
  and $k_2 = 1 + k'_2$ with $k'_2 = \iosqprd{\biosq}{\pfnetprdr{\anet}{\aenv}}$,
  and $\pair{k'_1}{k'_2}\in R$.
  If $n = n' + 1$, we distinguish cases:
  If $\pfnetprdr{\anet}{\aenv} = 0$, then $k_1 = k_2 = 0$.
  If $\pfnetprdr{\anet}{\aenv} = 1 + m$,
  then $\anet\mpebred\netpeb{\bnet}$ for some $\bnet\in\net$ with
  $\pfnetprdr{\bnet}{\aenv} = m$.
  Thus we get $k_1 = \pfnetprdr{\netbox{\iosqcat{\ioin^{n'}}{\iolstcns{\ioout}{\biosq}}}{\bnet}}{\aenv}$
  and $k_2 = \iosqprd{\iosqcat{\ioin^{n'}}{\iolstcns{\ioout}{\biosq}}}{\pfnetprdr{\bnet}{\aenv}}$,
  and $\pair{k_1}{k_2}\in R$ by induction hypothesis.
  \qed
\end{proof}

Next we show that production of a term $\anet\in\net$ (Def.~\ref{def:netprd})
coincides with the maximal number of pebbles produced by the net $\anet$ 
via the rewrite relation $\spebred$ (Def.~\ref{def:pfnetprd}).
\begin{lemma}\label{lem:netprd:pfnetprd}
  For all nets $\anet$ and assignments $\aenv$, %
  it holds\/ $\pfnetprdr{\anet}{\aenv} = \netprdr{\anet}{\aenv}$.
\end{lemma}

\begin{proof}
  The statement of the lemma can be proved by a straightforward induction
  on the number of $\mu$-bindings of a net $\anet$, with a subinduction
  on the size of $\anet$. %
  In the cases $\anet\equiv\netbox{\aiosq}{\anet'}$ and
  $\anet \equiv \netrec{\anam}{\anet'}$
  Lem.~\ref{lem:netprd:netbox} and
  Lem.~\ref{lem:netprod:netrec:lfp} are applied, respectively.
  \qed
\end{proof}
Subsequently, we will use the interpretation of terms in $\net$
and the production of pebbleflow nets, 
$\snetprd$ and $\spfnetprd$, interchangeably.

\section{Translation into Production Terms}\label{sec:translation}
In this section we define a translation
from stream constants in flat
or friendly nesting specifications to production terms.
In particular, the root $\rootsc$ of a specification $\atrs$
is mapped by the translation to a production term $\trnsl{\rootsc}$
with the property that if\/ $\atrs$ is flat (friendly nesting),
then the \daob{} lower bound on the production of $\rootsc$ in $\atrs$
equals (is bounded from below by)
the production of~$\trnsl{\rootsc}$.

\subsection{Translation of Flat and Friendly Nesting Symbols}\label{sec:translation:subsec:functions}

As a first step of the translation,
we describe how for a flat (or friendly nesting)
stream function symbol $\astrfun$ in a stream specification $\atrs$
a periodically increasing function $\sdtrnsl{\astrfun}$
can be calculated
that is (that bounds from below)
the \daob{} lower bound on the production of $\astrfun$
in $\atrs$. %

Let us again consider the rules
(i)~$\funap{\astrfun}{\strcns{\numsuc{x}}{\strcns{y}{\astr}}}
  \red \strcns{\numadd{\numsuc{x}}{y}}{\funap{\astrfun}{\strcns{y}{\astr}}}$,
and
(ii)~$\funap{\astrfun}{\strcns{\numzer}{\astr}}
  \red \strcns{\numzer}{\strcns{\numsuc{\numzer}}{\funap{\astrfun}{\astr}}}$
  from Fig.~\ref{fig:pascal}.
We model the \daob{} lower bound on the production of $\astrfun$
by a function from $\conat$ to $\conat$ %
defined as the unique solution for $\specvar{\astrfun}$ of the following system of equations.\label{transl:Pascal:f}
We disregard what the concrete stream elements are,
and therefore we take the infimum over all possible traces:
\[
  \specvarap{\astrfun}{n}
  = \inf\, \big\{ \specvarap{\astrfun,\text{(i)}}{n}, \specvarap{\astrfun,\text{(ii)}}{n} \big\}
\]
where the solutions for $\specvar{\astrfun,\text{(i)}}$ and $\specvar{\astrfun,\text{(ii)}}$
are the \daob{} lower bounds of $\astrfun$ assuming
that the first rule applied in the rewrite sequence is (i) or (ii), respectively.
The rule~(i) consumes two elements, produces one element
and feeds one element back to the recursive call.
For rule~(ii) these numbers are 1, 2, 0 respectively.
Therefore we get:
\begin{align*}
  \specvarap{\astrfun,\text{(i)}}{n}
  & = \letin{n' \defdby n-2},\;
      \text{if $n' < 0$ then $0$ else $1 + \specvarap{\astrfun}{1 + n'}$}
  \punc,
  \\
  \specvarap{\astrfun,\text{(ii)}}{n}
  & = \letin{n' \defdby n-1},\;
      \text{if $n' < 0$ then $0$ else $2 + \specvarap{\astrfun}{n'}$}
  \punc.
\end{align*}
The unique solution for $\specvar{\astrfun}$ is $n \mapsto \cosubtr{n}{1}$,
represented by the \ioterm{} $\ios{-}{-+}$. \label{trans:pascal:f}

In general, functions may have multiple arguments,
which during rewriting may get permuted, deleted or duplicated.
The idea is to trace single arguments, and to take the infimum
over traces in case an argument is duplicated.

In the definition of the translation of stream functions,
we need to distinguish the cases according to whether a symbol is weakly guarded or not:
On $\Ssmincns$ we define
\(
  {\sugsfsleadsto}
  \defdby
  \{
    \pair{\rootsymb{\ell}}{\rootsymb{r}}
    \where {\ell\to r} \in \Rs,\; \text{`$\sstrcns$'} \neq \rootsymb{r} \in \Ssmincns
  \}
  \punc.
\)
the \emph{dependency relation} between symbols in $\Ssmincns$.
We say that a symbol $\astrfun \in \Ssmincns$ is \emph{weakly guarded}
if $\astrfun$ is strongly normalising with respect to $\sugsfsleadsto$
and \emph{unguarded}, otherwise.
Note that since $\Ssmincns$ is finite it is (easily) decidable whether
a symbol $\astrfun\in\Ssmincns$ is weakly guarded or unguarded.

The translation of a stream function symbol
is defined as the unique solution of a (usually infinite)
system of defining equations where the unknowns are functions.
More precisely, for each symbol $\astrfun \in \Sfnest \supseteq \Ssf$ of a flat
or friendly nesting stream specification,
this system has a \pein{} function
$\sdtrnsl{\astrfun}$ %
as the solution for $\specvar{\astrfun}$,
which is unique among the continuous functions.
Later we will see 
(Prop.~\ref{prop:corr:pein:gate:transl} and Lem.~\ref{lem:transl:soundness})
that the translation $\sdtrnsl{\astrfun}$ of a flat (friendly nesting)
stream function symbol $\astrfun$ coincides with (is a lower bound for)
the \daob~lower bound $\doLow{\atrs}{\astrfun}$ of $\astrfun$. 
\begin{definition}\normalfont
  \label{def:transl:flat:fn:symbols:pifuncspec}
  Let $\atrs = \pair{\Sigma}{R}$ be a stream specification.
  For each flat or friendly nesting symbol
  $\astrfun \in \Sfnest \supseteq \Sflat$ with
  arities $k = \arityS{\astrfun}$ and $\ell = \arityD{\astrfun}$
  we define $\sdtrnsl{\astrfun}\funin\conat^{\,k}\to\conat$,
  called the \emph{(\pein~function) translation of\/ $\astrfun$ in $\atrs$}, 
  as the unique solution $s_{\specvar{\astrfun}}$ for $\specvar{\astrfun}$ 
  of the following system of defining equations,
  where the solution of an equation of the form
  $\specvarap{}{n_1,\ldots,n_k} = \ldots$ is a function 
  $s_\sspecvar\funin\conat^{k}\to\conat$
  (it is to be understood that $s_\sspecvar\in\conat$ if $k=0$):
  For all $n_1,\ldots,n_k\in\conat$, 
  $i \in \{\argnone,1,\ldots,k\}$, and $n\in\nat$:
  \begin{gather*}
    \specvarap{\astrfun}{n_1,\ldots,n_k}
    = \min(
        \specvar{\astrfun,\argnone},
        \specvarap{\astrfun,1}{n_1},
        \ldots,
        \specvarap{\astrfun,k}{n_k}
      )
    \punc{,}
    \\
    \specvar{\astrfun,\argnone}
    = \begin{cases}
        \inf\,
          \big\{
            \specvar{\astrfun,\argnone,\rho}
            \where \text{$\rho$ a defining rule of $\astrfun$}
          \big\}
        & \text{if $\astrfun$ is weakly guarded,} \\[-0.4ex]
        0 & \text{if $\astrfun$ is unguarded,}
      \end{cases}
    \\
    \specvarap{\astrfun,i}{n}
    = \begin{cases}
        \inf\,
          \big\{
            \specvarap{\astrfun,i,\rho}{n}
            \where \text{$\rho$ a defining rule of $\astrfun$}
          \big\}
        & \text{if $\astrfun$ is weakly guarded,} \\[-0.4ex]
        0 & \text{if $\astrfun$ is unguarded.}
      \end{cases}
  \end{gather*}
  We write $\strcns{\vec{\adattrm}_i}{\astr_i}$ for
  $\strcns{\adattrm_{i,1}}{\strcns{\ldots}{\strcns{\adattrm_{i,p}}{\astr_i}}}$,
  and $\lstlength{\vec{\adattrm}_i}$ for $p$.
  For specifying $\specvar{\astrfun,i,\rho}$ we distinguish
  the possible forms the rule $\rho$ can have.
  If $\rho$ is nesting, then
  $\specvar{\astrfun,\argnone,\rho} = \conattop$,
  and 
  $\specvarap{\astrfun,i,\rho}{n} = n$ for all $n \in \conat$.
  Otherwise, $\rho$ is non-nesting and of the form:
  \[
    \funap{\astrfun}{
      (\strcns{\vec{\adattrm}_1}{\astr_1}),
      \ldots,
      (\strcns{\vec{\adattrm}_{k}}{\astr_{k}}),
      \bdattrm_1,\ldots,\bdattrm_{\ell}
    }
    \red
    \strcns{\cdattrm_{1}}{\strcns{\ldots}{\strcns{\cdattrm_{m}}{\astrtrm}}}
    \punc,
  \]
  where either (a)~$\astrtrm \equiv \astr_j$, or
  (b)~\(
        \astrtrm \equiv
        \funap{\bstrfun}{
        (\strcns{\vec{\ddattrm}_{\!1}}{\astr_{\permut{1}}}),
        \ldots,
        (\strcns{\vec{\ddattrm}_{\!k'}}{\astr_{\permut{k'}}}),
        \edattrm_1,\ldots,\edattrm_{\ell'}}
      \)
  with $k' = \arityS{\bstrfun}$, $\ell' = \arityD{\bstrfun}$, and
  $\spermut \funin \{1,\ldots,k'\} \to \{1,\ldots,k\}$.
  Let:
  \begin{align*}
    \specvar{\astrfun,\argnone,\rho} = \
      &
      \begin{cases}
        \conattop &\text{case~(a)} \\
        m + \specvar{\bstrfun,\argnone} &\text{case~(b)}
      \end{cases}
    \\
    \specvarap{\astrfun,i,\rho}{n} = \
      &\letin{n' \defdby n - \lstlength{\vec{\adattrm}_i}}
      \text{, if $n' < 0$ then $0$ else }\\
      &
      m +
      \begin{cases}
        n'
        &\text{case~(a), $i = j$} \\
        \conattop
        &\text{case~(a), $i \neq j$} \\
        \inf\,
          \big\{
            \specvarap{\bstrfun,j}{n' + \lstlength{\vec{\ddattrm}_{\!j}}}
            \where j \in {\funap{\invfun{\spermut}}{i}}
          \big\}
        &\text{case~(b)}
      \end{cases}
  \end{align*}
  where we agree $\inf \setemp = \conattop$.
\end{definition}

We mention a useful intuition for understanding the fabric of 
the formal definition above of the translation $\sdtrnsl{\astrfun}$ 
of a stream function $\astrfun$,
the solution $s_{\specvar{\astrfun}}$ for $\specvar{\astrfun}$
of the system of equations in Def.~\ref{def:transl:flat:fn:symbols:pifuncspec}. 
For each $i\in\{1,\ldots,k\}$ where $k = \arityS{\astrfun}$
the solution $s_{\specvar{\astrfun,i}}$ for $\specvar{\astrfun,i}$
(a unary \pein~function) in this system
describes to what extent consumption from the $i$-th component of $\astrfun$
`delays' the overal production.
Since in a \daob~rewrite sequence from
$\funap{\astrfun}{\strcns{\pebble^{n_1}}{\aiosq_1}, \ldots, \strcns{\pebble^{n_k}}{\aiosq_k}}$
it may happen that all stream variables $\aiosq_1$, \ldots, $\aiosq_k$
are erased at some point, and that the sequence subsequently continues rewriting stream constants,
monitoring the delays caused by individual arguments is not sufficient
alone to define the production function of $\astrfun$.
This is the reason for the use, in the definition of $\sdtrnsl{\astrfun}$,
of the solution $s_{\specvar{\astrfun,\argnone}}$ for the variable
$\specvar{\astrfun,\argnone}$, which defines a `glass ceiling' for
the not adequate `overall delay' function
$ \min(
        s_{\specvarap{\astrfun,1}{n_1}},
        \ldots,
        s_{\specvarap{\astrfun,k}{n_k}}
      )
$,
taking account of situations in which in \daob~rewrite sequences
from terms
$\funap{\astrfun}{\strcns{\pebble^{n_1}}{\aiosq_1}, \ldots, \strcns{\pebble^{n_k}}{\aiosq_k}}$
all input components have been erased.

Concerning non-nesting rules on which defining rules for
friendly nesting symbols depend via $\sdependson$,
this translation uses the fact that
their production is bounded below by `$\min$'.
These bounds are not necessarily optimal,
but can be used to show productivity of examples like
$\strcf{X} \to \strcns{\datf{0}}{\funap{\astrfun}{\strcf{X}}}$ with
$\funap{\astrfun}{\strcns{x}{\astr}} \to \strcns{x}{\funap{\astrfun}{\funap{\astrfun}{\astr}}}$.

Def.~\ref{def:transl:flat:fn:symbols:pifuncspec} can be used to define, 
for every flat or friendly nesting symbol $\astrfun\in\Ssf$ 
in a stream specification~$\atrs$, a `gate translation' of $\astrfun$:
by defining this translation by choosing
a gate that represents the \pein~function~$\sdtrnsl{\astrfun}$. 

However, it is desirable to define this translation also in an
equivalent way that lends itself better for computation,
and where the result directly is a production term (gate) representation 
of a \pein~function:
by specifying the \ioterm{s} in a gate translation
as denotations of rational \ioseq{s} that are
the solutions of `\ioseq\ specifications'. 
In particular, the gate translation of a symbol $\astrfun\in\Ssf$
will be defined in terms of the solutions, 
for each argument place of a stream function symbol $\astrfun$, 
of a finite \ioseq\ specification
that is extracted from an infinite one which precisely
determines the \daob{} lower bound of $\astrfun$ in that argument place.

\begin{definition}\normalfont\label{def:infioseqspecs}
  Let $\setvars$ be a set of variables.
  The set of \emph{\infioseqspecexp{s} over $\setvars$}
  is defined by the following grammar:
  \begin{equation*}
    \iosqsyx
    \BNFis
    \iosqemp
    \BNFor
    X
    \BNFor
    \iosqcns{\ioin}{\iosqsyx}
    \BNFor
    \iosqcns{\ioout}{\iosqsyx}
    \BNFor
    \iotermspecinf{\iosqsyx}{\iosqsyx}
  \end{equation*}
  where $ X\in\setvars $.
  Guardedness is defined by induction:
  An \infioseqspecexp\ is called \emph{guarded}
  if it is of one of the forms 
  $\iosqemp$, $\iosqcns{\ioin}{\iosqsyx_0}$,
  or $\iosqcns{\ioout}{\iosqsyx_0}$, or if it is of the form
  $\iotermspecinf{\iosqsyx_1}{\iosqsyx_2}$ for guarded $\iosqsyx_1$
  and $\iosqsyx_2$.

  Suppose that $ \setvars = \{ X_{\alpha} \mid \alpha\in A \} $
  for some (countable) set $A$.
  Then by an \emph{\infioseqspec\ over (the set of recursion variables) $\setvars$}
  we mean
  a family $ \{ X_{\alpha} = \iosqsyx_{\alpha} \}_{\alpha\in A} $
  of \emph{recursion equations}, where, for all $\alpha\in A$,
  $\iosqsyx_{\alpha}$ is an \infioseqspecexp\ over $\setvars$.
  Let $\aioseqspec$ be an \infioseqspec.
  If $\aioseqspec$  consists of finitely many recursion equations,
  then it is called \emph{finite}.
  $\aioseqspec$ is called \emph{guarded} (\emph{weakly guarded})
  if the right-hand side of every recursion equation in $\aioseqspec$
  is guarded (or respectively, can be rewritten, using equational logic
  and the equations of $\aioseqspec$, to a guarded \infioseqspecexp).

  Let $ \aioseqspec = \{ X_{\alpha} = \iosqsyx_{\alpha} \}_{\alpha\in A} $
  an \infioseqspec. Furthermore let $\alpha_0\in A$ and $\biosq\in\infiosq$.
  We say that $\biosq$ is \emph{a solution of $\aioseqspec$ for $X_{\alpha_0}$}
  if there exist \ioseq{s} $\{\aiosq_{\alpha}\}_{\alpha\in A}$ 
  such that $ \aiosq_{\alpha_0} = \biosq $,
  and all of the recursion equations in $\aioseqspec$ 
  are true statements 
  (under the interpretation of $\siosqinf$ as defined in
   Def.~\ref{def:iosq:inf}),
  when, for all $\alpha\in A$,
  $\aiosq_{\alpha}$ is substituted for $X_{\alpha}$, respectively.
\end{definition}

It turns out that weakly guarded \infioseqspec{s} have unique solutions,
and that the solutions of finite, weakly guarded \infioseqspec{s} 
are rational \ioseq{s}.

\begin{lemma}\label{lem:wgnf}\label{lem:infioseqspec:ratsol}
  For every weakly guarded \infioseqspec~$\aioseqspec$
  and recursion variable $X$ of $\aioseqspec$,
  there exists a unique solution $\aiosq\in\iosq$ of $\aioseqspec$ for $X$.
  Moreover, if $\aioseqspec$ is finite then the solution $\aiosq$
  of $\aioseqspec$ for $X$ is a rational \ioseq{}, 
  and an \ioterm\ that denotes $\aiosq$ can
  be computed on the input of $\aioseqspec$ and $X$. 
\end{lemma}

  Let $\atrs$ be a stream specification, and $\astrfun\in\Ssf$.
  By exhaustiveness of $\atrs$ for $\astrfun$, there is at least one
  defining rule for $\astrfun$ in $\atrs$.
  Since $\atrs$ is a constructor stream TRS, it follows that
  every defining rule $\rho$ for $\astrfun$ is of the form:
 \label{page:simple:non-nest:form}
  \begin{equation}
    \funap{\astrfun}{p_1, \ldots, p_{\arityS{\astrfun}}, q_1, \ldots, q_{\arityD{\astrfun}} }
    \to \strcns{\adattrm_{1}}{\strcns{\ldots}{\strcns{\adattrm_{m}}{\astrtrm}}}
    \tag{$\rho$}
  \end{equation}
  with $p_1, \ldots, p_{\arityS{\astrfun}} \in \ter{\cnssymb{\asig}}_\sortS$,
  $q_1, \ldots, q_{\arityD{\astrfun}} \in \ter{\cnssymb{\asig}}_\sortD$,
  $\adattrm_1,\ldots,\adattrm_m\in\ter{\asig}_\sortD$
  and $\astrtrm \in \ter{\asig}_\sortS$ where $\rootsymb{\astrtrm}\ne\text{`$\sstrcns$'}$.
  If $\rho$ is non-nesting then either $\astrtrm \equiv \astr_j$ or
  \(
    \astrtrm \equiv
    \funap{\bstrfun}{
      \strcns{\vec{\cdattrm}_1}{\astr_{\permut{1}}},
      \ldots,
      \strcns{\vec{\cdattrm}_{\arityS{\strff{g}}}}{\astr_{\permut{\arityS{\strff{g}}}}},
      \ddattrm_1,\ldots,\ddattrm_{\arityD{\strff{g}}}
    }
  \)
  where $\strcns{\vec{\cdattrm}_i}{\astr_i}$ is shorthand for
  $\strcns{\cdattrm_{i,1}}{\strcns{\ldots}{\strcns{\cdattrm_{i,m_i}}{\astr_i}}}$,
  and $\spermut\funin{\{1,\ldots,\arityS{\strff{g}}\}\to\{1,\ldots,\arityS{\astrfun}\}}$
  is a function that describes how the stream arguments are permuted and
  replicated.

We now define, for every given stream specification $\atrs$, 
an infinite \infioseqspec~$\aioseqspec_{\atrs}$ that will be instrumental
for defining gate translations of the stream function symbols in $\atrs$.

\begin{definition}\normalfont
  \label{def:infioseqspec}
  Let $\atrs = \pair{\Sigma}{R}$ be a stream specification.
  Based on the set:
  \begin{align*}
    A \defdby & \: 
      \{ \pair{\lstemp}{\ios{-}{{}}},
         \pair{\lstemp}{\ios{+}{{}}},
         \pair{\lstemp}{\ios{-+}{{}}} 
      \}
    \\
      &
      \;\cup
      \Bigl\{
         \parbox{80pt}{$
           \pair{\astrfun}{\argnone},\,       
           \triple{\astrfun}{\argnone}{\rho},$\\
           $\triple{\astrfun}{i}{q} ,\,
           \quadruple{\astrfun}{i}{q}{\rho}
                       $}
           \,\Big| \;
             \parbox{170pt}{
               $\astrfun\in\Sfnest$, 
               $q\in\nat$,
               $\rho$ defining rule for $\astrfun$,
               \hspace*{\fill}%
               $1\le i\le \arityS{\astrfun}$%
               \hspace*{\fill}      
             }
           \Bigr\}
  \end{align*}
  of tuples
  we define the infinite \infioseqspec~%
  $\aioseqspec_{\atrs} = \{ X_{\alpha} = E_{\alpha} \}_{\alpha\in A}$
  by listing the equations of $\aioseqspec_{\atrs}$.
  We start with the equations
  \begin{align*}
  \Xm & = \iosqemp %
  \punc,
  &
  \Xp & = \iosqcns{\ioout}{\Xp}
  \punc,
  &
  \Xid & = \iosqcns{\ioin}{\iosqcns{\ioout}{\Xid}}
  \punc.
  \end{align*}
  Then we let, for all friendly nesting (or flat) $\astrfun\in\Ssf$
  with arities $k = \arityS{\astrfun}$ and $\ell = \arityD{\astrfun}$:
  \begin{equation*}
    \Xfi{\astrfun}{\argnone}
    = \begin{cases}
        \min
          \big\{
            \Xfiq{\astrfun}{\argnone}{\rho}
            \where \text{$\rho$ a defining rule of $\astrfun$}
          \big\}
        & \text{if $\astrfun$ is weakly guarded,} \\[-0.4ex]
        \Xm & \text{if $\astrfun$ is unguarded,}
      \end{cases}
  \end{equation*}
  and
  for all
  $1\leq i\leq\arityS{\astrfun}$, %
  and $q\in\nat$:
  \begin{equation*}
    \Xfiq{\astrfun}{i}{q}
    = \begin{cases}
        \siosqInfm
          \big\{
            \Xfiqr{\astrfun}{i}{q}{\rho}
            \where \text{$\rho$ a defining rule of $\astrfun$}
          \big\}
        & \text{if $\astrfun$ is weakly guarded,} \\[-0.4ex]
        \Xm & \text{if $\astrfun$ is unguarded.}
      \end{cases}
  \end{equation*}
  For specifying $\Xfiq{\astrfun}{\argnone}{\rho}$ and $\Xfiqr{\astrfun}{i}{q}{\rho}$
  we distinguish the possible forms the rule $\rho$ can have.
  In doing so, we abbreviate terms 
  $\strcns{\adattrm_{i,1}}{\strcns{\ldots}{\strcns{\adattrm_{i,p}}{\astr_i}}}$
  with $\astr_i$ a variable of sort stream by $\strcns{\vec{\adattrm}_i}{\astr_i}$
  and let $\lstlength{\vec{\adattrm}_i} \defdby p$.
  If $\rho$ is nesting, then we let 
  \begin{align*}
    \Xfiq{\astrfun}{\argnone}{\rho} &= \Xp
    \punc,
    &
    \Xfiqr{\astrfun}{i}{q}{\rho} &= \Xid
    \punc.
  \end{align*}
  Otherwise, $\rho$ is non-nesting and of the form:
  \[
    \funap{\astrfun}{
      (\strcns{\vec{\adattrm}_1}{\astr_1}),
      \ldots,
      (\strcns{\vec{\adattrm}_{k}}{\astr_{k}}),
      \bdattrm_1,\ldots,\bdattrm_{\ell}
    }
    \red
    \strcns{\cdattrm_{1}}{\strcns{\ldots}{\strcns{\cdattrm_{m}}{\astrtrm}}}
    \punc,
  \]
  where either (a)~$\astrtrm \equiv \astr_j$, or
  (b)~\(
        \astrtrm \equiv
        \funap{\bstrfun}{
        (\strcns{\vec{\ddattrm}_{\!1}}{\astr_{\permut{1}}}),
        \ldots,
        (\strcns{\vec{\ddattrm}_{\!k'}}{\astr_{\permut{k'}}}),
        \edattrm_1,\ldots,\edattrm_{\ell'}}
      \)
  with $k' = \arityS{\bstrfun}$, $\ell' = \arityD{\bstrfun}$, and
  $\spermut \funin \{1,\ldots,k'\} \to \{1,\ldots,k\}$.
  Let:
  \begin{align*}
    \Xfiq{\astrfun}{\argnone}{\rho} =
      & 
      \begin{cases}
        \Xp &\text{case~(a)} \\
        \iosqcns{\ioout^m}{\Xfi{\bstrfun}{\argnone}} &\text{case~(b)}
      \end{cases}
    \\
    \Xfiqr{\astrfun}{i}{q}{\rho} = 
      & \; \letin{p \defdby \cosubtr{\lstlength{\vec{\adattrm}_i}}{q}}, \:
                  q' \defdby \cosubtr{q}{\lstlength{\vec{\adattrm}_i}}
        \text{ in}
      \\
      &
      \iosqcat{\ioin^{p}}{\iosqcat{\ioout^{m}}}
      \begin{cases}
        \iosqcat{\ioout^{q'}}{\Xid}
        &\text{case~(a), $i = j$} \\
        \Xp
        &\text{case~(a), $i \neq j$} \\
        \siosqInfm
          \big\{
            \Xfiq{\bstrfun}{j}{q' + \lstlength{\vec{\ddattrm}_{\!j}}}
            \where j \in {\funap{\invfun{\spermut}}{i}}
          \big\}
        &\text{case~(b)} %
      \end{cases}
  \end{align*}
  where we agree $\iosqInfm{}{\setemp} \defdby \Xp$.
\end{definition}

We formally state an easy observation about the system $\aioseqspec_{\atrs}$
defined in Def.~\ref{def:transl:flat:fn:symbols:infioseqspec},
and an immediate consequence due to Lem.~\ref{lem:infioseqspec:ratsol}.

\begin{proposition}\label{prop:infioseqspec:wg}
  Let $\atrs$ be a stream specification. 
  The \infioseqspec~$\aioseqspec_{\atrs}$ (defined in
  Def.~\ref{def:transl:flat:fn:symbols:infioseqspec}) is weakly guarded.
  As a consequence, 
  $\aioseqspec_{\atrs}$ has a unique solution in $\iosq$
  for every variable $X$ in $\aioseqspec_{\atrs}$.
\end{proposition}

Another easy observation is that the solutions for the 
variables $\Xfi{\astrfun}{\argnone}$ in a system $\aioseqspec_{\atrs}$
correspond very directly to numbers in $\conat$.

\begin{proposition}\label{prop:infioseqspec:argnone}
  Let $\atrs$ be a stream specification and $\astrfun\in\Ssf$.
  The unique solution of $\aioseqspec_{\atrs}$ (defined in
  Def.~\ref{def:transl:flat:fn:symbols:infioseqspec})
  for $\Xfi{\astrfun}{\star}$ is of the form
  $\ioout^{\omega}$ or $\iosqcns{\ioout^n}{\iosqemp}$.
\end{proposition}

Furthermore observe that $\aioseqspec_{\atrs}$ is infinite in case that
$\Ssf\neq\emptyset$, and hence typically is infinite.
Whereas Lem.~\ref{lem:infioseqspec:ratsol} implies unique solvability
of $\aioseqspec_{\atrs}$ for individual variables
in view of (the first statement in) Prop.~\ref{prop:infioseqspec:wg}, 
it will usually not
guarantee that these unique solutions are rational \ioseq{s}.
Nevertheless it can be shown that all solutions
of %
$\aioseqspec_{\atrs}$ are rational \ioseq{s}.
For our purposes it will suffice to show this only for the solutions 
of $\aioseqspec_{\atrs}$ for certain of its variables.

We will only be interested in the solutions of $\aioseqspec_{\atrs}$
for variables $\Xfiq{\astrfun}{i}{0}$ and $\Xfi{\astrfun}{\argnone}$.
It turns out that from $\aioseqspec_{\atrs}$
a finite weakly guarded specification $\aioseqspec'_{\atrs}$ can be extracted
that, 
for each of the variables $\Xfiq{\astrfun}{i}{0}$ and $\Xfi{\astrfun}{\argnone}$,
has the same solution as $\aioseqspec_{\atrs}$, respectively.
Lem.~\ref{lem:infioseqspec:build:finite} below states that,
for a stream specification $\atrs$, the finite 
\infioseqspec~$\aioseqspec'_{\atrs}$ can always be obtained algorithmically,
which together with Lem.~\ref{lem:infioseqspec:ratsol} implies that
the unique solutions of $\aioseqspec_{\atrs}$ for the
variables $\Xfiq{\astrfun}{i}{0}$ and $\Xfi{\astrfun}{\argnone}$
are rational \ioseq{s}\
for which representing \ioterm{s} can be computed. 
As a consequence these solutions, for a stream specification $\atrs$,
of the recursion system $\aioseqspec_{\atrs}$, can be viewed to represent \pein~functions: 
the \pein~functions that are represented by the \ioterm\
denoting the respective solution, a rational \ioseq.

As an example let us consider an \infioseqspec\ that corresponds to 
 the defining rules for $\astrfun$ in Fig.~\ref{fig:pascal}:
\begin{align*}
  X %
  & = \iotermspecinf
        {\iosqcns{\ioin}{\iosqcns{\ioout}{\iosqcns{\ioout}{X}}}}
        {\iosqcns{\ioin}{\iosqcns{\ioin}{\iosqcns{\ioout}{Y}}}}
  &
  Y %
  & = \iotermspecinf
        {\iosqcns{\ioout}{\iosqcns{\ioout}{X}}}
        {\iosqcns{\ioin}{\iosqcns{\ioout}{Y}}} \; .
\end{align*}
The unique solution for $X$ of this system is the rational \ioseq\
(and respectively, \ioterm)
$\iosqcns{\ioin}{\iosqcyc{\iosqcns{\ioin}{\iosqcns{\ioout}{\nix}}}}$,
which is the translation of $\astrfun$ (as mentioned earlier).

\begin{lemma}\label{lem:infioseqspec:build:finite}
  Let $\atrs$ be a stream specification.
  There exists an \infioseqspec\
  $\aioseqspec'_{\atrs} = \{ X_{\alpha} = E'_{\alpha} \}_{\alpha\in A'}$
  such that:
  \begin{enumerate}
    \renewcommand{\labelenumi}{(\roman{enumi})}
    \item $\aioseqspec'_{\atrs}$ is finite and weakly guarded;    
    \item $\{ \langle \astrfun,i,0 \rangle,
              \langle \astrfun,\argnone \rangle \,\mid\,
          \astrfun\in\Ssf, \, i,q\in\nat, \, 
          i\in\{1,\ldots,\arityS{\astrfun}\} \,\}
          \subseteq A' $;\\
      for each $\astrfun\in\Ssf\cap\Sfnest$ 
      and $i\in\{1,\ldots,\arityS{\astrfun}\}$,
      $\aioseqspec'_{\atrs}$ has the same solution 
      for $\Xfiq{\astrfun}{i}{0}$,
      and respectively for $\Xfi{\astrfun}{\argnone}$,
      as the \infioseqspec{s}~$\aioseqspec_{\atrs}$ 
      (see Def.~\ref{def:infioseqspec});
    \item on the input of $\atrs$, $\aioseqspec'_{\atrs}$ 
      can be computed. 
  \end{enumerate}
\end{lemma}

\begin{proof}[Sketch]
  An algorithm for obtaining $\aioseqspec'_{\atrs}$
  from $\aioseqspec_{\atrs}$ can be obtained as follows.
  On the input of $\aioseqspec_{\atrs}$,
  set $\aioseqspec\defdby\aioseqspec_{\atrs}$ and 
  repeat the following step on $\aioseqspec$ as long as it is applicable:
  \begin{description}
    \item[{\sf (RPC)}]
      Detect and remove a reachable \emph{non-consuming pseudo-cycle} from the 
      \infioseqspec~$\aioseqspec\,$:
      Suppose that, for a function symbol $\strff{h}$, for $j,k,l\in\nat$,
      and for a recursion variable $\Xfiq{\strff{h}}{j}{k}$ that is reachable
      from $\Xfiq{\astrfun}{i}{0}$, we have
      $ \lts{\Xfiq{\strff{h}}{j}{k}}{w}{\Xfiq{\strff{h}}{j}{l}} $
      (from the recursion variable $\Xfiq{\strff{h}}{j}{k}$
       the variable $\Xfiq{\strff{h}}{j}{l}$ is reachable
       via a path in the specification on which the finite~\ioseq~$w$
       is encountered as the word formed by consecutive labels),
      where $k<l$ and $w$ only contains symbols `$\ioout$'.
      Then modify $\aioseqspec$ by setting
      $ \Xfiq{\strff{h}}{j}{k} = \Xp %
      $.
  \end{description}
  It is not difficult to show that a step {\sf (RPC)} 
  preserves weakly guardedness and the unique solution 
  of $\aioseqspec_{\atrs}$, 
  and that, on the input of $\aioseqspec_{\atrs}$, 
  the algorithm terminates in finitely many steps, 
  having produced an \infioseqspec~$\aioseqspec'_{\atrs}$ 
  with $\trnsli{\astrfun}{i}$ as the solution for $\Xfiq{\astrfun}{i}{0}$
  and the property
  that only finitely many recursion variables are reachable
  in $\aioseqspec'_{\atrs}$ from
  $\Xfiq{\astrfun}{i}{0}$.
  \qed
\end{proof}

As an immediate consequence of 
Lem.~\ref{lem:infioseqspec:build:finite},
 
and of Lem.~\ref{lem:infioseqspec:ratsol} we obtain the following lemma.

\begin{lemma}\label{lem:wd:def:transl:flat:fn:symbols:infioseqspec}
  Let $\atrs$ be a stream specification,
  and let $\astrfun\in\Ssf$.
  \begin{enumerate}
    \item For each $i\in\{1,\ldots,\arityS{\astrfun}\}$
       there is precisely one \ioseq~$\aiosq$ that solves
       $\aioseqspec_{\atrs}$ for $\Xfiq{\astrfun}{i}{0}$;
       furthermore, $\aiosq$ is rational.
       Moreover there exists an algorithm that,
       on the input of $\atrs$, $\astrfun$, and $i$,
       computes the shortest \ioterm\ that denotes the 
       solution of $\aioseqspec_{\atrs}$ for $\Xfiq{\astrfun}{i}{0}$.
    \item There is precisely one \ioseq~$\aiosq$ that solves
       $\aioseqspec_{\atrs}$ for $\Xfi{\astrfun}{\argnone}$;
       this \ioseq\ is rational.
       And there is an algorithm that, 
       on the input of $\atrs$ and $\astrfun$,
       computes the shortest \ioterm\ that denotes the 
       solution of $\aioseqspec_{\atrs}$ 
       for $\Xfi{\astrfun}{\argnone}$.
  \end{enumerate}
\end{lemma}

Lem.~\ref{lem:wd:def:transl:flat:fn:symbols:infioseqspec} guarantees 
the well-definedness in the definition below
of the `gate translation' for flat or friendly nesting stream functions 
in a stream specification. 

\begin{definition}\normalfont
  \label{def:transl:flat:fn:symbols:infioseqspec}
  Let $\atrs = \pair{\Sigma}{R}$ be a stream definition. 
  For each flat or friendly nesting symbol
  $\astrfun \in \Ssf \cap \Sfnest $ 
  with stream arity $k = \arityS{\astrfun}$ %
  we define the \emph{gate translation $\trnsl{\astrfun}$ of\/~$\astrfun$} by:
  \begin{equation*}
    \trnsl{\astrfun}
    \defdby
    \egate{\trnsli{\astrfun}{\argnone}}{
      \trnsli{\astrfun}{1},\ldots,\trnsli{\astrfun}{\arityS{\astrfun}}
    }
    \punc,
  \end{equation*}
  where $\trnsli{\astrfun}{\argnone}\in\conat$ is defined by:
  \begin{equation*}
    \trnsli{\astrfun}{\argnone} \defdby
      \begin{cases}
        n & \text{the unique solution of $\aioseqspec_{\atrs}$ 
              for $\Xfi{\astrfun}{\argnone}$ is $\iosqcns{\ioout^{n}}{\iosqemp}$}
        \\
        \conattop &
          \text{the unique solution of $\aioseqspec_{\atrs}$ 
              for $\Xfi{\astrfun}{\argnone}$ is $\ioout^{\omega}$}    
      
      \end{cases}
  \end{equation*}
  and, for $1\le i\le \arityS{\astrfun}$,
  $\trnsli{\astrfun}{i}$ is the shortest \ioterm\ that denotes
  the unique solution for $\Xfiq{\astrfun}{i}{0}$
  of the weakly guarded \infioseqspec~$\aioseqspec_{\atrs}$
  in Def.~\ref{def:infioseqspec}.
\end{definition}

From Lem.~\ref{lem:wd:def:transl:flat:fn:symbols:infioseqspec}
we also obtain that, for every stream specification,
the function which maps stream function symbols
to their gate translations is computable.

\begin{lemma}\label{lem:transl:termination}
  There is an algorithm that, on the input
  of a stream specification~$\atrs$, and
  a flat or friendly nesting symbol~$\astrfun\in\Ssf\cap\Sfnest$,
  computes the gate translation $\trnsl{\astrfun}$ of $\astrfun$.
\end{lemma}

\begin{example}
  Consider a flat stream specification
  consisting of the %
  rules:
  \begin{gather*}
    \funap{\astrfun}{\strcns{x}{\astr}}
    \to \strcns{x}{\tfunap{\bstrfun}{\astr}{\astr}{\astr}}
    \punc,
    \\
    \tfunap{\bstrfun}{\strcns{x}{\strcns{y}{\astr}}}{\bstr}{\cstr}
    \to \strcns{x}{\tfunap{\bstrfun}{\strcns{y}{\bstr}}{\strcns{y}{\cstr}}{\strcns{y}{\astr}}}
    \punc.
  \end{gather*}
  The translation of $\astrfun$ is $\trnsl{\astrfun}=\netgate{\trnsli{\astrfun}{1}}$,
  where $\trnsli{\astrfun}{1}$ is the unique solution for $\Xfiq{\astrfun}{1}{0}$
  of the \infioseqspec~$\aioseqspec_{\atrs}$:
  \begin{gather*}
    \Xfiq{\astrfun}{1}{0}
      =
    \iosqcns{\ioin}{\iosqcns{\ioout}{
      (\iotermspecinf{\Xfiq{\bstrfun}{1}{0}}{\iotermspecinf{\Xfiq{\bstrfun}{2}{0}}{\Xfiq{\bstrfun}{3}{0}}})
      }}
    \\
    \begin{aligned}
      \Xfiq{\bstrfun}{1}{0}
      & = \iosqcns{\ioin}{\iosqcns{\ioin}{\iosqcns{\ioout}{\Xfiq{\bstrfun}{3}{1}}}}
      &
     \Xfiq{\bstrfun}{1}{1}
      & = \iosqcns{\ioin}{\iosqcns{\ioout}{\Xfiq{\bstrfun}{3}{1}}}
    \end{aligned}
    \\
    \Xfiq{\bstrfun}{1}{q}
      = \iosqcns{\ioout}{\Xfiq{\bstrfun}{3}{q-1}}
      \quad (q \geq 2)
    \\
    \Xfiq{\bstrfun}{2}{q}
      = \iosqcns{\ioout}{\Xfiq{\bstrfun}{1}{q+1}}
    \quad (q\in\nat)
    \\
    \Xfiq{\bstrfun}{3}{q}
      = \iosqcns{\ioout}{\Xfiq{\bstrfun}{2}{q+1}}
    \quad (q\in\nat)
  \end{gather*}

  By the algorithm referred to in Lem.~\ref{lem:transl:termination}
  this infinite specification can be turned into a finite one. %
  The `non-consuming pseudocycle'
  $\lts{\Xfiq{\bstrfun}{3}{1}}{\ioout\ioout\ioout}{\Xfiq{\bstrfun}{3}{2}}$
  justifies the modification of $\aioseqspec_{\atrs}$ by setting
  $\Xfiq{\bstrfun}{3}{1} = \iosqcns{\ioout}{\Xfiq{\bstrfun}{3}{1}}$;
  likewise we set $\Xfiq{\bstrfun}{3}{0} = \iosqcns{\ioout}{\Xfiq{\bstrfun}{3}{0}}$.
  Furthermore all equations not reachable from $\Xfiq{\astrfun}{1}{0}$ are removed
  (garbage collection), 
  and we obtain a finite specification $\aioseqspec'_{\atrs}$, which,
  by Lem.~\ref{lem:infioseqspec:ratsol}, has a periodically increasing
  solution, with \ioterm-denotation
  $\trnsli{\astrfun}{1} = \pair{\ios{-+--}{\nix}}{\ios{+}{\nix}} $.
  The reader may try to calculate the gate corresponding to $\bstrfun$,
  it is:
  \(
    \trnsl{\bstrfun} = \netgate{\ios{--}{+},\ios{+-}{+},\ios{}{+}}
  \).

  Second, consider the flat stream specification
  with data constructor symbols $\datf{0}$ and $\datf{1}$:
  \begin{gather*}
    \begin{aligned}
      \funap{\astrfun}{\strcns{\datf{0}}{\astr}} & \to \funap{\bstrfun}{\astr} \punc,
      & \hspace*{1ex}
      \funap{\astrfun}{\strcns{\datf{1}}{\strcns{x}{\astr}}} & \to \strcns{x}{\funap{\bstrfun}{\astr}} \punc,
    \end{aligned}
    \\
    \funap{\bstrfun}{\strcns{x}{\strcns{y}{\astr}}} \to \strcns{x}{\strcns{y}{\funap{\bstrfun}{\astr}}} \punc,
  \end{gather*}
  denoted $\rho_{\astrfunsub\datf{0}}$, $\rho_{\astrfunsub\datf{1}}$,
  and $\rho_{\bstrfunsub}$, respectively.
  Then, $\trnsli{\astrfun}{1}$ is the solution
  for $\Xfiq{\astrfun}{1}{0}$ of %
  \begin{gather*}
    \Xfiq{\astrfun}{1}{0}
      =
    \iotermspecinf{
      \Xfiqr{\astrfun}{1}{0}{\rho_{\astrfunsub\datf{0}}}}{
      \Xfiqr{\astrfun}{1}{0}{\rho_{\astrfunsub\datf{1}}}
                   }
    \\
    \begin{aligned}
      \Xfiqr{\astrfun}{1}{0}{\rho_{\astrfunsub\datf{0}}}
        & =
      \iosqcns{\ioin}{\Xfiq{\bstrfun}{1}{0}}
      & \hspace*{0.5ex}
      \Xfiqr{\astrfun}{1}{0}{\rho_{\astrfunsub\datf{1}}}
        & =
      \iosqcns{\ioin}{\iosqcns{\ioin}{\iosqcns{\ioout}{\Xfiq{\bstrfun}{1}{0}}}}
    \end{aligned}
    \\
    \Xfiq{\bstrfun}{1}{0}
      =
    \iosqcns{\ioin}{\iosqcns{\ioin}{\iosqcns{\ioout}{\iosqcns{\ioout}{\Xfiq{\bstrfun}{1}{0}}}}}
    \punc.
  \end{gather*}
\end{example}

\begin{example}\label{ex:pure}
  Consider a pure stream specification with the function layer: %
  \begin{gather*}
    \funap{\astrfun}{\strcns{x}{\astr}} \to \strcns{x}{\tfunap{\bstrfun}{\astr}{\astr}{\astr}} \punc,
    \\
    \tfunap{\bstrfun}{\strcns{x}{\strcns{y}{\astr}}}{\bstr}{\cstr}
    \to \strcns{x}{\tfunap{\bstrfun}{\strcns{y}{\bstr}}{\strcns{y}{\cstr}}{\strcns{y}{\astr}}} \punc.
  \end{gather*}
  The translation of $\astrfun$ is $\sdtrnsl{\astrfun}$,
  the unique solution for $\specvar{\astrfun}$ of the system:
  \begin{align*}
    \specvarap{\astrfun}{n}
    =\ &\min(\specvarap{\astrfun,\argnone}{0}, \specvarap{\astrfun,1}{n})\\
    \specvarap{\astrfun,1}{n}
    =\ &\letin{n' \defdby n-1}\\
      &\text{if $n' < 0$ then $0$ else }
      1 +
      \inf\, \big\{ \specvarap{\bstrfun,1}{n'},\, \specvarap{\bstrfun,2}{n'},\,
                 \specvarap{\bstrfun,3}{n'} \big\}
    \\
    \specvarap{\astrfun,\argnone}{n}
    =\ &1 + \specvarap{\bstrfun,\argnone}{0}
    \\
    \specvarap{\bstrfun,1}{n}
    =\ &\letin{n' \defdby n-2},\;
      \text{if $n' < 0$ then $0$ else $1 + \specvarap{\bstrfun,3}{1 + n'}$}
    \\
    \specvarap{\bstrfun,2}{n}
    =\ &1 + \specvarap{\bstrfun,1}{1 + n}
    \\
    \specvarap{\bstrfun,3}{n}
    =\ &1 + \specvarap{\bstrfun,2}{1 + n}
    \\
    \specvarap{\bstrfun,\argnone}{n}
    =\ &1 + \specvarap{\astrfun,\argnone}{0}
  \end{align*}
  An algorithm for solving such systems of equations is described
  in (the proof of) Lemma~\ref{lem:infioseqspec:build:finite}; %
  here we solve the system directly.
  Note that
  $\specvarap{\bstrfun,3}{n} = 1 + \specvarap{\bstrfun,2}{n+1} = 2 + \specvarap{\bstrfun,1}{n+2} = 3 + \specvarap{\bstrfun,3}{n}$,
  hence $\myall{n \in \nat}{\specvarap{\bstrfun,3}{n} = \conattop}$.
  Likewise we obtain $\specvarap{\bstrfun,2}{n} = \conattop$ if $n \ge 1$ and $1$ for $n = 0$,
  and $\specvarap{\bstrfun,1}{n} = \conattop$ if $n \ge 2$ and $0$ for $n \le 1$.
  Then we get $\dtrnsl{\astrfun}{0}=0$, $\dtrnsl{\astrfun}{1}=\dtrnsl{\astrfun}{2}=1$,
  and $\dtrnsl{\astrfun}{n}=\conattop$ for all $n\geq 2$,
  represented by the gate $\trmrep{\sdtrnsl{\astrfun}}=\netgate{\ios{-+--}{+}}$.
  The gate corresponding to $\bstrfun$ is
  $\trmrep{\sdtrnsl{\bstrfun}}=\netgate{\ios{--}{+},\ios{+-}{+},\ios{}{+}}$.
  \end{example}
  \begin{example}\label{ex:flat}
  Consider a flat stream function specification
  with the following rules which use %
  pattern matching on the data constructors $\datf{0}$ and $\datf{1}$:
  \begin{align*}
    \funap{\astrfun}{\strcns{\datf{0}}{\astr}} &\to \funap{\bstrfun}{\astr}
    &
    \funap{\astrfun}{\strcns{\datf{1}}{\strcns{x}{\astr}}} &\to \strcns{x}{\funap{\bstrfun}{\astr}}
    &
    \funap{\bstrfun}{\strcns{x}{\strcns{y}{\astr}}} &\to \strcns{x}{\strcns{y}{\funap{\bstrfun}{\astr}}}
  \end{align*}
  denoted $\rho_{\astrfunsub_\datf{0}}$, $\rho_{\astrfunsub_\datf{1}}$,
  and $\rho_{\bstrfunsub}$, respectively.
  Then, $\sdtrnsl{\astrfun}$ is the solution for $\specvar{\astrfun,1}$ of:
  \begin{align*}
    \specvarap{\astrfun}{n}
    &= \min(\specvarap{\astrfun,\argnone}{0}, \specvarap{\astrfun,1}{n})\\
    \specvarap{\astrfun,1}{n}
    &=
    \inf\, \big\{
      \specvarap{\astrfun,1,\rho_{\astrfunsub_\datf{0}}}{n},\,
      \specvarap{\astrfun,1,\rho_{\astrfunsub_\datf{1}}}{n} \big\}
    \\
    \specvarap{\astrfun,1,\rho_{\astrfunsub_\datf{0}}}{n}
    & =
    \letin{n' \defdby n-1},\;
    \text{if $n' < 0$ then $0$ else $\specvarap{\bstrfun,1}{n'}$}
    \\
    \specvarap{\astrfun,1,\rho_{\astrfunsub_\datf{1}}}{n}
    &=
    \letin{n' \defdby n-2},\;
    \text{if $n' < 0$ then $0$ else $1 + \specvarap{\bstrfun,1}{n'}$}
    \\
    \specvarap{\astrfun,\argnone}{n}
    &= \min(\specvarap{\bstrfun,\argnone}{0}, 1 + \specvarap{\bstrfun,\argnone}{0})
    \\
    \specvarap{\bstrfun,1}{n}
    &=
    \letin{n' \defdby n-2},\;
    \text{if $n' < 0$ then $0$ else $2 + \specvarap{\bstrfun,1}{n'}$}
    \\
    \specvarap{\bstrfun,\argnone}{n}
    &= 2 + \specvarap{\bstrfun,\argnone}{0}
    \punc.
  \end{align*}
  As solution we obtain an overlapping of both traces
  $\sdtrnsli{\astrfun}{1,\rho_{\astrfunsub_\datf{0}}}$ and
  $\sdtrnsli{\astrfun}{1,\rho_{\astrfunsub_\datf{1}}}$,
  that is,
  $\dtrnsli{\astrfun}{1}{n} = \cosubtr{n}{2}$
  represented by the gate %
  $\trmrep{\sdtrnsl{\astrfun}} = \netgate{\ios{--}{-+}}$.
\end{example}

The following proposition explains the correspondence between
Def.~\ref{def:transl:flat:fn:symbols:pifuncspec}
and Def.~\ref{def:transl:flat:fn:symbols:infioseqspec}.

\begin{proposition}\label{prop:corr:pein:gate:transl}
  Let $\atrs$ be a stream specification.
  For each $\astrfun\in\Ssf$ it holds:
  \begin{equation*}
    \funap{\sdtrnsl{\astrfun}}{n_1,\ldots,n_k}
      = \iosqprd{\trnsl{f}}{n_1,\ldots,n_k}
    \qquad
    \text{(for all $n_1,\ldots,n_k\in\conat$),} 
  \end{equation*}
  where $k = \arityS{\astrfun}$.
  That is,
  the \pein~function translation $\sdtrnsl{\astrfun}$ of\/ $\astrfun$
  coincides with the \pein~function that is represented by the
  gate translation $\trnsl{\astrfun}$.
\end{proposition}

The following lemma states that the translation $\sdtrnsl{\astrfun}$
of a flat stream function symbol $\astrfun$ 
(as defined in Def.~\ref{def:transl:flat:fn:symbols:pifuncspec})
is the \daob{} lower bound on the production function of $\astrfun$.
For friendly nesting stream symbols $\astrfun$ it states
that $\sdtrnsl{\astrfun}$ pointwisely bounds from below
the \daob{} lower bound on the production function of $\astrfun$.

\begin{lemma}\label{lem:transl:soundness}
  Let $\atrs$ be a stream specification,
  and let\/ $\astrfun \in \Sfnest \supseteq \Sflat$.
  \begin{enumerate}
    \item If\/ $\astrfun$ is flat, then:\/
      $\sdtrnsl{\astrfun} = \doLow{\atrs}{\astrfun}$.
      Hence, $\doLow{\atrs}{\astrfun}$ is periodically increasing.
    \item If\/ $\astrfun$ is friendly nesting, then it holds:\/
      $\sdtrnsl{\astrfun} \le \doLow{\atrs}{\astrfun}$
      (pointwise inequality).
  \end{enumerate}
\end{lemma}

\subsection{Translation of Stream Constants}
  \label{sec:translation:subsec:constants}

In the second step, we now define a translation of
  stream constants
in a flat or friendly nesting stream specification into production terms
under the assumption that gate translations for the stream functions are given.
Here the idea is that the recursive definition of a stream constant
$\msf{M}$ is unfolded step by step;
the terms thus arising are translated according to their
structure using gate translations of the stream function symbols
from a given family of gates;
whenever a stream constant is met that has been unfolded before,
the translation stops after establishing a binding to a $\mu$-binder
created earlier.

\begin{definition}\label{def:trnsl:nets}\normalfont
  Let $\atrs = \pair{\asig}{R} $ be a stream specification,
  and $\afam = \{\fgate\}_{\astrfun\in\Ssf}$ a family of gates that
  are associated with the symbols in $\Ssf$.
  
  Let $\ter{\asig}_{\sortS}^0$ be the set of terms in $\atrs$
  of sort stream that do not contain variables of sort stream. 
  The \emph{translation function}
  $\trnslF{\cdot}{\afam} \funin \ter{\asig}_{\sortS}^0 \to \net$
  is defined %
  by $\astrtrm \mapsto \trnslF{\astrtrm}{\afam} \defdby \trnsliF{\astrtrm}{\setemp}{\afam}$
  based on the following definition of expressions
  $\trnsliF{\astrtrm}{\alst}{\afam}$, where $\alst \subseteq \Ssc$,
  by induction 
  on the structure of $\astrtrm\in\ter{\asig}_{\sortS}^0$,
  using the clauses:
  \begin{gather*}
    \begin{aligned}
    \trnsliF{\funap{\astrcon}{\vec{\adattrm}}}{\alst}{\afam}
    &\defdby
      \begin{cases}
      \netrec{M}{\snetmeet\; \{ \trnsliF{r}{\setunion{\alst}{\{\astrcon\}}}{\afam}
                         \where \funap{\astrcon}{\vec{\bdattrm}} \to r \in R\} }
      &\text{if $\astrcon\not\in\alst$}
      \\
      M
      &\text{if $\astrcon\in\alst$}
      \end{cases}
    \\
    \trnsliF{\strcns{\adattrm}{\astrtrm}}{\alst}{\afam}
    &\defdby \netpeb{\trnsliF{\astrtrm}{\alst}{\afam}}
    \end{aligned}
    \\
    \trnsliF{\funap{\astrfun}{\astrtrm_1,\ldots,\astrtrm_{\arityS{\astrfun}},\adattrm_1,\ldots,\adattrm_{\arityD{\astrfun}}}}%
            {\alst}{\afam}
    \defdby
    \funap{\fgate}
       {\trnsliF{\astrtrm_1}{\alst}{\afam},\ldots,
          \trnsliF{\astrtrm_{\arityS{\astrfun}}}{\alst}{\afam}}
  \end{gather*}
  where $\astrcon\in\Ssc$, 
  $\vec{\adattrm} = \langle \adattrm_1, \ldots, \adattrm_{\arityD{\astrcon}} \rangle$
  a vector of terms in $\ter{\asig}_{\sortD}$, 
  $\astrfun\in\Ssf$,\linebreak
  $\astrtrm, \astrtrm_1, \ldots \astrtrm_{\arityS{\astrfun}} \in\ter{\asig}_{\sortS}^0$,
  and 
  $\adattrm, \adattrm_1,\ldots,\adattrm_{\arityD{\astrfun}}\in\ter{\asig}_{\sortD}$.
  Note that the definition
  of $\trnsliF{\funap{\astrcon}{\vec{\adattrm}}}{\alst}{\afam}$
  does not depend on the vector $\vec{\adattrm}$ of terms
  in $\ter{\asig}_{\sortD}$.
  Therefore we define, for all $\astrcon\in\Ssc$,
  the \emph{translation of $\astrcon$ with respect to $\afam$} by 
  $\trnslF{\astrcon}{\afam} \defdby 
    \trnsliF{\funap{\astrcon}{\vec{x}}}{\emptyset}{\afam}
    \in \net $
  where $\vec{x} = \langle x_1, \ldots, x_{\arityD{\astrcon}} \rangle$
  is a vector of data variables. 
\end{definition}

The following lemma is the basis of our result in Sec.~\ref{sec:results},
Thm.~\ref{thm:decide:prod:flat:pure}, concerning the decidability of 
\daob{} productivity for flat stream specifications.
In particular the lemma states that if we use gates that represent \daobly{} optimal lower bounds
on the production of the stream functions,
then the translation of a stream constant $\astrcon$
yields a production term that rewrites to the \daob{} lower bound of the production of $\astrcon$.

\begin{lemma}
  \label{lem1:outsourcing}
  Let $\atrs = \pair{\Sigma}{R}$ be a stream specification, and
  $\afam = \{\fgate\}_{\astrfun\in\Ssf}$ be a family of gates
  such that, for all\/ $\astrfun\in\Ssf$, the arity of\/ $\fgate$
  equals the stream arity of\/ $\astrfun$.
  Then the following statements hold:
  \begin{enumerate}
    \renewcommand{\labelenumi}{(\roman{enumi})}
    \item\label{lem1:outsourcing:item:1}
      Suppose that\/ $\siosqprd{\fgate} = \doLow{\atrs}{\astrfun}$
      holds for all\/ $\astrfun \in \Ssf$.
      Then for all\/ $\astrcon\in\Ssc$ 
      and vectors $\vec{\adattrm} = \langle \adattrm_1, \ldots, \adattrm_{\arityD{\astrcon}} \rangle$ 
      of data terms
      $\; \netprd{\trnslF{\astrcon}{\afam}}
                = \doLow{\atrs}{\funap{\astrcon}{\vec{u}}} $
      holds.                   
      And consequently, $\atrs$ is \daobly{} productive
      if and only if\/ $\netprd{\trnslF{\rootsc}{\afam}} = \infty$.
      \vspace*{0.75ex}
    \item\label{lem1:outsourcing:item:2}
      Suppose that\/ $\siosqprd{\fgate} = \doUp{\atrs}{\astrfun}$
      holds for all\/ $\astrfun \in \Ssf$.
      Then for all\/ $\astrcon\in\Ssc$ 
      and vectors $\vec{\adattrm} = \langle \adattrm_1, \ldots, \adattrm_{\arityD{\astrcon}} \rangle$ 
      of data terms 
      $\; \netprd{\trnslF{\astrcon}{\afam}}
                = \doUp{\atrs}{\funap{\astrcon}{\vec{u}}} $
      holds.          
      And consequently, $\atrs$ is \daobly{} non-productive
      if and only if\/ $\netprd{\trnslF{\rootsc}{\afam}} < \infty$.
  \end{enumerate}
\end{lemma}
Lem.~\ref{lem1:outsourcing} is an immediate consequence of the following lemma,
which also is the basis of our result in Sec.~\ref{sec:results} 
concerning the recognizability of productivity for flat and friendly nesting
stream specifications. 
In particular the lemma below asserts that if we use gates that
represent \pein{} functions which are lower bounds 
on the production of the stream functions,
then the translation of a stream constant $\astrcon$
yields a production term that rewrites to a number in $\conat$
smaller or equal to the \daob{} lower bound of the production of $\astrcon$.

\begin{lemma}
  \label{lem2:outsourcing}
  Let $\atrs$ be a stream specification, and let\/
  $\afam = \{\fgate\}_{\astrfun\in\Ssf}$ be a family of gates
  such that, for all\/ $\astrfun\in\Ssf$, the arity of\/ $\fgate$
  equals the stream arity of\/ $\astrfun$.
  Suppose that one of the following statements holds:
  \begin{enumerate}[(a)]
    \item $\siosqprd{\fgate} \le \doLow{\atrs}{\astrfun}$ for all\/ $\astrfun \in \Ssf\,$;%
      \label{lem2:outsourcing:item:1}%
      \vspace{0.5ex}
    \item $\siosqprd{\fgate} \ge \doLow{\atrs}{\astrfun}$ for all\/ $\astrfun \in \Ssf\,$;%
      \label{lem2:outsourcing:item:2}%
      \vspace{0.5ex}
    \item $\doUp{\atrs}{\astrfun} \le \siosqprd{\fgate}$ for all\/ $\astrfun \in \Ssf\,$;%
      \label{lem2:outsourcing:item:3}%
      \vspace{0.5ex} 
    \item $\doUp{\atrs}{\astrfun} \ge \siosqprd{\fgate}$ for all\/ $\astrfun \in \Ssf\,$.%
      \label{lem2:outsourcing:item:4}%
      \vspace{0.5ex}
  \end{enumerate}
  Then, for all $\astrcon\in\Ssc$
  and vectors 
  $\vec{\adattrm} = \langle \adattrm_1, \ldots, \adattrm_{\arityD{\astrcon}} \rangle$ 
  of data terms in $\atrs$,
  the corresponding one of the following statements holds:
  \begin{enumerate}[(a)]
    \item $\netprd{\trnslF{\astrcon}{\afam}} \le \doLow{\atrs}{\funap{\astrcon}{\vec{u}}}\,$;%
      \vspace{0.5ex}
    \item $\netprd{\trnslF{\astrcon}{\afam}} \ge \doLow{\atrs}{\funap{\astrcon}{\vec{u}}}\,$;%
      \vspace{0.5ex}
    \item $\doUp{\atrs}{\funap{\astrcon}{\vec{u}}} \le \netprd{\trnslF{\astrcon}{\afam}}\,$;%
      \vspace{0.5ex}
    \item $\doUp{\atrs}{\funap{\astrcon}{\vec{u}}} \ge \netprd{\trnslF{\astrcon}{\afam}}\,$.
  \end{enumerate}
\end{lemma}

\section{Deciding Data-Oblivious Productivity}\label{sec:results}
In this section we assemble our results concerning
decision of \daob{} productivity, and automatable recognition
of productivity. %
We define methods:
{\setlength{\leftmargini}{10ex}
\begin{itemize}
  \item [(\algdecidedoprod)] for deciding \daob{} productivity of flat stream specifications,
  \item [(\algdecideprod)] for deciding productivity of pure stream specifications, and
  \item [(\algrecogprod)] for recognising productivity of friendly nesting stream specifications,
\end{itemize}}
\noindent
that proceed in the following steps:
\begin{enumerate}
  \item
    Take as input a
    (\algdecidedoprod) flat,
    (\algdecideprod) pure, or
    (\algrecogprod) friendly nesting
    stream specification~$\atrs = \pair{\asig}{R}$.
  \item
    Translate the stream function symbols %
    into %
    gates
    $\afam \defdby \{\trmrep{\sdtrnsl{\astrfun}} \}_{\astrfun\in\Ssf}$
    (Def.~\ref{def:transl:flat:fn:symbols:pifuncspec}).
  \item
    Construct the production term $\trnslF{\rootsc}{\afam}$ %
    with respect to $\afam$ (Def.~\ref{def:trnsl:nets}).
  \item\label{item:algo:production}
    Compute the production $k$
    of $\trnslF{\rootsc}{\afam}$ using $\scolred$ (Def.~\ref{def:coltrs}).
  \item
    Give the following output:
    \begin{itemize}
      \item [(\algdecidedoprod)]
        ``$\atrs$ is \daobly{} productive'' if $k=\conattop$,
        else ``$\atrs$ is not \daobly{} productive''.
      \item [(\algdecideprod)]
        ``$\atrs$ is productive'' if $k=\conattop$,
        else ``$\atrs$ is not productive''. %
      \item [(\algrecogprod)]
        ``$\atrs$ is productive'' if $k=\conattop$,
        else ``don't know''.
    \end{itemize}
\end{enumerate}
Note that all of these steps are automatable
(cf.\ our productivity tool, Sec.~\ref{sec:conclusion}).

Our main result states that \daob{} productivity is decidable
for flat stream specifications.
It rests on the fact that the algorithm \algdecidedoprod{} obtains
the lower bound $\doLow{\atrs}{\trnsl{\rootsc}}$ on the production
of $\rootsc$ in $\atrs$.
Since \daob{} productivity implies
productivity (Prop.~\ref{prop:doprod}), we obtain a computable, \daobly{} optimal,
sufficient condition for productivity of flat stream specifications, which
cannot be improved by any other \daob{} analysis.
Second, since for pure stream specifications
\daob{} productivity and productivity are the same,
we get that productivity is decidable for them.

\begin{theorem} \label{thm:decide:prod:flat:pure}
  \begin{enumerate}
    \item \algdecidedoprod{} decides \daob{} productivity of flat stream specifications,
    \item \algdecideprod{} decides productivity of pure stream specifications.
  \end{enumerate}
\end{theorem}
\begin{proof}
  Let $k$ be the production of the term $\trnslF{\rootsc}{\afam}\in\net$ in step~(iv) of
  \algdecidedoprod/\algdecideprod.
  \vspace{-1ex}
  \begin{enumerate}
  \item By Lem.~\ref{lem:transl:soundness}~(i), 
    Lem.~\ref{lem1:outsourcing}~\ref{lem1:outsourcing:item:1},
    Lem.~\ref{lem:collapse:pop}, Lem.~\ref{lem:collapse:sn:cr} %
    we find:
    $k = \doLow{\atrs}{\rootsc}$.
  \item For pure specifications we additionally note:
    $\terprd{\atrs}{\rootsc} = \doLow{\atrs}{\rootsc}$.
  \qed
  \end{enumerate}
\end{proof}

Third, we obtain a computable, sufficient condition for productivity
of friendly nesting stream specifications.
\begin{theorem} \label{thm:recog:prod:friendly:nesting:flat}
  A friendly nesting (flat) stream specification $\atrs$ is productive
  if the algorithm \algrecogprod (\/\algdecidedoprod) recognizes $\atrs$ as productive.
\end{theorem}

\begin{proof}
  Let $k$ be the production of the term $\trnslF{\rootsc}{\afam}\in\net$
  in step~(iv) of \algrecogprod/\algdecidedoprod.
  By Lem.~\ref{lem:transl:soundness}~(ii), 
  Lem.~\ref{lem2:outsourcing}~(a),
  Lem.~\ref{lem:collapse:pop} and Lem.~\ref{lem:collapse:sn:cr}:
  \mbox{}
  $k \le \doLow{\atrs}{\rootsc} \le \terprd{\atrs}{\rootsc}$. \qed
\end{proof}

\begin{example}
  We illustrate the translation
  and decision of \daob{} productivity by means of Pascal's triangle, Fig.~\ref{fig:pascal}.
  The translation of the stream function symbols
  is $\afam = \{\trmrep{\sdtrnsl{\astrfun}}\}$
  with $\trmrep{\sdtrnsl{\astrfun}} = \netgate{\ios{-}{-+}}$,
  see page~\pageref{transl:Pascal:f}.
  We calculate $\trnslF{\strcf{P}}{\afam}$, the translation of $\strcf{P}$,
  and reduce it to normal form with respect to $\scolred$:
  \begin{align*}
    \trnslF{\strcf{P}}{\afam} &= \netrec{P}{\netpeb{\netpeb{\netbox{\ios{-}{-+}}{P}}}}%
    \mcolred \netrec{P}{\netbox{\ios{++-}{-+}}{P}}
    \colred \netsrc{\conattop}
  \end{align*}
  Hence $\doLow{\atrs}{\strcf{P}} = \conattop$, and $\strcf{P}$
  is \daobly{} productive and therefore productive.
\end{example}

\section{Examples}

Productivity of all of the following examples
is recognized fully automatically
by a Haskell implementation of our decision algorithm for data-oblivious productivity.
The tool and a number of examples can be found at:
\begin{center}
  {\small\texttt{http://infinity.few.vu.nl/productivity}}.
\end{center}
In Subsections~\ref{sec:examples:subsec:pascal}--%
  \ref{sec:examples:subsec:example:traces} 
below, we give possible input-representations as well as the tool-output
for the examples of stream specifications in Section~\ref{sec:class}
and for an example of a stream function specification
in Section~\ref{sec:quantitative}.

\subsection{Example in Fig.~\ref{fig:pascal} 
            on Page~\pageref{fig:pascal}}
  \label{sec:examples:subsec:pascal}            

\noindent
For applying the automated productivity prover to the flat stream
specification in Fig.~\ref{fig:pascal} on page~\ref{fig:pascal}
of the stream of rows in Pascal's triangle, one can use the input:
{\small\begin{itemize}\item []
\begin{verbatim}
Signature(
  P : stream(nat),
  0 : nat,
  f : stream(nat) -> stream(nat),
  a : nat -> nat -> nat,
  s : nat -> nat
)
P = 0:s(0):f(P)
f(s(x):y:sigma) = a(s(x),y):f(y:sigma)
f(0:sigma) = 0:s(0):f(sigma)

a(s(x),y) = s(a(x,y))
a(0,y) = y
\end{verbatim}
\end{itemize}}
\noindent
On this input the automated productivity prover produces the following output:
\begin{shaded}

\noindent The automated productivity prover has been applied to:
{\small\begin{itemize}\item []
\begin{verbatim}
Signature(
  -- stream symbols -- 
  P : stream(nat),
  f : stream(nat) -> stream(nat),

  -- data symbols -- 
  0 : nat,
  a : nat -> nat -> nat,
  s : nat -> nat
)

-- stream layer -- 
P = 0:s(0):f(P)
f(s(x):y:sigma) = a(s(x),y):f(y:sigma)
f(0:sigma) = 0:s(0):f(sigma)

-- data layer -- 
a(s(x),y) = s(a(x,y))
a(0,y) = y
\end{verbatim}
\end{itemize}}
\vspace{1.5ex}\noindent
Termination of the data layer has been proven automatically.\\[1.5ex]
\noindent The function symbol $\msf{f}$ is flat, we can compute the precise data-oblivious lower bound:
\begin{align*}
\trnsl{\msf{f}} &= \netgate{\netbox{\trnsli{\msf{f}}{\argnone,0}}{\netsrc{\numzer}},\;\trnsli{\msf{f}}{1,0}}\\
\trnsli{\msf{f}}{1,0} &= \mu \trnslvar{f}{1}{0}.\wedge \begin{cases}\ioin \ioin \ioout \wedge \begin{cases}\mu \trnslvar{f}{1}{1}.\wedge \begin{cases}\ioin \ioout \wedge \begin{cases}\trnslvar{f}{1}{1}\end{cases}\\\ioout \ioout \wedge \begin{cases}\trnslvar{f}{1}{0}\end{cases}\end{cases}\end{cases}\\\ioin \ioout \ioout \wedge \begin{cases}\trnslvar{f}{1}{0}\end{cases}\end{cases}\\
 &= \ios{-}{-+}\\
\trnsli{\msf{f}}{\argnone,0} &= \mu \trnslvar{f}{\argnone}{0}.\wedge \begin{cases}\ioout \trnslvar{f}{\argnone}{0}\\\ioout \ioout \trnslvar{f}{\argnone}{0}\end{cases}\\
 &= \ios{}{+}\\
\end{align*}

\noindent $\msf{P}$ depends only on flat stream functions, we can \textbf{decide data-oblivious productivity}.

\noindent We translate $\mathsf{P}$ into a pebbleflow net and collapse it to a source:\textcolor{white}{\noteSimplify{}}
\begin{align*}
  \trnsl{\mathsf{P}} &= \netrec{P}{{\netpeb{\netpeb{\funap{\trnsl{\msf{f}}}{P}}}}}\\
                                    &=%
                                      \netrec{P}{\netpeb{\netpeb{\netbox{\iosqcns{\ioin}{\iosqcyc{\iosqcns{\ioin}{\iosqcns{\ioout}{\nix}}}}}{P}}}}\\
  &\mcolred \netrec{P}{\netbox{\iosqcns{\ioout}{\iosqcyc{\iosqcns{\ioin}{\iosqcns{\ioout}{\nix}}}}}{\netpeb{\netbox{\iosqcns{\ioin}{\iosqcyc{\iosqcns{\ioin}{\iosqcns{\ioout}{\nix}}}}}{P}}}}\\
  &\mcolred \netrec{P}{\netbox{\iosqcns{\ioout}{\iosqcyc{\iosqcns{\ioin}{\iosqcns{\ioout}{\nix}}}}}{\netbox{\iosqcns{\ioout}{\iosqcyc{\iosqcns{\ioin}{\iosqcns{\ioout}{\nix}}}}}{\netbox{\iosqcns{\ioin}{\iosqcyc{\iosqcns{\ioin}{\iosqcns{\ioout}{\nix}}}}}{P}}}}\\
  &\mcolred \netrec{P}{\netbox{\iosqcns{\ioout}{\iosqcyc{\iosqcns{\ioout}{\iosqcns{\ioin}{\nix}}}}}{\netbox{\iosqcns{\ioin}{\iosqcyc{\iosqcns{\ioin}{\iosqcns{\ioout}{\nix}}}}}{P}}}\\
  &\mcolred \netrec{P}{\netbox{\iosqcns{\ioout}{\iosqcns{\ioout}{\iosqcns{\ioin}{\iosqcyc{\iosqcns{\ioin}{\iosqcns{\ioout}{\nix}}}}}}}{P}}\\
  &\mcolred \netsrc{\conattop}\end{align*}
The specification of $\msf{P}$ is productive.
\end{shaded}

\subsection{Example in Fig.~\ref{fig:ternary_morse_flat}
            on Page~\pageref{fig:ternary_morse_flat}}

\noindent
For applying the automated productivity prover to the flat stream
specification in Fig.~\ref{fig:ternary_morse_flat} 
on page~\pageref{fig:ternary_morse_flat}
of the ternary Thue-Morse sequence, one can use the input:
{\small\begin{itemize}\item []
\begin{verbatim}
Signature(
  Q, Qprime : stream(char),
  f : stream(char) -> stream(char),
  a, b, c: char
)
Q = a:Qprime
Qprime = b:c:f(Qprime)
f(a:sigma) = a:b:c:f(sigma)
f(b:sigma) = a:c:f(sigma)
f(c:sigma) = b:f(sigma)
\end{verbatim}
\end{itemize}}
\noindent
On this input the automated productivity prover produces the following output:
\begin{shaded}
\noindent The automated productivity prover has been applied to:
{\small\begin{itemize}\item []
\begin{verbatim}
Signature(
  -- stream symbols -- 
  Q : stream(char),
  Qprime : stream(char),
  f : stream(char) -> stream(char),

  -- data symbols -- 
  a : char,
  b : char,
  c : char
)

-- stream layer -- 
Q = a:Qprime
Qprime = b:c:f(Qprime)
f(a:sigma) = a:b:c:f(sigma)
f(b:sigma) = a:c:f(sigma)
f(c:sigma) = b:f(sigma)

-- data layer -- 

\end{verbatim}
\end{itemize}}
\noindent The function symbol $\msf{f}$ is flat, we can compute the precise data-oblivious lower bound:
\begin{align*}
\trnsl{\msf{f}} &= \netgate{\netbox{\trnsli{\msf{f}}{\argnone,0}}{\netsrc{\numzer}},\;\trnsli{\msf{f}}{1,0}}\\
\trnsli{\msf{f}}{1,0} &= \mu \trnslvar{f}{1}{0}.\wedge \begin{cases}\ioin \ioout \ioout \ioout \wedge \begin{cases}\trnslvar{f}{1}{0}\end{cases}\\\ioin \ioout \ioout \wedge \begin{cases}\trnslvar{f}{1}{0}\end{cases}\\\ioin \ioout \wedge \begin{cases}\trnslvar{f}{1}{0}\end{cases}\end{cases}\\
 &= \ios{}{-+}\\
\trnsli{\msf{f}}{\argnone,0} &= \mu \trnslvar{f}{\argnone}{0}.\wedge \begin{cases}\ioout \ioout \ioout \trnslvar{f}{\argnone}{0}\\\ioout \ioout \trnslvar{f}{\argnone}{0}\\\ioout \trnslvar{f}{\argnone}{0}\end{cases}\\
 &= \ios{}{+}\\
\end{align*}

\noindent $\msf{Q}$ depends only on flat stream functions, we can \textbf{decide data-oblivious productivity}.

\noindent We translate $\mathsf{Q}$ into a pebbleflow net and collapse it to a source:\textcolor{white}{\noteSimplify{}}
\begin{align*}
  \trnsl{\mathsf{Q}} &= \netrec{Q}{{\netpeb{\netrec{Qprime}{{\netpeb{\netpeb{\funap{\trnsl{\msf{f}}}{Qprime}}}}}}}}\\
                                    &=\footnotemark[1] \netrec{Q}{\netpeb{\netrec{Qprime}{\netpeb{\netpeb{\netbox{\iosqcyc{\iosqcns{\ioin}{\iosqcns{\ioout}{\nix}}}}{Qprime}}}}}}\\
  &\mcolred \netrec{Q}{\netbox{\iosqcns{\ioout}{\iosqcyc{\iosqcns{\ioin}{\iosqcns{\ioout}{\nix}}}}}{\netrec{Qprime}{\netpeb{\netpeb{\netbox{\iosqcyc{\iosqcns{\ioin}{\iosqcns{\ioout}{\nix}}}}{Qprime}}}}}}\\
  &\mcolred \netrec{Q}{\netbox{\iosqcns{\ioout}{\iosqcyc{\iosqcns{\ioin}{\iosqcns{\ioout}{\nix}}}}}{\netrec{Qprime}{\netbox{\iosqcns{\ioout}{\iosqcyc{\iosqcns{\ioin}{\iosqcns{\ioout}{\nix}}}}}{\netpeb{\netbox{\iosqcyc{\iosqcns{\ioin}{\iosqcns{\ioout}{\nix}}}}{Qprime}}}}}}\\
  &\mcolred \netrec{Q}{\netbox{\iosqcns{\ioout}{\iosqcyc{\iosqcns{\ioin}{\iosqcns{\ioout}{\nix}}}}}{\netrec{Qprime}{\netbox{\iosqcns{\ioout}{\iosqcyc{\iosqcns{\ioin}{\iosqcns{\ioout}{\nix}}}}}{\netbox{\iosqcns{\ioout}{\iosqcyc{\iosqcns{\ioin}{\iosqcns{\ioout}{\nix}}}}}{\netbox{\iosqcyc{\iosqcns{\ioin}{\iosqcns{\ioout}{\nix}}}}{Qprime}}}}}}\\
  &\mcolred \netrec{Q}{\netbox{\iosqcns{\ioout}{\iosqcyc{\iosqcns{\ioin}{\iosqcns{\ioout}{\nix}}}}}{\netrec{Qprime}{\netbox{\iosqcns{\ioout}{\iosqcyc{\iosqcns{\ioout}{\iosqcns{\ioin}{\nix}}}}}{\netbox{\iosqcyc{\iosqcns{\ioin}{\iosqcns{\ioout}{\nix}}}}{Qprime}}}}}\\
  &\mcolred \netrec{Q}{\netbox{\iosqcns{\ioout}{\iosqcyc{\iosqcns{\ioin}{\iosqcns{\ioout}{\nix}}}}}{\netrec{Qprime}{\netbox{\iosqcns{\ioout}{\iosqcyc{\iosqcns{\ioout}{\iosqcns{\ioin}{\nix}}}}}{Qprime}}}}\\
  &\mcolred \netrec{Q}{\netbox{\iosqcns{\ioout}{\iosqcyc{\iosqcns{\ioin}{\iosqcns{\ioout}{\nix}}}}}{\netsrc{\conattop}}}\\
  &\mcolred \netrec{Q}{\netsrc{\conattop}}\\
  &\mcolred \netsrc{\conattop}\end{align*}
The specification of $\msf{Q}$ is productive.\\[1.5ex]

\noindent $\msf{Qprime}$ depends only on flat stream functions, we can \textbf{decide data-oblivious productivity}.

\noindent We translate $\mathsf{Qprime}$ into a pebbleflow net and collapse it to a source:\textcolor{white}{\noteSimplify{}}
\begin{align*}
  \trnsl{\mathsf{Qprime}} &= \netrec{Qprime}{{\netpeb{\netpeb{\funap{\trnsl{\msf{f}}}{Qprime}}}}}\\
                                    &=%
                                      \netrec{Qprime}{\netpeb{\netpeb{\netbox{\iosqcyc{\iosqcns{\ioin}{\iosqcns{\ioout}{\nix}}}}{Qprime}}}}\\
  &\mcolred \netrec{Qprime}{\netbox{\iosqcns{\ioout}{\iosqcyc{\iosqcns{\ioin}{\iosqcns{\ioout}{\nix}}}}}{\netpeb{\netbox{\iosqcyc{\iosqcns{\ioin}{\iosqcns{\ioout}{\nix}}}}{Qprime}}}}\\
  &\mcolred \netrec{Qprime}{\netbox{\iosqcns{\ioout}{\iosqcyc{\iosqcns{\ioin}{\iosqcns{\ioout}{\nix}}}}}{\netbox{\iosqcns{\ioout}{\iosqcyc{\iosqcns{\ioin}{\iosqcns{\ioout}{\nix}}}}}{\netbox{\iosqcyc{\iosqcns{\ioin}{\iosqcns{\ioout}{\nix}}}}{Qprime}}}}\\
  &\mcolred \netrec{Qprime}{\netbox{\iosqcns{\ioout}{\iosqcyc{\iosqcns{\ioout}{\iosqcns{\ioin}{\nix}}}}}{\netbox{\iosqcyc{\iosqcns{\ioin}{\iosqcns{\ioout}{\nix}}}}{Qprime}}}\\
  &\mcolred \netrec{Qprime}{\netbox{\iosqcns{\ioout}{\iosqcyc{\iosqcns{\ioout}{\iosqcns{\ioin}{\nix}}}}}{Qprime}}\\
  &\mcolred \netsrc{\conattop}\end{align*}
The specification of $\msf{Qprime}$ is productive.
\end{shaded}

\subsection{Example in Fig.~\ref{fig:ternary_morse_pure}
            on Page~\pageref{fig:ternary_morse_pure}}
\noindent
For applying the automated productivity prover to the pure stream
specification in Fig.~\ref{fig:ternary_morse_pure} 
on page~\pageref{fig:ternary_morse_pure}
of the ternary Thue-Morse sequence, one can use the input:
{\small\begin{itemize}\item []
\begin{verbatim}
Signature(
  Q : stream(char),
  M : stream(bit),
  zip : stream(x) -> stream(x) -> stream(x),
  inv : stream(bit) -> stream(bit),
  tail : stream(x) -> stream(x),
  diff : stream(bit) -> stream(char),
  i : bit -> bit,
  X : bit -> bit -> char,
  0, 1 : bit, 
  a,b,c : char
)
Q = diff(M)
M = 0:zip(inv(M),tail(M))
zip(x:s,t) = x:zip(t,s)
inv(x:s) = i(x):inv(s)
tail(x:s) = s
diff(x:y:s) = X(x,y):diff(y:s)

i(0) = 1
i(1) = 0
X(0,0) = b
X(0,1) = a
X(1,0) = c
X(1,1) = b
\end{verbatim}
\end{itemize}}
\noindent
The automated productivity prover then gives the following output:
\begin{shaded}
\noindent The automated productivity prover has been applied to:
{\small\begin{itemize}\item []
\begin{verbatim}
Signature(
  -- stream symbols -- 
  Q : stream(char),
  M : stream(bit),
  zip : stream(x) -> stream(x) -> stream(x),
  inv : stream(bit) -> stream(bit),
  tail : stream(x) -> stream(x),
  diff : stream(bit) -> stream(char),

  -- data symbols -- 
  i : bit -> bit,
  X : bit -> bit -> char,
  0 : bit,
  1 : bit,
  a : char,
  b : char,
  c : char
)

-- stream layer -- 
Q = diff(M)
M = 0:zip(inv(M),tail(M))
zip(x:s,t) = x:zip(t,s)
inv(x:s) = i(x):inv(s)
tail(x:s) = s
diff(x:y:s) = X(x,y):diff(y:s)

-- data layer -- 
i(0) = 1
i(1) = 0
X(0,0) = b
X(0,1) = a
X(1,0) = c
X(1,1) = b
\end{verbatim}
\end{itemize}}
\vspace{1.5ex}\noindent
Termination of the data layer has been proven automatically.\\[1.5ex]
\noindent The function symbol $\msf{zip}$ is pure, we can compute its precise production modulus:
\begin{align*}
\trnsl{\msf{zip}} &= \netgate{\netbox{\trnsli{\msf{zip}}{\argnone,0}}{\netsrc{\numzer}},\;\trnsli{\msf{zip}}{1,0},\;\trnsli{\msf{zip}}{2,0}}\\
\trnsli{\msf{zip}}{1,0} &= \mu \trnslvar{zip}{1}{0}.\wedge \begin{cases}\ioin \ioout \wedge \begin{cases}\mu \trnslvar{zip}{2}{0}.\wedge \begin{cases}\ioout \wedge \begin{cases}\trnslvar{zip}{1}{0}\end{cases}\end{cases}\end{cases}\end{cases}\\
 &= \ios{}{-++}\\
\trnsli{\msf{zip}}{2,0} &= \mu \trnslvar{zip}{2}{0}.\wedge \begin{cases}\ioout \wedge \begin{cases}\mu \trnslvar{zip}{1}{0}.\wedge \begin{cases}\ioin \ioout \wedge \begin{cases}\trnslvar{zip}{2}{0}\end{cases}\end{cases}\end{cases}\end{cases}\\
 &= \ios{}{+-+}\\
\trnsli{\msf{zip}}{\argnone,0} &= \mu \trnslvar{zip}{\argnone}{0}.\wedge \begin{cases}\ioout \trnslvar{zip}{\argnone}{0}\end{cases}\\
 &= \ios{}{+}\\
\end{align*}
\noindent The function symbol $\msf{inv}$ is pure, we can compute its precise production modulus:
\begin{align*}
\trnsl{\msf{inv}} &= \netgate{\netbox{\trnsli{\msf{inv}}{\argnone,0}}{\netsrc{\numzer}},\;\trnsli{\msf{inv}}{1,0}}\\
\trnsli{\msf{inv}}{1,0} &= \mu \trnslvar{inv}{1}{0}.\wedge \begin{cases}\ioin \ioout \wedge \begin{cases}\trnslvar{inv}{1}{0}\end{cases}\end{cases}\\
 &= \ios{}{-+}\\
\trnsli{\msf{inv}}{\argnone,0} &= \mu \trnslvar{inv}{\argnone}{0}.\wedge \begin{cases}\ioout \trnslvar{inv}{\argnone}{0}\end{cases}\\
 &= \ios{}{+}\\
\end{align*}
\noindent The function symbol $\msf{tail}$ is pure, we can compute its precise production modulus:
\begin{align*}
\trnsl{\msf{tail}} &= \netgate{\netbox{\trnsli{\msf{tail}}{\argnone,0}}{\netsrc{\numzer}},\;\trnsli{\msf{tail}}{1,0}}\\
\trnsli{\msf{tail}}{1,0} &= \mu \trnslvar{tail}{1}{0}.\wedge \begin{cases}\ioin \mu x.\ioin \ioout x\end{cases}\\
 &= \ios{-}{-+}\\
\trnsli{\msf{tail}}{\argnone,0} &= \mu \trnslvar{tail}{\argnone}{0}.\wedge \begin{cases}\mu x.\ioout x\end{cases}\\
 &= \ios{}{+}\\
\end{align*}
\noindent The function symbol $\msf{diff}$ is pure, we can compute its precise production modulus:
\begin{align*}
\trnsl{\msf{diff}} &= \netgate{\netbox{\trnsli{\msf{diff}}{\argnone,0}}{\netsrc{\numzer}},\;\trnsli{\msf{diff}}{1,0}}\\
\trnsli{\msf{diff}}{1,0} &= \mu \trnslvar{diff}{1}{0}.\wedge \begin{cases}\ioin \ioin \ioout \wedge \begin{cases}\mu \trnslvar{diff}{1}{1}.\wedge \begin{cases}\ioin \ioout \wedge \begin{cases}\trnslvar{diff}{1}{1}\end{cases}\end{cases}\end{cases}\end{cases}\\
 &= \ios{-}{-+}\\
\trnsli{\msf{diff}}{\argnone,0} &= \mu \trnslvar{diff}{\argnone}{0}.\wedge \begin{cases}\ioout \trnslvar{diff}{\argnone}{0}\end{cases}\\
 &= \ios{}{+}\\
\end{align*}

\noindent $\msf{Q}$ depends only on pure stream functions, we can \textbf{decide productivity}.

\noindent We translate $\mathsf{Q}$ into a pebbleflow net and collapse it to a source:\textcolor{white}{\noteSimplify{}}
\begin{align*}
  \trnsl{\mathsf{Q}} &= \netrec{Q}{{\funap{\trnsl{\msf{diff}}}{\netrec{M}{{\netpeb{\funap{\trnsl{\msf{zip}}}{\funap{\trnsl{\msf{inv}}}{M},\; \funap{\trnsl{\msf{tail}}}{M}}}}}}}}\\
                                    &=%
                                      \netrec{Q}{\netbox{\iosqcns{\ioin}{\iosqcyc{\iosqcns{\ioin}{\iosqcns{\ioout}{\nix}}}}}{\netrec{M}{\netpeb{\netmeet{\netbox{\iosqcyc{\iosqcns{\ioin}{\iosqcns{\ioout}{\iosqcns{\ioout}{\nix}}}}}{\netbox{\iosqcyc{\iosqcns{\ioin}{\iosqcns{\ioout}{\nix}}}}{M}}}{\netbox{\iosqcyc{\iosqcns{\ioout}{\iosqcns{\ioin}{\iosqcns{\ioout}{\nix}}}}}{\netbox{\iosqcns{\ioin}{\iosqcyc{\iosqcns{\ioin}{\iosqcns{\ioout}{\nix}}}}}{M}}}}}}}\\
  &\mcolred \netrec{Q}{\netbox{\iosqcns{\ioin}{\iosqcyc{\iosqcns{\ioin}{\iosqcns{\ioout}{\nix}}}}}{\netrec{M}{\netpeb{\netmeet{\netbox{\iosqcyc{\iosqcns{\ioin}{\iosqcns{\ioout}{\iosqcns{\ioout}{\nix}}}}}{M}}{\netbox{\iosqcns{\ioout}{\iosqcns{\ioin}{\iosqcyc{\iosqcns{\ioin}{\iosqcns{\ioout}{\iosqcns{\ioout}{\nix}}}}}}}{M}}}}}}\\
  &\mcolred \netrec{Q}{\netbox{\iosqcns{\ioin}{\iosqcyc{\iosqcns{\ioin}{\iosqcns{\ioout}{\nix}}}}}{\netrec{M}{\netbox{\iosqcns{\ioout}{\iosqcyc{\iosqcns{\ioin}{\iosqcns{\ioout}{\nix}}}}}{\netmeet{\netbox{\iosqcyc{\iosqcns{\ioin}{\iosqcns{\ioout}{\iosqcns{\ioout}{\nix}}}}}{M}}{\netbox{\iosqcns{\ioout}{\iosqcns{\ioin}{\iosqcyc{\iosqcns{\ioin}{\iosqcns{\ioout}{\iosqcns{\ioout}{\nix}}}}}}}{M}}}}}}\\
  &\mcolred \netrec{Q}{\netbox{\iosqcns{\ioin}{\iosqcyc{\iosqcns{\ioin}{\iosqcns{\ioout}{\nix}}}}}{\netrec{M}{\netmeet{\netbox{\iosqcns{\ioout}{\iosqcyc{\iosqcns{\ioin}{\iosqcns{\ioout}{\nix}}}}}{\netbox{\iosqcyc{\iosqcns{\ioin}{\iosqcns{\ioout}{\iosqcns{\ioout}{\nix}}}}}{M}}}{\netbox{\iosqcns{\ioout}{\iosqcyc{\iosqcns{\ioin}{\iosqcns{\ioout}{\nix}}}}}{\netbox{\iosqcns{\ioout}{\iosqcns{\ioin}{\iosqcyc{\iosqcns{\ioin}{\iosqcns{\ioout}{\iosqcns{\ioout}{\nix}}}}}}}{M}}}}}}\\
  &\mcolred \netrec{Q}{\netbox{\iosqcns{\ioin}{\iosqcyc{\iosqcns{\ioin}{\iosqcns{\ioout}{\nix}}}}}{\netrec{M}{\netmeet{\netbox{\iosqcyc{\iosqcns{\ioout}{\iosqcns{\ioin}{\iosqcns{\ioout}{\nix}}}}}{M}}{\netbox{\iosqcns{\ioout}{\iosqcns{\ioout}{\iosqcns{\ioin}{\iosqcyc{\iosqcns{\ioin}{\iosqcns{\ioout}{\iosqcns{\ioout}{\nix}}}}}}}}{M}}}}}\\
  &\mcolred \netrec{Q}{\netbox{\iosqcns{\ioin}{\iosqcyc{\iosqcns{\ioin}{\iosqcns{\ioout}{\nix}}}}}{\netmeet{\netrec{M}{\netbox{\iosqcyc{\iosqcns{\ioout}{\iosqcns{\ioin}{\iosqcns{\ioout}{\nix}}}}}{M}}}{\netrec{M}{\netbox{\iosqcns{\ioout}{\iosqcns{\ioout}{\iosqcns{\ioin}{\iosqcyc{\iosqcns{\ioin}{\iosqcns{\ioout}{\iosqcns{\ioout}{\nix}}}}}}}}{M}}}}}\\
  &\mcolred \netrec{Q}{\netmeet{\netbox{\iosqcns{\ioin}{\iosqcyc{\iosqcns{\ioin}{\iosqcns{\ioout}{\nix}}}}}{\netrec{M}{\netbox{\iosqcyc{\iosqcns{\ioout}{\iosqcns{\ioin}{\iosqcns{\ioout}{\nix}}}}}{M}}}}{\netbox{\iosqcns{\ioin}{\iosqcyc{\iosqcns{\ioin}{\iosqcns{\ioout}{\nix}}}}}{\netrec{M}{\netbox{\iosqcns{\ioout}{\iosqcns{\ioout}{\iosqcns{\ioin}{\iosqcyc{\iosqcns{\ioin}{\iosqcns{\ioout}{\iosqcns{\ioout}{\nix}}}}}}}}{M}}}}}\\
  &\mcolred \netmeet{\netrec{Q}{\netbox{\iosqcns{\ioin}{\iosqcyc{\iosqcns{\ioin}{\iosqcns{\ioout}{\nix}}}}}{\netrec{M}{\netbox{\iosqcyc{\iosqcns{\ioout}{\iosqcns{\ioin}{\iosqcns{\ioout}{\nix}}}}}{M}}}}}{\netrec{Q}{\netbox{\iosqcns{\ioin}{\iosqcyc{\iosqcns{\ioin}{\iosqcns{\ioout}{\nix}}}}}{\netrec{M}{\netbox{\iosqcns{\ioout}{\iosqcns{\ioout}{\iosqcns{\ioin}{\iosqcyc{\iosqcns{\ioin}{\iosqcns{\ioout}{\iosqcns{\ioout}{\nix}}}}}}}}{M}}}}}\\
  &\mcolred \netmeet{\netrec{Q}{\netbox{\iosqcns{\ioin}{\iosqcyc{\iosqcns{\ioin}{\iosqcns{\ioout}{\nix}}}}}{\netsrc{\conattop}}}}{\netrec{Q}{\netbox{\iosqcns{\ioin}{\iosqcyc{\iosqcns{\ioin}{\iosqcns{\ioout}{\nix}}}}}{\netsrc{\conattop}}}}\\
  &\mcolred \netmeet{\netrec{Q}{\netsrc{\conattop}}}{\netrec{Q}{\netsrc{\conattop}}}\\
  &\mcolred \netmeet{\netsrc{\conattop}}{\netsrc{\conattop}}\\
  &\mcolred \netsrc{\conattop}\end{align*}
The specification of $\msf{Q}$ is productive.\\[1.5ex]

\noindent $\msf{M}$ depends only on pure stream functions, we can \textbf{decide productivity}.

\noindent We translate $\mathsf{M}$ into a pebbleflow net and collapse it to a source:\textcolor{white}{\noteSimplify{}}
\begin{align*}
  \trnsl{\mathsf{M}} &= \netrec{M}{{\netpeb{\funap{\trnsl{\msf{zip}}}{\funap{\trnsl{\msf{inv}}}{M},\; \funap{\trnsl{\msf{tail}}}{M}}}}}\\
                                    &=%
                                      \netrec{M}{\netpeb{\netmeet{\netbox{\iosqcyc{\iosqcns{\ioin}{\iosqcns{\ioout}{\iosqcns{\ioout}{\nix}}}}}{\netbox{\iosqcyc{\iosqcns{\ioin}{\iosqcns{\ioout}{\nix}}}}{M}}}{\netbox{\iosqcyc{\iosqcns{\ioout}{\iosqcns{\ioin}{\iosqcns{\ioout}{\nix}}}}}{\netbox{\iosqcns{\ioin}{\iosqcyc{\iosqcns{\ioin}{\iosqcns{\ioout}{\nix}}}}}{M}}}}}\\
  &\mcolred \netrec{M}{\netpeb{\netmeet{\netbox{\iosqcyc{\iosqcns{\ioin}{\iosqcns{\ioout}{\iosqcns{\ioout}{\nix}}}}}{M}}{\netbox{\iosqcns{\ioout}{\iosqcns{\ioin}{\iosqcyc{\iosqcns{\ioin}{\iosqcns{\ioout}{\iosqcns{\ioout}{\nix}}}}}}}{M}}}}\\
  &\mcolred \netrec{M}{\netbox{\iosqcns{\ioout}{\iosqcyc{\iosqcns{\ioin}{\iosqcns{\ioout}{\nix}}}}}{\netmeet{\netbox{\iosqcyc{\iosqcns{\ioin}{\iosqcns{\ioout}{\iosqcns{\ioout}{\nix}}}}}{M}}{\netbox{\iosqcns{\ioout}{\iosqcns{\ioin}{\iosqcyc{\iosqcns{\ioin}{\iosqcns{\ioout}{\iosqcns{\ioout}{\nix}}}}}}}{M}}}}\\
  &\mcolred \netrec{M}{\netmeet{\netbox{\iosqcns{\ioout}{\iosqcyc{\iosqcns{\ioin}{\iosqcns{\ioout}{\nix}}}}}{\netbox{\iosqcyc{\iosqcns{\ioin}{\iosqcns{\ioout}{\iosqcns{\ioout}{\nix}}}}}{M}}}{\netbox{\iosqcns{\ioout}{\iosqcyc{\iosqcns{\ioin}{\iosqcns{\ioout}{\nix}}}}}{\netbox{\iosqcns{\ioout}{\iosqcns{\ioin}{\iosqcyc{\iosqcns{\ioin}{\iosqcns{\ioout}{\iosqcns{\ioout}{\nix}}}}}}}{M}}}}\\
  &\mcolred \netrec{M}{\netmeet{\netbox{\iosqcyc{\iosqcns{\ioout}{\iosqcns{\ioin}{\iosqcns{\ioout}{\nix}}}}}{M}}{\netbox{\iosqcns{\ioout}{\iosqcns{\ioout}{\iosqcns{\ioin}{\iosqcyc{\iosqcns{\ioin}{\iosqcns{\ioout}{\iosqcns{\ioout}{\nix}}}}}}}}{M}}}\\
  &\mcolred \netmeet{\netrec{M}{\netbox{\iosqcyc{\iosqcns{\ioout}{\iosqcns{\ioin}{\iosqcns{\ioout}{\nix}}}}}{M}}}{\netrec{M}{\netbox{\iosqcns{\ioout}{\iosqcns{\ioout}{\iosqcns{\ioin}{\iosqcyc{\iosqcns{\ioin}{\iosqcns{\ioout}{\iosqcns{\ioout}{\nix}}}}}}}}{M}}}\\
  &\mcolred \netmeet{\netsrc{\conattop}}{\netsrc{\conattop}}\\
  &\mcolred \netsrc{\conattop}\end{align*}
The specification of $\msf{M}$ is productive. 
\end{shaded}

\subsection{Example in Fig.~\ref{fig:morse_dol}
            on Page~\pageref{fig:morse_dol}}

For applying the automated productivity prover to the pure stream
specification in Fig.~\ref{fig:morse_dol} 
on page~\pageref{fig:morse_dol}
of the Thue--Morse sequence using the D0L-system
$ 0 \mapsto 01 $, $1 \mapsto 10$, one can use the input:
{\small\begin{itemize}\item []
\begin{verbatim}
Signature(
  M : stream(bit),
  Mprime : stream(bit),
  h : stream(bit) -> stream(bit),
  0, 1 : bit
)
M = 0:Mprime
Mprime = 1:h(Mprime)
h(0:sigma) = 0:1:h(sigma)
h(1:sigma) = 1:0:h(sigma)
\end{verbatim}
\end{itemize}}
\noindent The automated productivity prover then gives the following output:
\begin{shaded}
\noindent The automated productivity prover has been applied to:
{\small\begin{itemize}\item []
\begin{verbatim}
Signature(
  -- stream symbols -- 
  M : stream(bit),
  Mprime : stream(bit),
  h : stream(bit) -> stream(bit),

  -- data symbols -- 
  0 : bit,
  1 : bit
)

-- stream layer -- 
M = 0:Mprime
Mprime = 1:h(Mprime)
h(0:sigma) = 0:1:h(sigma)
h(1:sigma) = 1:0:h(sigma)

-- data layer -- 

\end{verbatim}
\end{itemize}}
\noindent The function symbol $\msf{h}$ is pure, we can compute its precise production modulus:
\begin{align*}
\trnsl{\msf{h}} &= \netgate{\netbox{\trnsli{\msf{h}}{\argnone,0}}{\netsrc{\numzer}},\;\trnsli{\msf{h}}{1,0}}\\
\trnsli{\msf{h}}{1,0} &= \mu \trnslvar{h}{1}{0}.\wedge \begin{cases}\ioin \ioout \ioout \wedge \begin{cases}\trnslvar{h}{1}{0}\end{cases}\\\ioin \ioout \ioout \wedge \begin{cases}\trnslvar{h}{1}{0}\end{cases}\end{cases}\\
 &= \ios{}{-++}\\
\trnsli{\msf{h}}{\argnone,0} &= \mu \trnslvar{h}{\argnone}{0}.\wedge \begin{cases}\ioout \ioout \trnslvar{h}{\argnone}{0}\\\ioout \ioout \trnslvar{h}{\argnone}{0}\end{cases}\\
 &= \ios{}{+}\\
\end{align*}

\noindent $\msf{M}$ depends only on pure stream functions, we can \textbf{decide productivity}.

\noindent We translate $\mathsf{M}$ into a pebbleflow net and collapse it to a source:\textcolor{white}{\noteSimplify{}}
\begin{align*}
  \trnsl{\mathsf{M}} &= \netrec{M}{{\netpeb{\netrec{Mprime}{{\netpeb{\funap{\trnsl{\msf{h}}}{Mprime}}}}}}}\\
                                    &=%
                                      \netrec{M}{\netpeb{\netrec{Mprime}{\netpeb{\netbox{\iosqcyc{\iosqcns{\ioin}{\iosqcns{\ioout}{\iosqcns{\ioout}{\nix}}}}}{Mprime}}}}}\\
  &\mcolred \netrec{M}{\netbox{\iosqcns{\ioout}{\iosqcyc{\iosqcns{\ioin}{\iosqcns{\ioout}{\nix}}}}}{\netrec{Mprime}{\netpeb{\netbox{\iosqcyc{\iosqcns{\ioin}{\iosqcns{\ioout}{\iosqcns{\ioout}{\nix}}}}}{Mprime}}}}}\\
  &\mcolred \netrec{M}{\netbox{\iosqcns{\ioout}{\iosqcyc{\iosqcns{\ioin}{\iosqcns{\ioout}{\nix}}}}}{\netrec{Mprime}{\netbox{\iosqcns{\ioout}{\iosqcyc{\iosqcns{\ioin}{\iosqcns{\ioout}{\nix}}}}}{\netbox{\iosqcyc{\iosqcns{\ioin}{\iosqcns{\ioout}{\iosqcns{\ioout}{\nix}}}}}{Mprime}}}}}\\
  &\mcolred \netrec{M}{\netbox{\iosqcns{\ioout}{\iosqcyc{\iosqcns{\ioin}{\iosqcns{\ioout}{\nix}}}}}{\netrec{Mprime}{\netbox{\iosqcyc{\iosqcns{\ioout}{\iosqcns{\ioin}{\iosqcns{\ioout}{\nix}}}}}{Mprime}}}}\\
  &\mcolred \netrec{M}{\netbox{\iosqcns{\ioout}{\iosqcyc{\iosqcns{\ioin}{\iosqcns{\ioout}{\nix}}}}}{\netsrc{\conattop}}}\\
  &\mcolred \netrec{M}{\netsrc{\conattop}}\\
  &\mcolred \netsrc{\conattop}\end{align*}
The specification of $\msf{M}$ is productive.\\[1.5ex]

\noindent $\msf{Mprime}$ depends only on pure stream functions, we can \textbf{decide productivity}.

\noindent We translate $\mathsf{Mprime}$ into a pebbleflow net and collapse it to a source:\textcolor{white}{\noteSimplify{}}
\begin{align*}
  \trnsl{\mathsf{Mprime}} &= \netrec{Mprime}{{\netpeb{\funap{\trnsl{\msf{h}}}{Mprime}}}}\\
                                    &=%
                                      \netrec{Mprime}{\netpeb{\netbox{\iosqcyc{\iosqcns{\ioin}{\iosqcns{\ioout}{\iosqcns{\ioout}{\nix}}}}}{Mprime}}}\\
  &\mcolred \netrec{Mprime}{\netbox{\iosqcns{\ioout}{\iosqcyc{\iosqcns{\ioin}{\iosqcns{\ioout}{\nix}}}}}{\netbox{\iosqcyc{\iosqcns{\ioin}{\iosqcns{\ioout}{\iosqcns{\ioout}{\nix}}}}}{Mprime}}}\\
  &\mcolred \netrec{Mprime}{\netbox{\iosqcyc{\iosqcns{\ioout}{\iosqcns{\ioin}{\iosqcns{\ioout}{\nix}}}}}{Mprime}}\\
  &\mcolred \netsrc{\conattop}\end{align*}
The specification of $\msf{Mprime}$ is productive.
\end{shaded}

\subsection{Example in Fig.~\ref{fig:convolution} 
            on Page~\pageref{fig:convolution}}

For applying the automated productivity prover to the pure stream
specification in Fig.~\ref{fig:convolution} 
on page~\pageref{fig:convolution}
of the Thue--Morse sequence using the D0L-system
$ 0 \mapsto 01 $, $1 \mapsto 10$, one can use the input:
{\small\begin{itemize}\item []
\begin{verbatim}
Signature(
  nats, ones : stream(nat),
  conv, add : stream(nat) -> stream(nat) -> stream(nat),
  times : stream(nat) -> nat -> stream(nat),
  0 : nat,
  s : nat -> nat,
  a, m : nat -> nat -> nat
)
nats = 0:conv(ones,ones)
ones = s(0):ones
conv(x:sigma,y:tau) = m(x,y):add(times(tau,x),conv(sigma,y:tau))
times(x:sigma,y) = m(x,y):times(sigma,y)
add(x:sigma,y:tau) = a(x,y):add(sigma,tau)

a(0,y) = y
a(s(x),y) = s(a(x,y))
m(0,y) = 0
m(s(x),y) = a(y,m(x,y))
\end{verbatim}
\end{itemize}}
\noindent The automated productivity gives as output:
\begin{shaded}
\noindent The automated productivity prover has been applied to:
{\small\begin{itemize}\item []
\begin{verbatim}
Signature(
  -- stream symbols -- 
  nats : stream(nat),
  ones : stream(nat),
  conv : stream(nat) -> stream(nat) -> stream(nat),
  add : stream(nat) -> stream(nat) -> stream(nat),
  times : stream(nat) -> nat -> stream(nat),

  -- data symbols -- 
  0 : nat,
  s : nat -> nat,
  a : nat -> nat -> nat,
  m : nat -> nat -> nat
)

-- stream layer -- 
nats = 0:conv(ones,ones)
ones = s(0):ones
conv(x:sigma,y:tau) = m(x,y):add(times(tau,x),conv(sigma,y:tau))
times(x:sigma,y) = m(x,y):times(sigma,y)
add(x:sigma,y:tau) = a(x,y):add(sigma,tau)

-- data layer -- 
a(0,y) = y
a(s(x),y) = s(a(x,y))
m(0,y) = 0
m(s(x),y) = a(y,m(x,y))
\end{verbatim}
\end{itemize}}
\vspace{1.5ex}\noindent
Termination of the data layer has been proven automatically.\\[1.5ex]
\noindent The function symbol $\msf{conv}$ is friendly nesting, we can compute a data-oblivious lower bound:
\begin{align*}
\trnsl{\msf{conv}} &= \netgate{\netbox{\trnsli{\msf{conv}}{\argnone,0}}{\netsrc{\numzer}},\;\trnsli{\msf{conv}}{1,0},\;\trnsli{\msf{conv}}{2,0}}\\
\trnsli{\msf{conv}}{1,0} &= \mu \trnslvar{conv}{1}{0}.\wedge \begin{cases}\mu x.\ioin \ioout x\end{cases}\\
 &= \ios{}{-+}\\
\trnsli{\msf{conv}}{2,0} &= \mu \trnslvar{conv}{2}{0}.\wedge \begin{cases}\mu x.\ioin \ioout x\end{cases}\\
 &= \ios{}{-+}\\
\trnsli{\msf{conv}}{\argnone,0} &= \mu \trnslvar{conv}{\argnone}{0}.\wedge \begin{cases}\mu x.\ioin \ioout x\end{cases}\\
 &= \ios{}{-+}\\
\end{align*}
\noindent The function symbol $\msf{add}$ is pure, we can compute its precise production modulus:
\begin{align*}
\trnsl{\msf{add}} &= \netgate{\netbox{\trnsli{\msf{add}}{\argnone,0}}{\netsrc{\numzer}},\;\trnsli{\msf{add}}{1,0},\;\trnsli{\msf{add}}{2,0}}\\
\trnsli{\msf{add}}{1,0} &= \mu \trnslvar{add}{1}{0}.\wedge \begin{cases}\ioin \ioout \wedge \begin{cases}\trnslvar{add}{1}{0}\end{cases}\end{cases}\\
 &= \ios{}{-+}\\
\trnsli{\msf{add}}{2,0} &= \mu \trnslvar{add}{2}{0}.\wedge \begin{cases}\ioin \ioout \wedge \begin{cases}\trnslvar{add}{2}{0}\end{cases}\end{cases}\\
 &= \ios{}{-+}\\
\trnsli{\msf{add}}{\argnone,0} &= \mu \trnslvar{add}{\argnone}{0}.\wedge \begin{cases}\ioout \trnslvar{add}{\argnone}{0}\end{cases}\\
 &= \ios{}{+}\\
\end{align*}
\noindent The function symbol $\msf{times}$ is pure, we can compute its precise production modulus:
\begin{align*}
\trnsl{\msf{times}} &= \netgate{\netbox{\trnsli{\msf{times}}{\argnone,0}}{\netsrc{\numzer}},\;\trnsli{\msf{times}}{1,0}}\\
\trnsli{\msf{times}}{1,0} &= \mu \trnslvar{times}{1}{0}.\wedge \begin{cases}\ioin \ioout \wedge \begin{cases}\trnslvar{times}{1}{0}\end{cases}\end{cases}\\
 &= \ios{}{-+}\\
\trnsli{\msf{times}}{\argnone,0} &= \mu \trnslvar{times}{\argnone}{0}.\wedge \begin{cases}\ioout \trnslvar{times}{\argnone}{0}\end{cases}\\
 &= \ios{}{+}\\
\end{align*}

\noindent $\msf{nats}$ depends only on flat stream functions, we \textbf{try to prove productivity}.

\noindent We translate $\mathsf{nats}$ into a pebbleflow net and collapse it to a source:\textcolor{white}{\noteSimplify{}}
\begin{align*}
  \trnsl{\mathsf{nats}} &= \netrec{nats}{{\netpeb{\funap{\trnsl{\msf{conv}}}{\netrec{ones}{{\netpeb{ones}}},\; \netrec{ones}{{\netpeb{ones}}}}}}}\\
                                    &=%
                                      \netrec{nats}{\netpeb{\netmeet{\netbox{\iosqcyc{\iosqcns{\ioin}{\iosqcns{\ioout}{\nix}}}}{\netsrc{0}}}{\netmeet{\netbox{\iosqcyc{\iosqcns{\ioin}{\iosqcns{\ioout}{\nix}}}}{\netrec{ones}{\netpeb{ones}}}}{\netbox{\iosqcyc{\iosqcns{\ioin}{\iosqcns{\ioout}{\nix}}}}{\netrec{ones}{\netpeb{ones}}}}}}}\\
  &\mcolred \netrec{nats}{\netbox{\iosqcns{\ioout}{\iosqcyc{\iosqcns{\ioin}{\iosqcns{\ioout}{\nix}}}}}{\netmeet{\netbox{\iosqcyc{\iosqcns{\ioin}{\iosqcns{\ioout}{\nix}}}}{\netsrc{0}}}{\netmeet{\netbox{\iosqcyc{\iosqcns{\ioin}{\iosqcns{\ioout}{\nix}}}}{\netrec{ones}{\netpeb{ones}}}}{\netbox{\iosqcyc{\iosqcns{\ioin}{\iosqcns{\ioout}{\nix}}}}{\netrec{ones}{\netpeb{ones}}}}}}}\\
  &\mcolred \netrec{nats}{\netmeet{\netbox{\iosqcns{\ioout}{\iosqcyc{\iosqcns{\ioin}{\iosqcns{\ioout}{\nix}}}}}{\netbox{\iosqcyc{\iosqcns{\ioin}{\iosqcns{\ioout}{\nix}}}}{\netsrc{0}}}}{\netbox{\iosqcns{\ioout}{\iosqcyc{\iosqcns{\ioin}{\iosqcns{\ioout}{\nix}}}}}{\netmeet{\netbox{\iosqcyc{\iosqcns{\ioin}{\iosqcns{\ioout}{\nix}}}}{\netrec{ones}{\netpeb{ones}}}}{\netbox{\iosqcyc{\iosqcns{\ioin}{\iosqcns{\ioout}{\nix}}}}{\netrec{ones}{\netpeb{ones}}}}}}}\\
  &\mcolred \netrec{nats}{\netmeet{\netbox{\iosqcyc{\iosqcns{\ioout}{\iosqcns{\ioin}{\nix}}}}{\netsrc{0}}}{\netbox{\iosqcns{\ioout}{\iosqcyc{\iosqcns{\ioin}{\iosqcns{\ioout}{\nix}}}}}{\netmeet{\netbox{\iosqcyc{\iosqcns{\ioin}{\iosqcns{\ioout}{\nix}}}}{\netrec{ones}{\netpeb{ones}}}}{\netbox{\iosqcyc{\iosqcns{\ioin}{\iosqcns{\ioout}{\nix}}}}{\netrec{ones}{\netpeb{ones}}}}}}}\\
  &\mcolred \netmeet{\netrec{nats}{\netbox{\iosqcyc{\iosqcns{\ioout}{\iosqcns{\ioin}{\nix}}}}{\netsrc{0}}}}{\netrec{nats}{\netbox{\iosqcns{\ioout}{\iosqcyc{\iosqcns{\ioin}{\iosqcns{\ioout}{\nix}}}}}{\netmeet{\netbox{\iosqcyc{\iosqcns{\ioin}{\iosqcns{\ioout}{\nix}}}}{\netrec{ones}{\netpeb{ones}}}}{\netbox{\iosqcyc{\iosqcns{\ioin}{\iosqcns{\ioout}{\nix}}}}{\netrec{ones}{\netpeb{ones}}}}}}}\\
  &\mcolred \netmeet{\netrec{nats}{\netsrc{1}}}{\netrec{nats}{\netmeet{\netbox{\iosqcns{\ioout}{\iosqcyc{\iosqcns{\ioin}{\iosqcns{\ioout}{\nix}}}}}{\netbox{\iosqcyc{\iosqcns{\ioin}{\iosqcns{\ioout}{\nix}}}}{\netrec{ones}{\netpeb{ones}}}}}{\netbox{\iosqcns{\ioout}{\iosqcyc{\iosqcns{\ioin}{\iosqcns{\ioout}{\nix}}}}}{\netbox{\iosqcyc{\iosqcns{\ioin}{\iosqcns{\ioout}{\nix}}}}{\netrec{ones}{\netpeb{ones}}}}}}}\\
  &\mcolred \netmeet{\netrec{nats}{\netsrc{1}}}{\netrec{nats}{\netmeet{\netbox{\iosqcyc{\iosqcns{\ioout}{\iosqcns{\ioin}{\nix}}}}{\netrec{ones}{\netpeb{ones}}}}{\netbox{\iosqcyc{\iosqcns{\ioout}{\iosqcns{\ioin}{\nix}}}}{\netrec{ones}{\netpeb{ones}}}}}}\\
  &\mcolred \netmeet{\netsrc{1}}{\netmeet{\netrec{nats}{\netbox{\iosqcyc{\iosqcns{\ioout}{\iosqcns{\ioin}{\nix}}}}{\netrec{ones}{\netpeb{ones}}}}}{\netrec{nats}{\netbox{\iosqcyc{\iosqcns{\ioout}{\iosqcns{\ioin}{\nix}}}}{\netrec{ones}{\netpeb{ones}}}}}}\\
  &\mcolred \netmeet{\netsrc{1}}{\netmeet{\netrec{nats}{\netbox{\iosqcyc{\iosqcns{\ioout}{\iosqcns{\ioin}{\nix}}}}{\netrec{ones}{\netbox{\iosqcns{\ioout}{\iosqcyc{\iosqcns{\ioin}{\iosqcns{\ioout}{\nix}}}}}{ones}}}}}{\netrec{nats}{\netbox{\iosqcyc{\iosqcns{\ioout}{\iosqcns{\ioin}{\nix}}}}{\netrec{ones}{\netbox{\iosqcns{\ioout}{\iosqcyc{\iosqcns{\ioin}{\iosqcns{\ioout}{\nix}}}}}{ones}}}}}}\\
  &\mcolred \netmeet{\netsrc{1}}{\netmeet{\netrec{nats}{\netbox{\iosqcyc{\iosqcns{\ioout}{\iosqcns{\ioin}{\nix}}}}{\netsrc{\conattop}}}}{\netrec{nats}{\netbox{\iosqcyc{\iosqcns{\ioout}{\iosqcns{\ioin}{\nix}}}}{\netsrc{\conattop}}}}}\\
  &\mcolred \netmeet{\netsrc{1}}{\netmeet{\netrec{nats}{\netsrc{\conattop}}}{\netrec{nats}{\netsrc{\conattop}}}}\\
  &\mcolred \netmeet{\netsrc{1}}{\netmeet{\netsrc{\conattop}}{\netsrc{\conattop}}}\\
  &\mcolred \netmeet{\netsrc{1}}{\netsrc{\conattop}}\\
  &\mcolred \netsrc{1}\end{align*}
Failed to prove productivity of $\msf{nats}$.\\[1.5ex]

\noindent $\msf{ones}$ depends only on pure stream functions, we can \textbf{decide productivity}.

\noindent We translate $\mathsf{ones}$ into a pebbleflow net and collapse it to a source:\textcolor{white}{\noteSimplify{}}
\begin{align*}
  \trnsl{\mathsf{ones}} &= \netrec{ones}{{\netpeb{ones}}}\\
                                    &=%
                                      \netrec{ones}{\netpeb{ones}}\\
  &\mcolred \netrec{ones}{\netbox{\iosqcns{\ioout}{\iosqcyc{\iosqcns{\ioin}{\iosqcns{\ioout}{\nix}}}}}{ones}}\\
  &\mcolred \netsrc{\conattop}\end{align*}
The specification of $\msf{ones}$ is productive.
\end{shaded}

\subsection{Example on Page~\pageref{example:traces}}
  \label{sec:examples:subsec:example:traces}

Now we consider the stream function specification below Def.~\ref{def:dorange} 
on page~\pageref{example:traces},
which is not a stream specification 
because it does not contain a stream constant as its root.
Still, the automated productivity prover can be used to compute
the precise data-oblivious lower bound of the stream function $\strff{f}$. 
On the input:
{\small\begin{itemize}\item []
\begin{verbatim}
Signature(
  f : stream(bit) -> stream(bit),
  g : stream(bit) -> stream(bit) -> stream(bit),
  0, 1 : bit
)
f(s) = g(s,s)
g(0:y:s,x:t) = 0:0:g(s,t)
g(1:s,x1:x2:x3:x4:t) = 0:0:0:0:0:g(s,t)
\end{verbatim}
\end{itemize}}
\noindent
the productivity prover generates the following output:
\begin{shaded}
\noindent The automated productivity prover has been applied to:
{\small\begin{itemize}\item []
\begin{verbatim}
Signature(
  -- stream symbols -- 
  f : stream(bit) -> stream(bit),
  g : stream(bit) -> stream(bit) -> stream(bit),

  -- data symbols -- 
  0 : bit,
  1 : bit
)

-- stream layer -- 
f(s) = g(s,s)
g(0:y:s,x:t) = 0:0:g(s,t)
g(1:s,x1:x2:x3:x4:t) = 0:0:0:0:0:g(s,t)

-- data layer -- 

\end{verbatim}
\end{itemize}}
\noindent The function symbol $\msf{f}$ is flat, we can compute the precise data-oblivious lower bound:
\begin{align*}
\trnsl{\msf{f}} &= \netgate{\netbox{\trnsli{\msf{f}}{\argnone,0}}{\netsrc{\numzer}},\;\trnsli{\msf{f}}{1,0}}\\
\trnsli{\msf{f}}{1,0} &= \mu \trnslvar{f}{1}{0}.\wedge \begin{cases}\wedge \begin{cases}\mu \trnslvar{g}{1}{0}.\wedge \begin{cases}\ioin \ioin \ioout \ioout \wedge \begin{cases}\trnslvar{g}{1}{0}\end{cases}\\\ioin \ioout \ioout \ioout \ioout \ioout \wedge \begin{cases}\trnslvar{g}{1}{0}\end{cases}\end{cases}\\\mu \trnslvar{g}{2}{0}.\wedge \begin{cases}\ioin \ioout \ioout \wedge \begin{cases}\trnslvar{g}{2}{0}\end{cases}\\\ioin \ioin \ioin \ioin \ioout \ioout \ioout \ioout \ioout \wedge \begin{cases}\trnslvar{g}{2}{0}\end{cases}\end{cases}\end{cases}\end{cases}\\
 &= \ios{----++-++-+--++-+}{-++-}\\
\trnsli{\msf{f}}{\argnone,0} &= \mu \trnslvar{f}{\argnone}{0}.\wedge \begin{cases}\mu \trnslvar{g}{\argnone}{0}.\wedge \begin{cases}\ioout \ioout \trnslvar{g}{\argnone}{0}\\\ioout \ioout \ioout \ioout \ioout \trnslvar{g}{\argnone}{0}\end{cases}\end{cases}\\
 &= \ios{}{+}\\
\end{align*}
\noindent The function symbol $\msf{g}$ is flat, we can compute the precise data-oblivious lower bound:
\begin{align*}
\trnsl{\msf{g}} &= \netgate{\netbox{\trnsli{\msf{g}}{\argnone,0}}{\netsrc{\numzer}},\;\trnsli{\msf{g}}{1,0},\;\trnsli{\msf{g}}{2,0}}\\
\trnsli{\msf{g}}{1,0} &= \mu \trnslvar{g}{1}{0}.\wedge \begin{cases}\ioin \ioin \ioout \ioout \wedge \begin{cases}\trnslvar{g}{1}{0}\end{cases}\\\ioin \ioout \ioout \ioout \ioout \ioout \wedge \begin{cases}\trnslvar{g}{1}{0}\end{cases}\end{cases}\\
 &= \ios{}{--++}\\
\trnsli{\msf{g}}{2,0} &= \mu \trnslvar{g}{2}{0}.\wedge \begin{cases}\ioin \ioout \ioout \wedge \begin{cases}\trnslvar{g}{2}{0}\end{cases}\\\ioin \ioin \ioin \ioin \ioout \ioout \ioout \ioout \ioout \wedge \begin{cases}\trnslvar{g}{2}{0}\end{cases}\end{cases}\\
 &= \ios{--}{--++-++-+}\\
\trnsli{\msf{g}}{\argnone,0} &= \mu \trnslvar{g}{\argnone}{0}.\wedge \begin{cases}\ioout \ioout \trnslvar{g}{\argnone}{0}\\\ioout \ioout \ioout \ioout \ioout \trnslvar{g}{\argnone}{0}\end{cases}\\
 &= \ios{}{+}\\
\end{align*}
\end{shaded}

\section{Conclusion and Further Work}\label{sec:conclusion}
In order to formalize quantitative approaches for recognizing productivity of stream specifications, 
we defined the notion of \daob{} rewriting and investigated \daob{} productivity.
For the syntactic class of flat stream specifications
(that employ pattern matching on data), 
we devised a decision algorithm for \daob{} productivity.
In this way we settled the productivity recognition problem 
for flat stream specifications from a \daob{} perspective.
For the even larger class including friendly nesting stream function rules,
we obtained a computable sufficient condition for productivity.
For the subclass of pure stream specifications 
(a substantial extension of the class given in~\cite{endr:grab:hend:isih:klop:2007})
we showed that productivity and \daob{} productivity coincide,
and thereby obtained a decision algorithm for productivity
of pure specifications.

We have implemented in Haskell the decision algorithm for \daob{} productivity.
This tool, together with more information including a manual, examples, our related papers,
and a comparison of our criteria with those 
of~\cite{hugh:pare:sabr:1996,telf:turn:1997,buch:2005}
can be found at our web page 
{\small{\texttt{http://infinity.few.vu.nl/productivity}}}.
The reader is invited to experiment with our tool.

It is not possible to obtain a \daobly{} optimal criterion
for non-productivity of flat specifications
in an analogous way to how we established
such a criterion for productivity.
This is because \daob{} upper bounds
on the production of stream functions
in flat stream specifications are not in general periodically increasing functions. 
For example, for the following stream function specification:
\begin{align*}
  \bfunap{\astrfun}{\strcns{x}{\astr}}{\bstr}
  \to \strcns{x}{\bfunap{\astrfun}{\astr}{\bstr}} 
  \punc,
  &&
  \bfunap{\astrfun}{\astr}{\strcns{y}{\bstr}}
  \to \strcns{y}{\bfunap{\astrfun}{\astr}{\bstr}} 
  \punc,
\end{align*}  
it holds that $\doUp{}{\strff{f}}(n_1,n_2) = n_1 + n_2$, which is not \pein{}. 
While this example is not orthogonal, 
$\doUp{}{\strff{f}}$ is also not \pein{} for the following 
similar orthogonal example:
\begin{align*}
  \bfunap{\astrfun}{\strcns{\datf{0}}{\strcns{x}{\astr}}}{\strcns{y}{\bstr}}
  & \to \strcns{x}{\bfunap{\astrfun}{\astr}{\bstr}}
  \punc,
  &
  \bfunap{\astrfun}{\strcns{\datf{1}}{\astr}}{\strcns{x}{\strcns{y}{\bstr}}}
  & \to \strcns{y}{\bfunap{\astrfun}{\astr}{\bstr}}
  \punc.
\end{align*}

Currently we are developing a method 
that goes beyond a \daob{} analysis,
one that would, e.g., prove productivity of the example~$\strcf{B}$
given in the introduction. 
Moreover, we study a refined 
production calculus that accounts for the delay of evaluation of stream elements, 
in order to obtain a faithful modelling of lazy evaluation,
needed for example for $\strcf{S}$ on page~\pageref{page:example:sijtsma},
where the first element depends on a `future' expansion of $\strcf{S}$.

\paragraph{Acknowledgement.} 
  We thank Jan Willem Klop, Carlos Lombardi, Vincent van Oostrom,
  and Roel de Vrijer for useful discussions.

\bibliography{main}

\appendix

\newcommand{\appsection}[1]{
  \let\oldthesection\thesection
  \renewcommand{\thesection}{Appendix \oldthesection}
  \section{#1}
  \let\thesection\oldthesection
}

\input{appendix/macros.def}

\newpage
\appsection{Solving Weakly Guarded IO-Term Specifications 
  (Lemma~\ref{lem:infioseqspec:ratsol})}%

Let $\aioseqspec = \{ X_{\alpha} = \iosqsyx_{\alpha} \}_{\alpha\in A} $
be a weakly guarded finite \infioseqspec{},
$\sroot\in A$, and $\solroot$ the unique solution of $\aioseqspec$ for $\Xroot$.
\begin{definition}\normalfont\label{def:g}
  We define a finite, rooted graph $\g = \triple{\gNodes}{\gEdges}{\gRoot}$
  that represents traces through the expressions in $\aioseqspec$.
  The set of nodes $\gNodes$ consists of 
  all pairs of $\alpha \in A$ 
  together with positions in $\iosqsyx_{\alpha}$,
  the right-hand side of the corresponding equation in $\aioseqspec$:
  \[ 
    \gNodes 
    = \{ \pair{\alpha}{p} \where \alpha \in A,\; p \in \pos{\iosqsyx_{\alpha}} \}
    \punc.
  \]
  Note that the set of nodes $\gNodes$ is finite.
  The root node $\gRoot$ is $\pair{\sroot}{\posemp}$.
  The set of labelled edges 
  $\gEdges \subseteq \gNodes \times \{\ioin,\ioout,\lemp\} \times \gNodes$
  is the smallest set such that:
  for all $\alpha \in A$, and every position $p \in \pos{\iosqsyx_{\alpha}}$ 
  the following conditions hold:
  \begin{enumerate}
    \item if $\iosqsyx_{\alpha}|_p = X_{\alpha'}$ 
          then $\gEdge{\pair{\alpha}{p}}{\lemp}{\pair{\alpha'}{\posemp}} \in \gEdges$;
    \item if $\iosqsyx_{\alpha}|_p = \iosi{-}{E'}$ 
          then $\gEdge{\pair{\alpha}{p}}{\ioin}{\pair{\alpha}{\poscat{p}{1}}} \in \gEdges$;
    \item if $\iosqsyx_{\alpha}|_p = \iosi{+}{E'}$ 
          then $\gEdge{\pair{\alpha}{p}}{\ioout}{\pair{\alpha}{\poscat{p}{1}}} \in \gEdges$;
    \item if $\iosqsyx_{\alpha}|_p = \iosqinf{E'}{E''}$ 
          then $\gEdge{\pair{\alpha}{p}}{\lemp}{\pair{\alpha}{\poscat{p}{1}}}$ 
          and $\gEdge{\pair{\alpha}{p}}{\lemp}{\pair{\alpha}{\poscat{p}{2}}} \in \gEdges$.
  \end{enumerate}
  where we have written $\gEdge{\agNode}{\ell}{\bgNode}$ 
  to denote $\triple{\agNode}{\ell}{\bgNode}\in\gEdges$.
  Moreover, for $\ell\in\{\ioin,\ioout,\lemp\}$, 
  we let $\slabrel{\ell}$ denote the relation 
  $\{\pair{\agNode}{\bgNode}\where \gEdge{\agNode}{\ell}{\bgNode}\}$.
\end{definition}
Note that by weak guardedness of $\aioseqspec$
the relation $\sepsrel$ is terminating.

We proceed by defining the set of traces $\lang{\g}{\gNode}$ 
through the graph $\g$ starting at node $\gNode\in\gNodes$.
\begin{definition}\normalfont\label{def:G:trace}
  For a path 
  $\gPath : \gNode_1 \stackrel{l_1}{\to} \gNode_2 \stackrel{l_2}{\to} \ldots$
  in $\g$, where $l_1,l_2,\ldots\in\{\ioin,\ioout,\lemp\}$, we use $\trace{\gPath}$ 
  to denote $l_1 \cdot l_2\cdots \in \iosq$, 
  the word obtained by concatenating the symbols `$\ioin$' and `$\ioout$'
  along the path $\gPath$, thus treating the label `$\lemp$' as the empty word.
  Note that $\trace{\gPath} \in \iostr$ whenever the path $\gPath$ is infinite
  by weak guardedness of $\aioseqspec$.
  The set of traces $\lang{\g}{\gNode} \subseteq \iostr$ through $\g$ 
  from $\gNode \in \gNodes$ is defined as follows:
  \[
    \lang{\g}{\gNode} 
    \defdby 
    \{
      \trace{\gPath} 
      \where 
      \text{$\gPath$ is an infinite path starting from node $\agNode$ in $\g$} 
    \}
    \punc.
  \]
\end{definition}

For $Z \equiv \{ \ibiosq{1},\ldots,\ibiosq{n} \} \subseteq\iostr$, 
let $\iosqInfm{\biosq\in Z}{\biosq}$ denote 
$\iosqinf{\ibiosq{1}}{\iosqinf{\ldots}{\ibiosq{n}}}$,
where $\siosqinf$ is the binary operation defined in Def.~\ref{def:iosq:inf}.
If $Z = \setemp$, we set $\iosqInfm{\biosq\in Z}{\biosq} = \iosqcyc{\ioout}$.

The following lemma establishes a correspondence between the specification $\aioseqspec$
and the graph $\g$, in particular it implies that 
the unique solution $\solroot$ of $\aioseqspec$ for $\Xroot$
equals the infimum over all traces through $\g$ from the root node. %
\begin{lemma}\label{lem:G}
  For $\alpha\in A$, $p\in\pos{\iosqsyx_{\alpha}}$, 
  let $\solexp{\alpha}{p}\in\iostr$ denote the unique solution for the expression
  $\iosqsyx_{\alpha}|_p$ in $\aioseqspec$, and let $\trcexp{\alpha}{p}$ 
  be shorthand for $\iosqInfm{\biosq\in\lang{\g}{\pair{\alpha}{p}}}{\biosq}$.
  Then, for all $\alpha\in A$, $p\in\pos{\iosqsyx_{\alpha}}$,
  we have $\solexp{\alpha}{p} = \trcexp{\alpha}{p}$.
\end{lemma}
\begin{proof}
  By coinduction, i.e., we show that the relation $R\in\iostr\times\iostr$ defined by:
  \[
    R 
    \defdby 
    \{ 
      \pair{\solexp{\alpha}{p}}{\trcexp{\alpha}{p}} 
      \where \alpha \in A,\; p\in\pos{\iosqsyx_{\alpha}} 
    \}
  \]
  is a bisimulation.
  We proceed by well-founded induction on pairs $\pair{\alpha}{p}$ with respect to $\sepsrel$.

  (Base)~If $\agNode\defdby\pair{\alpha}{p}$ is in normal form with respect to $\sepsrel$,
  then $\agNode$ has precisely one outgoing edge:
  $\gEdge{\agNode }{\ioin}{}$ or $\gEdge{\agNode}{\ioout}{}$.
  In case $\gEdge{\agNode}{\ioin}{}$, 
  we have that $\iosqsyx_{\alpha}|_p = \iosi{-}{E'}$, and therefore
  $\solexp{\alpha}{p} = \iosi{-}{\solexp{\alpha}{\poscat{p}{1}}}$.
  Furthermore, we have
  $\trcexp{\alpha}{p} = \iosi{-}{\trcexp{\alpha}{\poscat{p}{1}}}$, and
  clearly $\pair{\solexp{\alpha}{\poscat{p}{1}}}{\trcexp{\alpha}{\poscat{p}{1}}} \in R$.
  The proof for case $\gEdge{\agNode}{\ioout}{}$ proceeds analogously.
  
  (Step)~If $\agNode\defdby\pair{\alpha}{p}$ is not a normal form with respect to $\sepsrel$,
  then either $\agNode$ has one outgoing edge $\gEdge{\agNode}{\lemp}{\agNode'}$
  or two outgoing edges $\gEdge{\agNode}{\lemp}{\agNode'}$ 
  and $\gEdge{\agNode}{\lemp}{\agNode''}$.
  In the first case we have $\iosqsyx_{\alpha}|_p = X_{\alpha'}$,
  and then $\gNode' = \pair{\alpha'}{\posemp}$,
  $\solexp{\alpha}{p} = \solexp{\alpha'}{\posemp}$ 
  and $\trcexp{\alpha}{p} = \trcexp{\alpha'}{\posemp}$, and 
  we conclude by an application of the induction hypothesis for $\pair{\alpha'}{\posemp}$.
  In the second case we have that $\iosqsyx_{\alpha}|_p = \iosqinf{E'}{E''}$, 
  and then $\agNode' = \pair{\alpha}{\poscat{p}{1}}$ 
  and $\agNode'' = \pair{\alpha}{\poscat{p}{2}}$.
  Thus we get
  $\solexp{\alpha}{p} = \iosqinf{\solexp{\alpha}{\poscat{p}{1}}}{\solexp{\alpha}{\poscat{p}{2}}}$,
  and, by associativity of $\siosqinf$,
  $\trcexp{\alpha}{p} = \iosqinf{\trcexp{\alpha}{\poscat{p}{1}}}{\trcexp{\alpha}{\poscat{p}{2}}}$.
  We conclude by a double application of the induction hypothesis.
  \qed
\end{proof}

The set $\lang{\g}{\gNode}$ consists of infinite, possibly non-productive \ioseq{s}.
Therefore we adopt the definition of the 
production function $\siosqprd{\aiosq}\funin{\setfun{\conat}{\conat}}$
to be applicable to non-productive \ioseq{s} $\aiosq$ 
by dropping $\conattop$ from the domain.
\begin{definition}\normalfont
  The \emph{production function $\siosqprdx{\aiosq}$} of a sequence $\aiosq\in\iostr$
  is corecursively defined by, for all $n\in\nat$, 
  $\iosqprdx{\aiosq}{n} \defdby \iosqprdxr{\aiosq}{n}$:
  \begin{align*}
    \iosqprdxr{\iosqcns{\ioout}{\aiosq}}{n} 
    &= 1 + \iosqprdxr{\aiosq}{n}
    \punc,
    &
    \iosqprdxr{\iosqcns{\ioin}{\aiosq}}{\numzer} 
    &= \numzer
    \punc,
    &
    \iosqprdxr{\iosqcns{\ioin}{\aiosq}}{n+1}
    &= \iosqprdxr{\aiosq}{n}
    \punc.
  \end{align*}
\end{definition}

A property of the operation $\siosqinf$ 
given in Def.~\ref{def:iosq:inf} is that 
the mapping $\ssiosqprd\funin{\iosq\to\conat\to\conat}$ 
preserves infimum, see Prop.~\ref{prop:iosq:inf:interpret}.

Clearly, a similar property holds for $\ssiosqprdx$:
for all $\iaiosq{1},\iaiosq{2}\in\iostr$, we have:
\(
  \siosqprdx{\iosqinf{\iaiosq{1}}{\iaiosq{2}}}
  = 
  \funinf{\siosqprdx{\iaiosq{1}}}{\siosqprdx{\iaiosq{2}}}
\).
Then, the following corollary directly follows from Lemma~\ref{lem:G}.
\begin{corollary}\label{cor:G}
  For all $n\in\nat$ we have 
  \(
    \iosqprdx{\solroot}{n} 
    = \inf\{\iosqprdx{\biosq}{n} \where \biosq\in\lang{\g}{\gRoot}\}
  \).
  Moreover, given $n\in\nat$,
  we can always pick a trace $\biosq\in\lang{\g}{\gRoot}$ 
  such that $\iosqprdx{\biosq}{n} = \iosqprdx{\solroot}{n}$.
\end{corollary}

Using $\g$ we define a `two-dimensional diagram' $\T \subseteq \gNodes \times \nat \times \nat$
representing the traces of $\aioseqspec$ starting from $\Xroot$ as step functions.
For $\gNode \in \gNodes$, $x,y \in \nat$
let $\tEntry{\gNode}{x}{y}$ be shorthand for $\triple{\gNode}{x}{y}$.
The horizontal axis ($x$-axis) of $T$ corresponds to input, 
and the vertical axis ($y$-axis) to output;
that is, a `$\ioin$' in the trace yields a step to the right in the diagram
and `$\ioout$' a step upwards.
\begin{definition}\normalfont
  For $U \subseteq \gNodes \times \nat \times \nat$ we define 
  $\g(U) \subseteq \gNodes \times \nat \times \nat$,
  the \emph{one-step closure of $U$ under $\g$}, as follows:
  \begin{align*}
    \g(U) &\defdby U \cup \gi{\lemp}(U) \cup \gi{\ioout}(U) \cup \gi{\ioin}(U)\\
    \gi{\lemp}(U) &\defdby \{\tEntry{\bgNode}{x}{y} \where \tEntry{\agNode}{x}{y} \in U,\; \gEdge{\agNode}{\lemp}{\bgNode} \in \gEdges\}\\
    \gi{\ioout}(U) &\defdby \{\tEntry{\bgNode}{x}{y+1} \where \tEntry{\agNode}{x}{y} \in U,\; \gEdge{\agNode}{\ioout}{\bgNode} \in \gEdges\}\\
    \gi{\ioin}(U) &\defdby \{\tEntry{\bgNode}{x+1}{y} \where \tEntry{\agNode}{x}{y} \in U,\; \gEdge{\agNode}{\ioin}{\bgNode} \in \gEdges\}
  \end{align*}
  Furthermore we define
  \[\displaystyle \g^*(U) \defdby \bigcup_{n \in \nat} \g^n(U)\;,\]
  the many-step closure of $U$ under $\g$.
\end{definition}
\begin{definition}\normalfont
  We define $\T \subseteq \gNodes \times \nat \times \nat$ as $\T \defdby \g^*(\{\tRoot\})$.
\end{definition}
The diagram $\T$ contains
the traces through $\g$ starting from $\gRoot$ as `step functions'.
\begin{lemma}\label{lem:T:tracesG}
  For all $\gNode \in \gNodes$, $x,y \in \nat$
  we have $\tEntry{\gNode}{x}{y} \in \T$ if and only if
  there exists a path $\gPath : \gRoot \to^* \gNode$ in $\g$ such that
  $x = \seqcnt{\ioin}{\trace{\gPath}}$, $y = \seqcnt{\ioout}{\trace{\gPath}}$.
\end{lemma}
\begin{proof}
  Follows immediately from the definition of $\T$. \qed
\end{proof}
We proceed by showing that the solution $\solroot$ of $\aioseqspec$ for $\Xroot$
coincides with the `lower bound' of the traces in $\T$.
\begin{definition}\normalfont\label{def:T:low}
  For $U \subseteq \gNodes \times \nat \times \nat$ 
  and $\g = \triple{\gNodes}{\gEdges}{\gRoot}$ as above,
  we define $\tBound{U} \funin \nat \to \conat$, 
  the \emph{lower bound of $U$}, as follows:
  \[
    \tBound{U}(x) \defdby 
    \inf \{ y\in\nat \where \tEntry{\gNode}{x}{y}\in U,\; 
                              \myex{\bgNode}{\gEdge{\gNode}{\ioin}{\bgNode} \in \gEdges} \}
    \punc,
  \]
  where $\inf \setemp = \conattop$.
\end{definition}
\begin{lemma}\label{lem:T}
  The lower bound of $\T$ is the solution $\solroot$ of $\aioseqspec$ for $\Xroot$:
  \[ \myallin{n}{\nat}{\tBound{\T}(n) = \iosqprdx{\solroot}{n}} \]
\end{lemma}
\begin{proof}
  Let $n \in \nat$, 
  we show $\iosqprdx{\solroot}{n} \ge \tBound{\T}(n)$
  and $\iosqprdx{\solroot}{n} \le \tBound{\T}(n)$.
  \begin{itemize}
    \item [$\ge$]
      For $\iosqprdx{\solroot}{n} = \conattop$ the proof obligation is trivial,
      thus assume $\iosqprdx{\solroot}{n} < \conattop$.
      From Corollary~\ref{cor:G} it follows that there exists a sequence 
      $\biosq \in \lang{\g}{\gRoot}$ such that 
      $\iosqprdx{\biosq}{n} = \iosqprdx{\solroot}{n}$.
      Since $\iosqprdx{\biosq}{n} < \conattop$ and $\biosq \in \iostr$,
      there exists $m \in \nat$ with
      $\seqcnt{\ioin}{\seqsub{\biosq}{0}{m-1}} = n < n + 1 = \seqcnt{\ioin}{\seqsub{\biosq}{0}{m}}$,
      and
      $\seqcnt{\ioout}{\seqsub{\biosq}{0}{m-1}} = \iosqprdx{\biosq}{n}$.
      Then there is a path $\gPath : \gRoot \to^* \gEdge{\gNode}{\ioin}{\gNode'}$ in $\g$
      such that $\trace{\gPath} = \seqsub{\biosq}{0}{m}$.
      By Lemma~\ref{lem:T:tracesG} we have $\tEntry{\gNode}{n}{\iosqprdx{\solroot}{n}} \in \T$,
      and together with $\gEdge{\agNode}{\ioin}{\agNode'}$
      we conclude $\iosqprdx{\solroot}{n} \ge \tBound{\T}(n)$ according to Definition~\ref{def:T:low}.
    \item [$\le$] Assume $\tBound{\T}(n) < \conattop$.
      From Definition~\ref{def:T:low}
      it follows that there exist nodes
      $\agNode,\agNode' \in \gNodes$ with $\tEntry{\agNode}{n}{\tBound{\T}(n)} \in \T$
      and $\gEdge{\agNode}{\ioin}{\agNode'}$.
      Consequently by Lemma~\ref{lem:T:tracesG}
      there exists a path $\gPath : \gRoot \to^* \gEdge{\agNode}{\ioin}{\agNode'}$ in $\g$
      such that $\seqcnt{\ioin}{\trace{\gPath}} = n+1$ 
      and $\seqcnt{\ioout}{\trace{\gPath}} = \tBound{\T}(n)$.
      We choose a $\biosq\in\iostr$ that has $\trace{\gPath}$ as a prefix,
      so that $\iosqprdx{\biosq}{n} = \tBound{\T}(n)$.
      Then, by Corollary~\ref{cor:G}, we know that
      $\iosqprdx{\solroot}{n}\leq\tBound{\T}(n)$.
  \end{itemize}
  We conclude $\iosqprdx{\solroot}{n} = \tBound{\T}(n)$. \qed
\end{proof}
Although in principle $\T$ could be used to calculate 
$\solroot$ for every $n \in \nat$,
it remains to be shown how a rational representation
of the solution $\solroot$ of $\aioseqspec$ for $\Xroot$ can be obtained.
For this purpose we construct a `simpler' diagram $\TO \subseteq \T$
that has the same lower bound.
We construct $\TO$ such that there are only finitely many nodes on each vertical line;
thereby we employ the idea that whenever there
are two nodes $\tEntry{\bgNode}{x}{y}$, $\tEntry{\bgNode}{x}{y'}$
with $y < y'$,
then we can omit the upper occurrence $\tEntry{\bgNode}{x}{y'}$
(including all traces rising from $\tEntry{\bgNode}{x}{y'}$)
without having an influence on the lower bound $\tBound{\T}(x)$.

\begin{definition}\normalfont\label{def:TO}
  We define functions 
  $\g^0, \sTomit \funin (\gNodes \times \nat \times \nat) \to (\gNodes \times \nat \times \nat)$:
  \begin{align*}
    \g^0(U)   &\defdby U \cup \gi{\lemp}(U) \cup \gi{\ioout}(U) \\
    \Tomit{U} &\defdby 
      \{ \tEntry{\gNode}{x}{y} 
         \where 
         \tEntry{\gNode}{x}{y} \in U, \neg \myex{y' < y}{\tEntry{\gNode}{x}{y'} \in U}
      \}
  \end{align*}
  Let $\TOn{n} \subseteq \gNodes \times \nat \times \nat$ for all $n \in \nat$ 
  be defined as follows:
  \begin{align*}
    \TOn{0} &\defdby \Tvclosure{\{\tRoot\}}\\
    \TOn{n+1} &\defdby \TOn{n} \cup \TOon{n+1}\\
    \TOon{n+1} &\defdby \Tvclosure{ \stepright{n} }\\
    \stepright{n} &\defdby
      \{\tEntry{\bgNode}{n+1}{y} 
        \where \tEntry{\agNode}{n}{y} \in \TOn{n},\; \gEdge{\agNode}{\ioin}{\bgNode} \in \gEdges\}
      \punc,
  \end{align*}
  where $\Tvclosure{U} \subseteq \gNodes \times \nat \times \nat$ 
  is defined to be $\vcloseomit^m(U)$
  for the smallest $m \in \nat$ such that 
  $\vcloseomit^{m+1}(U) = \vcloseomit^m(U)$
  where $\vcloseomit \defdby \funcmp{\sTomit}{\g^0}$,
  the vertical one-step closure followed by an application of $\sTomit$.
  Furthermore we define
  \[\textstyle \TO \defdby \bigcup_{n \in \nat} \TOn{n}\punc.\]
\end{definition}
Note that $\TOn{n}$ is finite for all $n \in \nat$.
The important step in the construction
is the termination of $\Tvclosure{\_}$
which is guaranteed by the following lemma.
\begin{lemma}
  The computation of $\Tvclosure{U}$ is terminating
  for finite sets $U \subseteq \gNodes \times \nat \times \nat$,
  that is, there exists $m \in \nat$ such that  
  $\vcloseomit^{m+1}(U) = \vcloseomit^m(U)$.
\end{lemma}
\begin{proof}
  For $U \subseteq \gNodes \times \nat \times \nat$ we introduce auxiliary definitions:
  \begin{align*}
    X(U) &\defdby \{x \in \nat \where \myex{\agNode \in \gNodes,y \in \nat}{\tEntry{\agNode}{x}{y} \in U}\} \subseteq \nat\\
    \gNodes(U,x) &\defdby \{\agNode \in \gNodes \where \myex{y \in \nat}{\tEntry{\agNode}{x}{y} \in U}\} \subseteq \gNodes\\
    Y(U,x,\agNode) &\defdby \{y \in \nat \where \tEntry{\agNode}{x}{y} \in U\} \subseteq \nat\\
    y(U,x,\agNode) &\defdby \min (Y(U,x,\agNode)) \in \nat
    \punc.
  \end{align*}
  Note that $Y(W,x,\agNode) = \{y(W,x,\agNode)\}$ is a singleton set
  for every $W \defdby \Tomit{U}$, $x \in X(W)$ and $\agNode \in \gNodes(W,x)$ 
  with $U \subseteq \gNodes \times \nat \times \nat$ arbitrary.

  Let $U \subseteq \gNodes \times \nat \times \nat$. We define $U_0 \defdby U$ and for $n = 1,2,\ldots$ let
  \[U_n \defdby \funap{\vcloseomit}{U_{n-1}} \punc.\]
  Note that $U_n = \funap{\vcloseomit^n}{U}$.
  From the definition of $\g^0$ and $\sTomit$ it follows that:
  \begin{enumerate}
  \item $\myall{n \in \nat}{X(U_n) = X(U)}$,
  \item $\myall{n \in \nat}{U_n \subseteq \funap{\g^0}{U_n}}$,
  \item $\myall{n \in \nat, x \in X(U)}{\gNodes(U_n,x) \subseteq \gNodes(U_{n+1},x)}$, since
        \[\gNodes(U_n,x) \subseteq \gNodes(\funap{\g^0}{U_n},x) = \gNodes(\funap{\vcloseomit}{U_n},x) = \gNodes(U_{n+1},x)\]
  \item $\myall{n \ge 1, x \in X(U)}{\card{\gNodes(U_n,x)} \le \card{\gNodes}}$.
  \end{enumerate}
  By (iii) and (iv) we conclude that there exists $n_0 \in \nat$:
  \begin{equation}\label{eq:vclosure1}
    \myall{n \ge n_0, x \in X(U)}{\gNodes(U_{n},x) = \gNodes(U_{n_0},x)}\punc,
  \end{equation}
  that is, the set of nodes $\gNodes(U_n,x)$ on each vertical strip $x$ becomes fixed after $n_0$ steps.
  Furthermore by (ii) and the definition of $\sTomit$ we get
  \begin{equation}\label{eq:vclosure2}
    \myall{n \in \nat, x \in X(U), \agNode \in \gNodes(U_n,x)}{y(U_{n+1},x,\agNode) \le y(U_n,x,\agNode)}\punc,
  \end{equation}
  that is, the heights of the nodes are non-increasing.
  Therefore for each $x \in X(U)$ and node $\agNode \in \gNodes(U_{n_0},\agNode)$ there exists
  $m_{x,\agNode} \ge n_0$ such that:
  \[\myall{n \ge m_{x,\agNode}}{y(U_n,x,\agNode) = y(U_{m_{x,\agNode}},x,\agNode)}\]
  Let $m \defdby \max \{m_{x,\agNode} \where x \in X(U),\; \agNode \in \gNodes(U_{n_0},\agNode)\}$,
  then the heights of all nodes are fixed from $m$ onwards
  and consequently we have $\myall{n \ge m}{U_n = U_m}$. \qed
\end{proof}
\begin{lemma}\label{lem:TO}
  $\T$ and $\TO$ have the same lower bound, that is, $\tBound{\T} =\tBound{\TO}$.
\end{lemma}
\begin{proof}
  
\end{proof}

Now we construct a rational representation of the lower bound of $\TO$,
which by Lemma~\ref{lem:TO} and \ref{lem:T} is 
the solution $\solroot$ of $\aioseqspec$ for $\Xroot$.
The construction is based on a `repetition search' on $\TO$,
that is, we search two vertical strips $x_1$ and $x_2$ 
that contain similar constellations of nodes 
allowing to
conclude that from the vertical position $x_1$ onwards the lower bound is `quasi periodic'.

\begin{definition}\normalfont
For $x \in \nat$, $\agNode \in \gNodes$ let
  \begin{align*}
    \gNodes(x) &\defdby \{\agNode \in \gNodes \where \myex{y \in \nat}{\tEntry{\agNode}{x}{y} \in \TOn{x}}\} \subseteq \gNodes\\
    y(x,\agNode) &\defdby 
       \begin{cases}
         y &\text{if there exists $y \in \nat$ with } \tEntry{\agNode}{x}{y} \in \TOn{x}\\
         \conattop &\text{otherwise}
       \end{cases}
    \punc.
  \end{align*}
\end{definition}
Note that $y(x,\agNode)$ is well-defined since there exists at most one such $y \in \nat$.

\begin{definition}\normalfont
For $x \in \nat$ we define
\[h(x) \defdby \min \{y(x,\gNode) \where \gNode \in \gNodes\}\punc,\]
the height of the lowest node on the vertical strip $x$.
\end{definition}

\begin{definition}\normalfont
For $x \in \nat$, $\agNode \in \gNodes$ let
\[y_r(x,\agNode) \defdby y(x,\agNode) - h(x)\punc,\]
the height of the node $\agNode$ relative to the lowest node on strip $x$.
\end{definition}

\begin{definition}\normalfont
For $x_1 < x_2 \in \nat$ with $\gNodes(x_1) \ne \setemp$ and $\gNodes(x_2) \ne \setemp$ let
\[d(x_1,x_2) \defdby h(x_2) - h(x_1) \punc,\]
the height difference between the vertical strips $x_1$ and $x_2$.
\end{definition}

\begin{definition}\normalfont
For $x_1 < x_2 \in \nat$, 
$\agNode \in \gNodes$ with $\agNode \in \gNodes(x_1)$ and $\agNode \in \gNodes(x_2)$
let
\[d(x_1,x_2,\agNode) \defdby y(x_2,\agNode) - y(x_1,\agNode)\punc,\]
the height difference of the node $\agNode$ between the vertical strips $x_1$ and $x_2$,
and
\[d_r(x_1,x_2,\agNode) \defdby h(x_1,x_2,\agNode) - d(x_1,x_2)\punc,\]
the relative height difference of the node $\agNode$ between the vertical strips $x_1$ and $x_2$.
\end{definition}

\begin{lemma}
  For $x_1,x_2 \in \nat$, $v \in \gNodes(x_1)$ it holds:
  \[y_r(x_2,\agNode) = y_r(x_1,\agNode) + d_r(x_1,x_2,\agNode)\punc.\]
\end{lemma}
\begin{proof}
  Immediately from the definitions.\qed
\end{proof}

We proceed with the definition of `pseudo repetitions' $\pair{x_1}{x_2} \in \nat^2$.
That is, we require that the vertical strips $x_1$ and $x_2$
possess a similar constellation of nodes.
A pseudo repetition does not yet guarantee
`quasi periodicity' of the lower bound (having found the cyclic part of the \ioseq),
as will be explained after the definition.
\begin{definition}\normalfont
A \emph{pseudo repetition} is a pair $\pair{x_1}{x_2} \in \nat^2$ such that:
\begin{enumerate}
  \item The vertical strips $x_1$ and $x_2$ contain the same set of nodes $\agNode \in \gNodes$:
        \[\gNodes(x_1) = \gNodes(x_2)\]
  \item From strip $x_1$ to $x_2$ all nodes have increased their height by at least $d(x_1,x_2)$:
        \[\myall{\gNode \in \gNodes}{d_r(x_1,x_2,\agNode) \ge 0}\]
\end{enumerate}
\end{definition}

For a pseudo repetition $\pair{x_1}{x_2}$
we can distinguish between nodes $\agNode$
for which $d_r(x_1,x_2,\agNode) = 0$
and those with $d_r(x_1,x_2,\agNode) > 0$.
If for all nodes $\agNode \in \gNodes(x_1)$ it holds that $d_r(x_1,x_2,\agNode) = 0$
then we have an exact repetition,
yielding `quasi periodicity' of the lower bound.
On the other hand 
assume that there is a node $\agNode \in \gNodes(x_1)$
which increases its relative height
and contributes to the lower bound between the vertical strips $x_1$ and $x_2$.
Then we do not yet have a `quasi periodic' lower bound,
because due to the increasing relative height of $\agNode$,
the contribution $\agNode$ to the lower bound will change.
For this reason we will later strengthen the conditions on pseudo repetitions,
yielding `repetitions' which then will guarantee `quasi periodicity', see Definition~\ref{def:repetition}.

\begin{lemma}
  There exist pseudo repetitions in $\TO$.
\end{lemma}
\begin{proof}
  For every $x \in \nat$ we have $\gNodes(x) \in \powerset{\gNodes}$
  and $\powerset{\gNodes}$ is a finite set.
  Hence it follows from the Pigeonhole Principle
  that there exist indices $i_1 < i_2 < i_3 < \ldots$
  with $\gNodes(i_1) = \gNodes(i_2) = \gNodes(i_3) = \ldots$, 
  let $\gNodes' \defdby \gNodes(i_1) = \{\agNode_1,\ldots,\agNode_k\}$ with $k = \card{\gNodes(i_1)}$.

  Let $\le$ on $\nat^k$ be defined as follows:
  \[(a_1,\ldots,a_l) \le (b_1,\ldots,b_l) \Leftrightarrow a_1 \le b_1 \wedge \ldots \wedge a_l \le b_l \punc.\]
  For $x \in \{i_1,i_2,\ldots\}$ we define
  the tuple $\vec{t_x} \defdby (y_r(x,\agNode_1),\ldots,y_r(x,\agNode_k)) \in \nat^k$.
  It is a well-known fact
  that $\le$ is an almost full relation on $\nat^k$, see \cite{terese:2003}.
  Hence there exist $x_1,x_2 \in \{i_1,i_2,\ldots\}$
  with $\vec{t_{x_1}} \le \vec{t_{x_2}}$.
  Then $d_r(x_1,x_2,\agNode) \ge 0$ for all $\agNode \in \gNodes'$
  which establishes that $\pair{x_1}{x_2}$ is a pseudo repetition.\qed
\end{proof}

The following definitions and lemmas are auxiliary for the purpose 
of proving Lemma~\ref{lem:pseudoreprep}.
\begin{definition}
  For $U, W \subseteq \gNodes \times \nat \times \nat$ and $d_x,d_y \in \nat$ we define
  \[U \le_{d_x,d_y} W \Longleftrightarrow \myall{\tEntry{\agNode}{x}{y} \in W}{\myex{y' < y-d}{\tEntry{\agNode}{x-d_x}{y'} \in U}}\punc.\]
\end{definition}
\begin{lemma}
  For $x_1 < x_2 \in \nat$ and $d_y \in \nat$ we have:
  \[\TOon{x_1} \le_{x_2-x_1,d_y} \TOon{x_2} 
    \Longleftrightarrow \myall{\agNode \in \gNodes(x_1)}{y(x_2,\agNode) \ge d_y + y(x_1,\agNode)}\]
\end{lemma}
\begin{proof}
  Follows from the definitions.\qed
\end{proof}
\begin{lemma}\label{lem:V:vco:equality}
  If $U_1, U_2 \subseteq \gNodes \times \nat \times \nat$, $d_x,d_y \in \nat$ with $U_1 \le_{d_x,d_y} U_2$ then:
  \begin{enumerate}
    \item $\vcloseomit(U_1) \le_{d_x,d_y} \vcloseomit(U_2)$, and
    \item $\Tvclosure{U_1} \le_{d_x,d_y} \Tvclosure{U_2}$.
  \end{enumerate}
\end{lemma}
\begin{proof}
  From $U_1 \le_{d_x,d_y} U_2$ it follows that $\g^0(U_1) \le_{d_x,d_y} \g^0(U_2)$ holds, 
  and then we get $\funcmp{\sTomit}{\g^0}(U_1) \le_{d_x,d_y} \funcmp{\sTomit}{\g^0}(U_2)$ by definition of $\sTomit$.

  By induction we get
  $\myall{n \in \nat}{\vcloseomit^n(U_1) \le_{d_x,d_y} \vcloseomit^n(U_2)}$.
  Furthermore there exist $m_1,m_2 \in \nat$ such that
  $\Tvclosure{U_i} = \vcloseomit^{m_i}(U_i)$ for $i \in \{1,2\}$ by definition.
  From $\myall{m_i' \ge m_i}{\vcloseomit^{m_i'}(U_i) = \vcloseomit^{m_i}(U_i)}$
  it follows that for $i \in \{1,2\}$ we have $\Tvclosure{U_i} = \vcloseomit^{\max(m_1,m_2)}(U_i)$.
  Hence $\Tvclosure{U_1} \le_{d_x,d_y} \Tvclosure{U_2}$. \qed
\end{proof}
\begin{lemma}\label{lem:V:equality}
  If $x_1 < x_2 \in \nat$ with $\gNodes(x_1) = \gNodes(x_2)$,
  then $\gNodes(x_1+1) = \gNodes(x_2+1)$.
  If moreover $\TOon{x_1} \le_{x_2-x_1,d_y} \TOon{x_2}$ for $d_y \in \nat$,
  then $\TOon{x_1+1} \le_{x_2-x_1,d_y} \TOon{x_2+1}$.
\end{lemma}
\begin{proof}
  From $\gNodes(x_1) = \gNodes(x_2)$ and the definitions of $\TOn{x_1+1}$ and $\TOn{x_2+1}$
  it follows that $\gNodes(x_1+1) = \gNodes(x_2+1)$.

  Assume $\TOon{x_1} \le_{x_2-x_1,d_y} \TOon{x_2}$ and consider Definition~\ref{def:TO}.
  Then it follows that $\stepright{x_1} \le_{x_2-x_1,d_y} \stepright{x_2}$.
  We conclude with an application of Lemma~\ref{lem:V:vco:equality} yielding
  $\TOon{x_1+1} \le_{x_2-x_1,d_y} \TOon{x_2+1}$.
\qed
\end{proof}
\begin{lemma}\label{lem:pseudoreprep}
  Let $\pair{x_1}{x_2}$ be a pseudo repetition.
  Then the pair $\pair{x_1 + m}{x_2 + m}$ is a pseudo repetition for all $m \in \nat$.
\end{lemma}
\begin{proof}
  It suffices to show that $\pair{x_1 + 1}{x_2 + 1}$ is a pseudo repetition,
  for $m > 1$ the claim follows by induction.

  We get $\gNodes(x_1+1) = \gNodes(x_2+1)$ by Lemma~\ref{lem:V:equality}.
  Furthermore by the definition of pseudo repetition we have
  $\myall{\gNode \in \gNodes(x_1)}{d_r(x_1,x_2,\agNode) \ge 0}$,
  which is equivalent to
  $\TOon{x_1} \le_{x_2-x_1,d(x_1,x_2)} \TOon{x_2}$ since $\gNodes(x_1) = \gNodes(x_2)$.
  Emloying Lemma~\ref{lem:V:equality}
  we get $\TOon{x_1+1} \le_{x_2-x_1,d(x_1,x_2)} \TOon{x_2+1}$,
  
  that is $\myall{\gNode \in \gNodes(x_1+1)}{d_r(x_1+1,x_2+1,\agNode) \ge 0}$.
  Hence $\pair{x_1 + 1}{x_2 + 1}$ is a pseudo repetition.
  \qed
\end{proof}

We proceed with the definition of `repetitions'.
Then we show that repetitions exist
and that these indeed guarantee quasi periodicity of the lower bound.

\begin{definition}\normalfont\label{def:repetition}
  A \emph{repetition} is a pseudo repetition $\pair{x_1}{x_2}$ if:
  \begin{enumerate}
  \item $\pair{x_1}{x_2}$ is a pseudo repetition
  \item Let $I(x_1,x_2) \defdby \{\tEntry{\gNode}{x}{y} \where \tEntry{\gNode}{x}{y} \in T_0,\; x = x_1,\; i(x_1,x_2,p) = 0\}$. We require that only nodes from $I(x_1,x_2)$ contribute to the lower bound
        between $x_1$ and $x_2$
        and these nodes reconstruct themselfs on the strip $x_2$:
        \begin{gather}
        \myall{x \in \nat, x_1 \le x \le x_2}{\tBound{\g(I)}(x) = \tBound{T_0}(x)}\\
        \myall{\tEntry{\gNode}{x}{y} \in I(x_1,x_2)}{\tEntry{\gNode}{x_2}{y+d(x_1,x_2)} \in \g(I)}
        \end{gather}
  \end{enumerate}
\end{definition}

Both conditions can be effectively checked since
$\g(I)|_{\le n}$ can be computed for every $n \in \nat$.
Having such a repetition
is is easy to see that the lower bounds behaves `quasi periodic'
from this point onwards, that is the \ioseq{} is indeed rational.
This is because the the nodes with $i(x_1,x_2,\gNode) = 0$
consitute a exactly repeating pattern.
Furthermore the nodes with increasing relative height $i(x_1,x_2,\gNode) > 0$
do not contribute to the lower bound between $x_1$ and $x_2$,
and since their height cannot decrease
they will also not contribute in the future.

It remains to show that we indeed always encounter a repetition $\pair{x_1}{x_2}$.
We know that there exists a pseudo repetition $\pair{x_1}{x_2}$.
Then we have 
for all $m \in \nat$ we have $\pair{x_1 + m \cdot (x_2-x_1)}{x_2 + m \cdot (x_2-x_1)}$
is a pseudo repetition.
For all $m \in \nat$ and $\gNode \in \gNodes$
we have $i(x_1 + m \cdot (x_2-x_1),\;x_2 + m \cdot (x_2-x_1),\gNode) \ge 0$,
therefore it follows that
there are some nodes for which 
the relative height increase will eventually stay at $0$
while other nodes continue increasing their relative height.
The latter will from some point on no longer contribute
to the lower bound and the pattern.
Hence there exists $m_0$ for which 
$\pair{x_1 + m_0 \cdot (x_2-x_1)}{x_2 + m_0 \cdot (x_2-x_1)}$
is a repetition.

\newpage
\appsection{Termination of the Function Layer Translation 
  (Lemma~\ref{lem:transl:termination})}%
  \label{app:transl:termination}

It suffices to show that, for every stream function symbol $\astrfun$ 
and argument position $i$ of $\astrfun$,
the infinite \infioseqspec~$\aioseqspec_{\atrs}$
defined in Def.~\ref{def:infioseqspec}
that has $\trnsli{\astrfun}{i}$ as unique solution for the 
recursion variable $\Xfiq{\strff{f}}{i}{0}$
can be transformed in finitely many steps
into an \infioseqspec~$\aioseqspec'$ with the same solution for the 
recursion variable $\Xfiq{\strff{f}}{i}{0}$ and the property
that only finitely many recursion variables are reachable
in $\aioseqspec'$ from
$\Xfiq{\strff{f}}{i}{0}$.
(This is because from such an \infioseqspec~$\aioseqspec'$
  eventually a finite
  \infioseqspec~$\aioseqspec'_0$ with $\trnsli{\astrfun}{i}$ 
  as unique solution for $\Xfiq{\strff{f}}{i}{0}$ can be extracted.)
An algorithm for this purpose is obtained as follows. 
On the input of $\aioseqspec_{\atrs}$,
set $\aioseqspec\defdby\aioseqspec_{\atrs}$ and 
repeat the following step on $\aioseqspec$ as long as it is applicable:
\begin{description}
  \item[{\sf (RPC)}]
    Detect and remove a reachable \emph{non-consuming pseudo-cycle} from the 
    \infioseqspec~$\aioseqspec\,$:
    Suppose that, for a function symbol $\strff{h}$, for $j,k,l\in\nat$,
    and for a recursion variable $\Xfiq{\strff{h}}{j}{k}$ that is reachable
    from $\Xfiq{\strff{f}}{i}{0}$, we have
    $ \lts{\Xfiq{\strff{h}}{j}{k}}{w}{\Xfiq{\strff{h}}{j}{l}} $
    (from the recursion variable $\Xfiq{\strff{h}}{j}{k}$
      the variable $\Xfiq{\strff{h}}{j}{l}$ is reachable
      via a path in the specification on which the finite~\ioseq~$w$
      is encountered as the word formed by consecutive labels),
    where $l<k$ and $w$ only contains symbols `$\ioout$'.
    Then modify $\aioseqspec$ by setting
    $ \Xfiq{\strff{h}}{j}{k} = \iosqcns{\ioout}{\Xfiq{\strff{h}}{j}{k}} $.
\end{description}
It is not difficult to show that a step {\sf (RPC)} 
preserves weakly guardedness and the unique solution of $\aioseqspec$, 
and that, on input $\aioseqspec_{\atrs}$, 
the algorithm terminates in finitely many steps, 
producing an \infioseqspec~$\aioseqspec'$ 
with $\trnsli{\astrfun}{i}$ as the solution for $\Xfiq{\strff{f}}{i}{0}$
and the property
that only finitely many recursion variables are reachable
in $\aioseqspec'$ from
$\Xfiq{\strff{f}}{i}{0}$.

\newpage
\appsection{Soundness of the Function Layer Translation}%
  \label{app:transl:soundness}

Throughout, let $\atrs = \pair{\Sigma}{R}$ be a fixed stream specification.
First we introduce a few axillary lemmas and definitions.

\begin{definition}\normalfont
  For all stream terms $t \in \ter{\databstr{\asig}}$
  we define by
  \[\npebs{t} \defdby \sup \{ n \in \nat \where t = \trspeb^{k}(t') \}\]
  the \emph{number of leading stream pebbles of $t$}. 
\end{definition}

\begin{definition}\normalfont
  Let $\strff{h} \in \Ssf$ be a a stream function symbol and $j, k \in \nat$.
  We define 
  $\strff{h}_{j,k} \defdby
  \strff{h}(\ldots\trspeb^{\conattop}\ldots,\trspeb^{k}(\sigma),\ldots\trspeb^{\conattop}\ldots) \in \ter{\databstr{\asig}}$
  where the $j$-th argument position is $\trspeb^{k}(\sigma)$
  while all other arguments are $\trspeb^{\conattop}$.
\end{definition}

\begin{definition}\label{def:absprefix}\normalfont
  For terms 
  \begin{align*}
  s &\equiv \strff{g}(\trspeb^{n_1}(\sigma),\ldots,\trspeb^{n_k}(\sigma)) \in \ter{\databstr{\asig}}\\
  t &\equiv \strff{h}(\trspeb^{m_1}(\sigma),\ldots,\trspeb^{m_k}(\sigma)) \in \ter{\databstr{\asig}}
  \end{align*}
  with $n_1,\ldots,n_k,m_1,\ldots,m_k \in \conat$,
  we write $s \sisprefixof t$ if $n_i \le m_i$ for all $i$.
\end{definition}

\begin{lemma}\label{lem:doprefix}
  If $s,t \in \ter{\databstr{\asig}}$ as in Def.~\ref{def:absprefix} with $s \sisprefixof t$
  then $\doLow{\atrs}{s} \le \doLow{\atrs}{t}$.
\end{lemma}

\begin{proof}
  Immediately from the definition.\qed
\end{proof}

We proof the soundness of the function layer translation, that is:
\begin{lemma}
  Let $\atrs$ be a stream definition, and let $\astrfun\in\Ssf$.
  \begin{enumerate}
    \item If\/ $\atrs$ is flat, then it holds:
      $\siosqprd{\trnsl{\astrfun}} = \doLow{\atrs}{\astrfun}$.\\
      Hence, $\doLow{\atrs}{\astrfun}$ is a periodically
      increasing function.
    \item If\/ $\atrs$ is friendly nesting, then it holds:
      $\siosqprd{\trnsl{\astrfun}} \le \doLow{\atrs}{\astrfun}$.
  \end{enumerate}
\end{lemma}

\begin{proof}
  We start with (i).
  Let $\aioseqspec_{\atrs}$ be the weakly guarded \infioseqspec{} (over variables $\setvars$)
  obtained from the translation.
  Let $\eval \funin \setvars \to (\nat \to \conat)$ be the unique solution of $\aioseqspec_{\atrs}$,
  which exists by Lemma~\ref{lem:infioseqspec:ratsol}.
  Then by definition we have
  $\siosqprd{\trnsl{\astrfun}} = \min(\eval(\Xfiq{\astrfun}{1}{0}),\ldots,\eval(\Xfiq{\astrfun}{n}{0}))$.

  Assume $\strff{f}$ is not weakly guarded. 
  There exists $\strff{g} \in \Ssf$
  and
  $\gamma : \strff{f} \mugsfsleadsto \strff{g} \tcugsfsleadsto \strff{g}$
  without production; let $\Sigma_\gamma$ be the set of symbols occurring in $\gamma$
  and $\rho_1$,\ldots,$\rho_n \in R$ the rules applied in $\gamma$ (all having zero production).
  We define $T_\gamma$ as the set of terms $t \in \ter{\databstr{\asig}}$ with $\rootsymb{t} \in \Sigma_\gamma$.
  Accoring to Lemma~\ref{lem:gforce} we can contruct
  a data-exchange function $\sdg$ such that for every term $t \in T_\gamma$
  either $\dg{t}$ is not a redex or a redex with respect to a $\rho_i$.
  Then every reduct of a term $t \in T_\gamma$ in $\aars_\atrs(\sdg)$
  as again in $T_\gamma$ and has therefore no production.
  Hence $\lowf = 0$
  and by definition $\siosqprd{\trnsl{\astrfun}} = 0$ since 
  $\Xfiq{\astrfun}{i}{0} = \Xm$ for all $1 \le i \le \arityS{\strff{f}}$.

  The remaining proof obligation is $\siosqprd{\trnsl{\astrfun}} = \lowf$ for weakly guarded $\strff{f}$.
  We start with $\siosqprd{\trnsl{\astrfun}} \ge \lowf$.
  Let $p_1,\ldots,p_{\arityS{\strff{f}}} \in \nat$
  and define $o \defdby \siosqprd{\trnsl{\astrfun}}(p_1,\ldots,p_{\arityS{\strff{f}}})$.
  We have
  $\eval(\Xfiq{\astrfun}{i}{0})(p_i) = o$
  for some $1 \le i \le \arityS{\strff{f}}$.
  Note that
  \[\low{\strff{f}}(p_1,\ldots,p_i,\ldots,p_{\arityS{\strff{f}}}) 
  \le \low{\strff{f}}(\ldots\conattop\ldots,p_i,\ldots\conattop\ldots)
  = \low{\strff{f}_{i,p_i}}\;.\]
  Hence it suffices to show
  \[\low{\strff{g}_{i,q+p}} \le \eval({\Xfiq{\strff{g}}{i}{q}})(p)\]
  for all $\strff{g} \in \Ssf$, argument positions $1 \le i \le \arityS{\strff{g}} \in \nat$ and $q,p \in \nat$.
  For the case $\eval({\Xfiq{\strff{g}}{i}{q}})(p) = \conattop$ there is nothing to be shown, hence let $\eval({\Xfiq{\strff{g}}{i}{q}})(p) < \conattop$.
  Employing well-founded induction with respect to the ordering $>$ over tuples 
  $$\quadruple{\strff{g}}{i}{q}{p} \text{ with } \eval({\Xfiq{\strff{g}}{i}{q}})(p) < \conattop$$
  we construct a data-exchange function $\sdg$
  together with
  an outermost-fair rewrite sequence in $\aars_\atrs(\sdg)$
  starting from $\strff{g}_{i,p+q}$
  as witness of $\low{\strff{g}_{i,q+p}} \le \eval({\Xfiq{\strff{g}}{i}{q}})(p)$.
  The ordering $>$ is defined as follows:
  \begin{align*}
    \quadruple{\strff{g}}{i}{q}{p} >  \quadruple{\strff{g'}}{i'}{q'}{p'} \Longleftrightarrow
    &\ \strff{g} \text{ is weakly-guarded}\\
    &\ {\wedge}\ (o > o' \;\;\vee\;\; (o = o' \wedge \strff{g} \tcugsfsleadsto \strff{g'}))\\
    &\ \text{where }
      o \defdby \eval({\Xfiq{\strff{g}}{i}{q}})(p), \; 
      o' \defdby \eval({\Xfiq{\strff{g'}}{i'}{q'}})(p')
  \end{align*}

  Let $\sdg$ be an arbitrary data-exchange function.
  The induction basis consists of the tuples $\quadruple{\strff{g}}{i}{q}{p}$ where 
  \begin{itemize}
    \item $\strff{g}$ is not weakly guarded, or
    \item $\eval({\Xfiq{\strff{g}}{i}{q}})(p) = 0$ and $\strff{g}$ is a normal form with respect to $\tcugsfsleadsto$.
  \end{itemize}
  The not weakly guarded symbols have been discussed above,
  therefore let $\strff{g}$ be a normal form with respect to $\tcugsfsleadsto$.
  By definition
  $\Xfiq{\strff{g}}{i}{q} = \iosqInfm{\rho\in\Rsff{\strff{g}}} \; \Xfiqr{\strff{g}}{i}{q}{\rho}$,
  and hence 
  $\eval(\Xfiqr{\strff{g}}{i}{q}{\rho})(p) = 0$
  for some $\rho \in \Rsf$.

  Then $\Xfiqr{\strff{g}}{i}{q}{\rho}$ is either of shape
  \begin{align*}
  (\text{shape 1}) &\ \iosqcat{\ioin^{\cosubtr{\sfscons{\rho}{i}}{q}}}{\iosqcat{\ioout}{\ldots}}
  & \text{or} &&
  (\text{shape 2}) &\ \iosqcat{\ioin^{\cosubtr{(\sfscons{\rho}{i}}{q})+1}}{\iosqcat{\ioout}{\ldots}}
  \end{align*}
  To see this consider the case distinction in the definition of 
  $\Xfiqr{\strff{g}}{i}{q}{\rho}$:
  \begin{align*}
    \Xfiqr{\strff{g}}{i}{q}{\rho}
    =\ 
    &\iosqcat{\ioin^{\cosubtr{\sfscons{\rho}{i}}{q}}}{\iosqcat{\ioout^{\sfsprod{\rho}}}}{ \\
    & \begin{cases}
        \iosqInfm{j\in{\funap{\invfun{\spermut}}{i}}} \:
        \Xfiq{\strff{h}}{j}{(\cosubtr{q}{\sfscons{\rho}{i}}) + \sfsfb{\rho}{j} }
        & \text{ for~(a)} \hfill (C1) \\[0.5ex]
        \ioout^{\cosubtr{q}{\sfscons{\rho}{i}}}{\Xid}
        & \text{ for~(b) if $j = i$} \hfill (C2) \\[0.2ex]
        \Xp
        & \text{ for~(b) if $j \neq i$} \hspace{3ex} \hfill (C3)
      \end{cases}
     }
  \end{align*}
  We have (shape 1) in case of:
  \begin{itemize}
    \item (C1) since $\strff{g}$ being a $\ugsfsleadsto$ normal form and 
          therefore $\sfsprod{\rho} > 0$,
    \item (C2) if $\sfsprod{\rho} + (\cosubtr{q}{\sfscons{\rho}{i}}) > 0$, and
    \item (C3)
  \end{itemize}
  The (shape 2) occurs in case (C2) if $\sfsprod{\rho} + (\cosubtr{q}{\sfscons{\rho}{i}}) = 0$.
  Let $t \equiv \strff{g}_{i,q+p}$ and $\sdg' = \sdg(t,\rho)$, then:
  \begin{itemize}
  \item \emph{Assume $\Xfiqr{\strff{g}}{i}{q}{\rho}$ is of (shape 1).}\vspace{1ex}\\
    From $\eval(\Xfiqr{\strff{g}}{i}{q}{\rho})(p) = 0$
    it follows $\cosubtr{\sfscons{\rho}{i}}{q} > p$ and hence $\sfscons{\rho}{i} > q + p$.
    $\sdg'(t)$ cannot be a $\rho$-redex
    since $\rho$ has at argument position $i$ a consumption greater than $q + p$.
    Hence $\strff{g}_{i,q+p}$ is a normal form in $\aars_\atrs(\sdg')$  by Lem.~\ref{lem:gforce}
    and it follows that $\low{\strff{g}_{i,q+p}} = 0$.\vspace{1ex}
  \item \emph{Assume $\Xfiqr{\strff{g}}{i}{q}{\rho}$ is of (shape 2).}\vspace{1ex}\\
    Then we have $\sfsprod{\rho} + (\cosubtr{q}{\sfscons{\rho}{i}}) = 0$, 
    hence $\sfsprod{\rho} = 0$ and $q \le \sfscons{\rho}{i}$.
    From $\eval(\Xfiqr{\strff{g}}{i}{q}{\rho})(p) = 0$
    it follows that $(\cosubtr{\sfscons{\rho}{i}}{q})+1 = \sfscons{\rho}{i} - q + 1> p$
    and therefore $\sfscons{\rho}{i} \ge q + p$.
    Then
    $\strff{g}_{i,q+p}$ is a normal form if $\sfscons{\rho}{i} > q + p$ 
    or
    $\strff{g}_{i,q+p} \to \sigma$ if $\sfscons{\rho}{i} = q + p$ 
    in $\aars_\atrs(\sdg')$;
    in both cases $\low{\strff{g}_{i,q+p}} = 0$.
  \end{itemize}
  This concludes the base case of the induction.

  Now we consider the induction step.
  By definition
  $\Xfiq{\strff{g}}{i}{q} = \iosqInfm{\rho\in\Rsff{\strff{g}}} \; \Xfiqr{\strff{g}}{i}{q}{\rho}$,
  and hence 
  $\eval(\Xfiq{\strff{g}}{i}{q})(p) = \eval(\Xfiqr{\strff{g}}{i}{q}{\rho})(p)$
  for some $\rho \in \Rsf$.
  Let $t \equiv \strff{g}_{i,q+p}$ and $\sdg' = \sdg(t,\rho)$.
  Again we consider the case distinction (C1),(C2) and (C3) in the definition of $\Xfiqr{\strff{g}}{i}{q}{\rho}$.

  If $\cosubtr{\sfscons{\rho}{i}}{q} > p$
  then in all cases $\eval(\Xfiqr{\strff{g}}{i}{q}{\rho})(p) = 0$ and 
  $\strff{g}_{i,q+p}$ is a normal form in $\aars_\atrs(\sdg')$,
  consequently $\low{\strff{g}_{i,q+p}} = 0$.
  Therefore assume $\cosubtr{\sfscons{\rho}{i}}{q} \le p$.
  \begin{itemize}
    \item [(C1)]
      $\Xfiqr{\strff{g}}{i}{q}{\rho} =
              \iosqcat{\ioin^{\cosubtr{\sfscons{\rho}{i}}{q}}}{\iosqcat{\ioout^{\sfsprod{\rho}}}}
              \iosqInfm{j\in{\funap{\invfun{\spermut}}{i}}} \:
              \Xfiq{\strff{h}}{j}{(\cosubtr{q}{\sfscons{\rho}{i}}) + \sfsfb{\rho}{j} }$\vspace{1ex}\\
      There exists $j \in \funap{\invfun{\spermut}}{i}$ with
      $\eval(\Xfiqr{\strff{g}}{i}{q}{\rho})(p) = e^\eval(p)$ where
      \begin{align*}
        e  &\defdby
              \iosqcat{\ioin^{\cosubtr{\sfscons{\rho}{i}}{q}}}{\iosqcat{\ioout^{\sfsprod{\rho}}}}
              \Xfiq{\strff{h}}{j}{q'} &
        q' &\defdby (\cosubtr{q}{\sfscons{\rho}{i}}) + \sfsfb{\rho}{j} \punc.
      \end{align*}
      Since $\cosubtr{\sfscons{\rho}{i}}{q} \le p$, let $p' \defdby p - (\cosubtr{\sfscons{\rho}{i}}{q})$. We get 
      \begin{align}
        \eval(\Xfiqr{\strff{g}}{i}{q}{\rho})(p) &= \sfsprod{\rho} + \eval(\Xfiq{\strff{h}}{j}{q'})(p') 
        \label{eq:giqp}\punc.
      \end{align}
      $\sdg'(\strff{g}_{i,q+p})$ is a $\rho$-redex by Lem.~\ref{lem:gforce}, hence 
      $$\strff{g}_{i,q+p} 
        \to_{\aars_\atrs(\sdg')} 
        \trspeb^{\sfsprod{\rho}}(\strff{h}(\trspeb^{n_1}(\sigma),\ldots,\trspeb^{n_{\arityS{\strff{h}}}}(\sigma)))$$
     where by definition of $\rho$:
     \begin{align*}
       n_j &= \sfsfb{\rho}{j} + (q + p - \sfscons{\rho}{i}) = q' + p'\punc.
     \end{align*}
     Therefore
     $$\strff{h}(\trspeb^n_1(\sigma),\ldots,\trspeb^n_{\arityS{\strff{h}}}(\sigma))
       \sisprefixof 
       \strff{h}_{j,q'+p'}\punc.$$
     We have $\eval(\Xfiqr{\strff{g}}{i}{q}{\rho})(p) > \eval(\Xfiqr{\strff{h}}{j}{q'}{\rho})(p')$
     or $\strff{g} \ugsfsleadsto \strff{h}$ by \eqref{eq:giqp},
     therefore we get
     $\low{\strff{h}_{j,q'+p'}} \le \eval({\Xfiq{\strff{h}}{j}{q'}})(p')$
     by induction hypothesis and
     \begin{gather*}
       \low{\strff{g}_{i,q+p}} \le \sfsprod{\rho} + \low{\strff{h}_{j,q'+p'}} \\
       \le \sfsprod{\rho} + \eval({\Xfiq{\strff{h}}{j}{q'}})(p') 
           = \eval(\Xfiqr{\strff{g}}{i}{q}{\rho})(p)
           = \eval(\Xfiq{\strff{g}}{i}{q})(p)
     \end{gather*}
     which proves the claim.
  \vspace{1ex}
  \item [(C2)]
      $\Xfiqr{\strff{g}}{i}{q}{\rho} =
              \iosqcat{\ioin^{\cosubtr{\sfscons{\rho}{i}}{q}}}{\iosqcat{\ioout^{\sfsprod{\rho}}}}
              \ioout^{\cosubtr{q}{\sfscons{\rho}{i}}}{\Xid}$\vspace{1ex}\\
      Then
      \begin{align*}
        \eval(\Xfiqr{\strff{g}}{i}{q}{\rho})(p) 
        &= \sfsprod{\rho} + (\cosubtr{q}{\sfscons{\rho}{i}}) + \eval(\Xid)(p- (\cosubtr{\sfscons{\rho}{i}}{q}))\\
        &= \sfsprod{\rho} + (\cosubtr{q}{\sfscons{\rho}{i}}) + (p - (\cosubtr{\sfscons{\rho}{i}}{q}))\\
        &= \sfsprod{\rho} + q + p - \sfscons{\rho}{i}
      \end{align*}
      Furthermore in case (C2) $\rho$ is collapsing with $i = j$, therefore
      $$\strff{g}_{i,q+p} \to_{\aars_\atrs(\sdg')} \trspeb^{\sfsprod{\rho} + (q + p - \sfscons{\rho}{i})}(\sigma)\punc.$$
      Hence 
      $$\low{\strff{g}_{i,q+p}} 
        \le \sfsprod{\rho} + q + p - \sfscons{\rho}{i} 
        = \eval(\Xfiqr{\strff{g}}{i}{q}{\rho})(p)\punc.$$
  \vspace{1ex}
  \item [(C3)]
      $\Xfiqr{\strff{g}}{i}{q}{\rho} =
              \iosqcat{\ioin^{\cosubtr{\sfscons{\rho}{i}}{q}}}{\iosqcat{\ioout^{\sfsprod{\rho}}}}{\Xp}$\vspace{1ex}\\
      Then $\eval(\Xfiqr{\strff{g}}{i}{q}{\rho})(p) = \conattop = \low{\strff{g}_{i,q+p}}$
      since $\rho$ is collapsing with $i \ne j$.
  \end{itemize}
  This concludes the direction $\siosqprd{\trnsl{\astrfun}} \ge \lowf$.

  The remaining proof obligation is $\siosqprd{\trnsl{\astrfun}} \le \lowf$.
  Let $p_1,\ldots,p_{\arityS{\strff{f}}} \in \nat$
  and define 
  $o \defdby \low{\strff{f}}(p_1,\ldots,p_{\arityS{\strff{f}}}) 
     = \low{\strff{f}(\trspeb^{p_1},\ldots,\trspeb^{p_{\arityS{\strff{f}}}})}$.
  Assume $o < \conattop$, otherwise there is nothing to be shown.
  Then there exists a data-exchange function $\sdg$
  and a finite reduction sequence
  $$\strff{f}(\trspeb^{p_1}(\sigma),\ldots,\trspeb^{p_{\arityS{\strff{f}}}}(\sigma))
    \to^n_{\aars_\atrs(\sdg')} \trspeb^o(t)$$
  such that $\low{t} =  0$ and in particular
  \begin{itemize}
    \item $t$ is a normal form in $\aars_\atrs(\sdg')$, or
    \item $\strff{g}$ is not weakly guarded.
  \end{itemize}
  W.l.o.g.\ choose $\sdg$ and the reduction sequence such that $n$ is minimal.
  We apply induction on the length of the reduction $n$ to show
  \begin{itemize}
    \item There exists $1 \le i \le \arityS{\strff{f}}$ such that
          $$\low{\strff{f}(\trspeb^{p_1}(\sigma),\ldots,\trspeb^{p_{\arityS{\strff{f}}}}(\sigma))}
           = \low{\strff{f}(\trspeb^{\conattop},\ldots,\trspeb^{p_i},\ldots,\trspeb^{\conattop})}\punc,$$
    \item and 
          $\siosqprd{\trnsl{\astrfun}} 
           \le \low{\strff{f}(\trspeb^{p_1}(\sigma),\ldots,\trspeb^{p_{\arityS{\strff{f}}}}(\sigma))}$.
  \end{itemize}

\end{proof}

\newpage
\appsection{Soundness of the Stream Layer Translation 
  (Lemma~\ref{lem1:outsourcing} and Lemma~\ref{lem2:outsourcing})}%
  \label{prf:outsourcing}%

\subsection{Proof Sketch of Lemma~\ref{lem1:outsourcing}}

We only sketch a proof of statement~\ref{lem1:outsourcing:item:1} of
Lem.~\ref{lem1:outsourcing} because item~\ref{lem1:outsourcing:item:2}
of this lemma can be demonstrated analogously.

\begin{proof}[of Lem.~\ref{lem1:outsourcing}, \ref{lem1:outsourcing:item:1}]
    Let $\atrs$ be an \SCS, and let $\afam = \{\fgate\}_{\astrfun\in\Ssf}$
    be a family of gates such that, 
    for all $\astrfun\in\Ssf$, $\sgateprd{\fgate} = \doLow{\atrs}{\astrfun}$.
    We give the ideas for the proofs of the inequations
    ``$\le$'' and ``$\ge$'' that make part of the equation
    \begin{equation}
      \label{eq1:lem:clpatmat}
      \doLow{\atrs}{\rootsc} =  \netprd{\trnslF{\rootsc}{\afam}} \punc.
    \end{equation}
        
    We assume, without loss of generality, that every $\strff{f}\in\Ssf$ has
    precisely one occurrence in the \SCS-layer.
    Note that, by replacing for every $\strff{f}\in\Ssf$ the individual occurrences of $\strff{f}$
    in the \SCS-layer of $\atrs$ by individual occurrences of the symbols 
    $\strff{f}$, $\strff{f'}$, $\strff{f''}$, \ldots,
    respectively, and
    introducing defining rules for $\strff{f'}$, $\strff{f''}$, \ldots\ in the \SFS-part
    analogous to the defining rules for $\strff{f}$ that are already present there,
    an \SCS~$\atrs'$ is obtained that satisfies the restriction,
    and for which it is easy to establish that it holds:
    $ \doLow{\atrs}{\rootsc} = \doLow{\atrs'}{\rootsc} $,
    $ \trnslF{\rootsc}{\afam} = \trnslF{\rootsc}{\afam'} $,
    and hence
    $ \netprd{\trnslF{\rootsc}{\afam}} = \netprd{\trnslF{\rootsc}{\afam'}} $,
    where $\afam' = \{ \fgate \}_{\astrfun\in\Ssf'}$ is the family of gates 
    ($\Ssf'$ is the stream function signature of $\atrs'$)
    by adding to $\afam$, for all symbols $\astrfun\in\Ssf$, 
    the corresponding copies $\fpgate, \fppgate, \ldots$ of $\fgate$.
    Note that we also have, for all $\strff{f'}\in\Ssf'$, 
    $\sgateprd{\fpgate} = \doLow{\atrs}{\strff{f'}}$.
    \begin{description}
      \item[\mbox{\normalfont ``$\le$'':}]
        We only have to consider the case
        in which for $ k\defdby \netprd{\trnslF{\rootsc}{\afam}}$
        it holds $k<\infty$. Hence we assume that $k<\infty$. 
        Then there exists a finite rewrite sequence
        \begin{equation*}
          \arewseq : \;  \trnslF{\rootsc}{\afam}  \mpebred  \netpebn{k}{\anet'} \punc.
        \end{equation*}
        Then with $\anet \defdby \trnslF{\rootsc}{\afam}$ 
        it follows that $\anet'$ has the same `gate structure' as $\anet$,
        but possibly the \ioseq{s} in the gates have changed, and there
        are likely to appear many pebbles queuing inside $\anet'$. 
        By the definition of $k$ it follows that $\netprd{\anet'} = 0$.
        Without loss of generality, we assume that during $\arewseq$
        all extractable pebbles have been extracted from\
        `inside' the gates in $\anet'$,
        that is, 
        there is no occurrence of a redex in $\anet'$ with respect to the pebbleflow rule
        $\netmeet{\netpeb{\ianet{1}}}{\netpeb{\ianet{2}}} \red \netpeb{\netmeet{\ianet{1}}{\ianet{2}}}$.
        Note that the only occurrences of the symbol $\snetmeet$ in terms of $\arewseq$
        are part of gates.
        (If this would not be the case for $\arewseq$, then we would
         extend $\arewseq$ by finitely many $\spebred$-steps 
         to a rewrite sequence $\arewseq'$ which has this property,
         and base the further argumentation on $\arewseq'$.)

        Due to the assumption that function symbols in $\Ssf$ only occur
        once in the \SCS-layer, to every $\strff{f}\in\Ssf$
        there corresponds in $\anet$ 
        (via the net translation in Def.~\ref{def:trnsl:nets})
        a unique gate  $\fgate = \trnsl{\strff{f}}$. 
        We now monitor, for every gate $\fgate$ in
        $\anet$, the consumption/production behaviour during $\arewseq$.
        Let $\strff{f}\in\Ssf$.
        Suppose that during $\arewseq$, $\fgate$ consumes
        $n_1$, \ldots, $n_{\arityS{\strff{f}}}$ pebbles from its input
        components, respectively, and that in total it produces $n$ pebbles.
        By the assumption that no extractable pebbles are stuck inside
        $\fgate$ at the end of $\arewseq$, 
        in particular,
        $\gateprd{\fgate}{n_1,\ldots,n_{\arityS{\strff{f}}}} = n $ follows.
        Moreover, by the assumption that the data-oblivious lower bound 
        $\doLow{\atrs}{\strff{f}}$ of a symbol $\astrfun\in\Ssf$ coincides 
        with the production function of the gate $\fgate$, we get
        $\doLow{\atrs}{\strff{f}}(n_1,\ldots,n_{\arityS{\strff{f}}}) = n $.
        As a consequence, there is a rewrite sequence:
        \begin{equation*}
          \tilde{\arewseq}_{\astrfunsub} : \;
          \funap{\strff{f}}{ 
            \strcns{\pebble^{n_1}}{\astrvar_1}, 
            \ldots,
            \strcns{\pebble^{n_{\arityS{\strff{f}}}}}{\astrvar_{\arityS{\strff{f}}}},
            \underbrace{\pebble,\ldots,\pebble}_{\arityD{\strff{f}}}
          }
          \imred{\dgars{\atrs}{\sdg}}
          \strcns{\pebble^n}{\bstrtrm}
        \end{equation*}
        in an ARS~$\dgars{\atrs}{\sdg}$
        such that $\tilde{\arewseq}_{\astrfunsub}$ can be extended
        to an outermost-fair rewrite sequence that does not produce
        further pebbles.
        Then from $\tilde{\arewseq}_{\astrfunsub}$ we extract a rational gate
        $\tilde{\agate}_{\astrfunsub}$ with the properties:
        (a)~$\tilde{\agate}_{\astrfunsub}$ precisely describes the
            consumption/production behaviour that happens during 
            $\tilde{\arewseq}_{\astrfunsub}$,
        (b)~thereafter $\tilde{\agate}_{\astrfunsub}$ continues
            like the gate corresponding to $\strff{f}$ in $\anet'$.
        We gather the new gates in the new family of gates:
        $ \tilde{\afam} \defdby \{ \tilde{\agate}_{\astrfunsub} \}_{\strff{f}\in\Ssf} $.
        Note that, for all $\strff{f}\in\Ssf$, 
        $\sgateprd{\agate_{\astrfunsub}} \le \sgateprd{\tilde{\agate}_{\astrfunsub}}$
        holds.
        Nevertheless, it can be shown that there is actually also
        a rewrite sequence:
        \begin{equation*}
          \tilde{\arewseq} : \; 
          \trnslF{\rootsc}{\tilde{\afam}}  
          \mpebred  
          \netpebn{k}{\tilde{\anet}'} 
          \punc.
        \end{equation*}
        such that $\netprd{\tilde{\anet}'}=0$. 
        Using this pebbleflow rewrite sequence $\tilde{\arewseq}$
        as well as the rewrite sequences 
        $\tilde{\arewseq}_{\strff{f}}$ on the abstracted versions of $\atrs$,
        it is possible to obtain a rewrite sequence in 
        an~ARS~$\dgars{\atrs}{\sdg}$ of the form:
        \begin{equation*}
          \bar{\arewseq} : \;
          \treeunfT{\atrs}{\rootsc}
          \infred_{\dgars{\atrs}{\sdg}}^{m\omega}
          \strcns{\pebble^k}{\astrtrm} 
        \punc,  %
        \end{equation*}
        where $m$ is the length of the reduction sequence $\tilde{\arewseq}$,
        that can be extended in $\dgars{\atrs}{\sdg}$
        to a rewrite sequence $\bar{\bar{\arewseq}}$
        which consists of outermost-fair sequences of rewrite sequences of length $\omega$
        in which respectively, for some $\bstrfun\in\Ssf$,
        $\bstrfun$-redexes are reduced uniformously in all unfoldings,
        and such that furthermore $\bar{\bar{\arewseq}}$
        does not contain terms with a prefix of more than $k$ pebbles. 
        The existence of the rewrite sequence $\bar{\bar{\arewseq}}$
        on $\rootsc$ in an ARS~$\dgars{\atrs}{\sdg}$
        demonstrates that $\doLow{\atrs}{\rootsc}\le k$.
        \vspace{1ex}
      \item[\mbox{\normalfont ``$\ge$'':}]
        This direction of the proof is a consequence of three statements
        that are mentioned below. 
        First, it holds that:
        \begin{equation}
          \label{eq1:prf:lem:clpatmat}
          \text{For all $\anet\in\net\,$: } \;\;\;
               \netprd{\treeunf{\anet}} = \netprd{\anet} \; ,
        \end{equation}
        where $\treeunf{\anet}$ denotes the infinite unfolding of $\anet$
        into a possibly infinite `pebbleflow tree'.
        Second: 
        \begin{equation}
          \label{eq2:prf:lem:clpatmat}
          \text{For all stream terms $\astrtrm$ in $\atrs$: } \;\;\;
               \terprd{\atrs}{\astrtrm} =  \terprd{\atrs}{\treeunfT{\atrs}{\astrtrm}} 
          \punc,
        \end{equation}
        where by $\astrtrm\in\ter{\asig}$ we mean the result of
        infinitely unfolding all stream constants in $\astrtrm$. 
        Finally, third, it holds that:
        \begin{equation}
          \label{eq3:prf:lem:clpatmat}
          \text{For all stream terms $\astrtrm$ in $\atrs$: } \;\;\;
          \doLow{\atrs}{\treeunfT{\atrs}{\astrtrm}}
            \ge \netprd{\treeunf{\trnslF{\astrtrm}{\afam}}} \punc.
        \end{equation}
        Putting \eqref{eq1:prf:lem:clpatmat}, 
                \eqref{eq2:prf:lem:clpatmat},
                and \eqref{eq3:prf:lem:clpatmat} together,
        we obtain:
        \begin{equation*}
          \doLow{\atrs}{\rootsc} 
            = \doLow{\atrs}{\treeunfT{\atrs}{\rootsc}}
            \ge \netprd{\treeunf{\trnslF{\rootsc}{\afam}}}
            = \netprd{\trnslF{\rootsc}{\afam}} \punc,   
        \end{equation*}
        which shows  ``$\ge$''. 
    \end{description}
    This concludes our sketch of the proof of Lem.~\ref{lem1:outsourcing}~\ref{lem1:outsourcing:item:1}. 
\end{proof}

\subsection{Proof of Lemma~\ref{lem2:outsourcing}}

Let $\atrs$ be a stream specification, and let\/
$\afam = \{\fgate\}_{\astrfun\in\Ssf}$ be a family of gates
such that, for all\/ $\astrfun\in\Ssf$, the arity of\/ $\fgate$
equals the stream arity of\/ $\astrfun$. 
Suppose that one of the following statements holds:
\begin{enumerate}[(a)]
  \item $\siosqprd{\fgate} \le \doLow{\atrs}{\astrfun}$ for all\/ $\strff{f} \in \Ssf\,$;
    \vspace*{0.5ex}
  \item $\siosqprd{\fgate} \ge \doLow{\atrs}{\astrfun}$ for all\/ $\strff{f} \in \Ssf\,$;
    \vspace*{0.5ex}
  \item $\doUp{\atrs}{\astrfun} \le \siosqprd{\fgate}$ for all\/ $\strff{f} \in \Ssf\,$;
    \vspace*{0.5ex}
  \item $\doUp{\atrs}{\astrfun} \ge \siosqprd{\fgate}$ for all\/ $\strff{f} \in \Ssf\,$.
\end{enumerate}
In this section we show that then, for all $\astrcon\in\Ssc$,
the corresponding one of the following statements holds:
\begin{enumerate}[(a)]
  \item $\netprd{\trnslF{\astrcon}{\afam}} \le \doLow{\atrs}{\astrcon}\,$;
    \vspace*{0.5ex}
  \item $\netprd{\trnslF{\astrcon}{\afam}} \ge \doLow{\atrs}{\astrcon}\,$;
    \vspace*{0.5ex}
  \item $\doUp{\atrs}{\astrcon} \le \netprd{\trnslF{\astrcon}{\afam}}\,$;
    \vspace*{0.5ex}
  \item $\doUp{\atrs}{\astrcon} \ge \netprd{\trnslF{\astrcon}{\afam}}\,$.
\end{enumerate}

\subsubsection*{Multiple Numbered Contexts}

\begin{definition}\label{def:stream:context}
  Let $\atrs = \pair{\asig}{R}$ be a stream specification, and $k,l \in \nat$.
  A (multiple numbered) \emph{stream context} $\acxt$ 
  over stream holes $\holeS_1,\ldots,\holeS_k$,
  and data holes $\holeD_1,\ldots,\holeD_{\ell}$ 
  is a (stream) term 
  \[
    \acxt \in 
    \iter{\asig \uplus \{\holeS_1,\ldots,\holeS_k\} 
               \uplus \{\holeD_1,\ldots,\holeD_{\ell}\},X}_\sortS
  \]
  where the $\ofsort{\holeS_i}{\sortS}$ and $\ofsort{\holeD_i}{\sortD}$ 
  are distinct constant symbols not occurring in $\asig$.

  We denote by $C[t_1,\ldots,t_k; u_1,\ldots,u_{\ell}]$ 
  the result of replacing the occurrences of 
  $\holeS_i$ by $s_i$ and $\holeD_j$ by $t_j$ in $\acxt_\atrs$, respectively.
\end{definition}

\begin{remark}
  Each of the hole symbols 
  $\holeS_1,\ldots,\holeS_k$, and $\holeD_1,\ldots,\holeD_{\ell}$ 
  may have multiple occurrences in a stream context $\acxt$ over 
  $\holeS_1,\ldots,\holeS_k$ and $\holeD_1,\ldots,\holeD_{\ell}$,
  and this includes the case that it does not occur in $\acxt$ at all.
\end{remark}

\begin{example}
  Let $\atrs$ be a stream specification with data constants $\msf{b}$ and $\msf{c}$,
  a unary stream function symbol $\astrfun$,
  and a binary stream function symbol $\bstrfun$.
  Then the expression
  \(
    \strcns{\msf{b}}{
      \funap{\astrfun}{
        \strcns{\msf{c}}{
          \bfunap{\bstrfun}{
            \strcns{\holeD_1}{\holeS_1}}{
              \bfunap{\bstrfun}{\holeS_2}{
                \bfunap{\bstrfun}{\holeS_1}{\holeS_3}
              }
            }
          }
        }
      }
  \)
  is a stream context over $\holeS_1$, $\holeS_2$, $\holeS_3$ and $\holeD_1$.
  Also $\funap{\astrfun}{\strcns{\msf{b}}{\holeS_3}}$
  is a stream context over $\holeS_1$, $\holeS_2$, $\holeS_3$ and $\holeD_1$.
  But $\strcns{\msf{b}}{\strcns{\holeS_1}{\astr}} $ is not a stream context
  over $\holeS_1$, because the symbol $\holeS_1$ occurs at a 
  `data position' (this is excluded by many-sortedness, 
  $\ofsort{\sstrcns}{\sortD \to \sortS \to \sortS}$ 
  and $\ofsort{\holeS_1}{\sortS}$).
\end{example}

We extend the \daob{} production function of stream function symbols $\astrfun \in \Ssf$,
Def.~\ref{def:dorange}, to multiple numbered contexts.

\begin{definition}\normalfont\label{def:dorange:ctx}
  Let $\atrs = \pair{\asig}{R}$ be a stream specification,
  and
  $\acxt$ a multiple numbered stream context
  over stream holes $\holeS_1,\ldots,\holeS_k$,
  and data holes $\holeD_1,\ldots,\holeD_{\ell}$.
  The \emph{\daob{} production range 
  $\doRng{\atrs}{\acxt} \funin \nat^{k} \to \conat$ 
  of the context $\acxt$} is:
  \[
    \funap{\doRng{\atrs}{\acxt}}{n_1,\ldots,n_{k}}
    \defdby
    \doRng{\atrs}{\,\funap{\databstr{\acxt}}{
      (\strcns{\trspeb^{n_1}}{\astr}),\ldots,(\strcns{\trspeb^{n_{k}}}{\astr})}\,}
    \punc,
    \]
    \\[-3ex]
  where $\strcns{\trspeb^{m}}{\astr} \defdby \overbrace{\trspeb\xstrcns\ldots\xstrcns\trspeb\,\sstrcns}^{\text{$m$ times}}\,\astr$.

  The \emph{\daob{} lower and upper bounds on the production of $\strff{g}$}
  are defined by $\doLow{\atrs}{\acxt} \defdby \inf(\doRng{\atrs}{\acxt})$ and
  $\doUp{\atrs}{\acxt} \defdby \sup(\doRng{\atrs}{\acxt})$, respectively.
\end{definition}
Note that it is indeed an extension of Def.~\ref{def:dorange}, for every $\bstrfun \in \Ssf$ we have 
\[\doRng{\atrs}{\bstrfun} =
  \doRng{\atrs}{\funap{\bstrfun}{\holeS_1,\ldots,\holeS_{\arityS{\bstrfun}},\holeD_1,\ldots,\holeD_{\arityD{\bstrfun}}}}
  \punc.
\]

The following lemma states that data-oblivious lower and upper bounds
are reachable.
\begin{lemma}
  Let $\atrs = \pair{\asig}{R}$ be a stream specification, and $s \in \iter{\asig}_\sortS$.
  
\end{lemma}

\subsubsection*{Colouring Terms, Tracing Data Flow}

For the purpose of tracing the flow of data-elements during rewrite sequences, we introduce
auxiliary `tracker symbols' and partition terms into `coloured' multiple numbered contexts.
Usual rewriting will then be allowed only within contexts.
Rewrite steps crossing the border between contexts are prohibited.
In order to circumvent the loss of possible reduction steps,
we introduce rules for exchanging pebbles between neighboring contexts
(recolouring of the pebbles).

We introduce
auxiliary, unary function symbols $\stracker{c}$ and $\strackerx{c}$
in the pebbleflow rewrite system as well as in the stream specification $\atrs$.
The symbols $\stracker{c}$ mark the top of a context of colour $c$
and at the same time terminate the context above it.
The symbols $\strackerx{c}$ only close, but do not open a context.
We call them \emph{tracker symbols} in the sequel,
since they are also used to track pebble/data exchange 
between neighboring contexts.
Therefore we enrich the pebbleflow rewrite system 
with the following rules
\begin{align}\label{rules_tracker_pf}
  \tracker{c}{\netpeb{\anet}} &\to \netpeb{\tracker{c}{\anet}}
  &
  \trackerx{c}{\netpeb{\anet}} &\to \netpeb{\trackerx{c}{\anet}}
\end{align}
and the stream specification $\atrs$ with rules
\begin{align}\label{rules_tracker_scs}
  \tracker{c}{\strcns{t}{\aseq}} &\to \strcns{t}{\tracker{c}{\aseq}} 
  &
  \trackerx{p}{\strcns{t}{\aseq}} &\to \strcns{t}{\trackerx{c}{\aseq}}
\end{align}
for every function symbol $\stracker{c}$ and $\strackerx{c}$.
We define $\ofsort{\stracker{c},\strackerx{c}}{\sortstream \to \sortstream}$,
that is, the tracker symbols cannot be placed at data positions.

\begin{definition}\normalfont
  For  terms $s,t$ we say that \emph{$t$ is an enrichment of $s$ with tracker symbols},
  denoted $\enrichment{s}{t}$,
  if $s$ can be obtained from $t$ by removing tracker symbols,
  that is, $t \mred s$ via rules $\tracker{c}{u} \to u$ and $\trackerx{c}{u} \to u$.
\end{definition}

\begin{lemma}
  Every pebbleflow reduction can be lifted to enrichted terms:
  \[\myall{s,s' \in \net, s \mpebred s', \enrichment{s}{t}}{\myex{t'}{t \mpebredt t' \wedge \enrichment{s'}{t'}}}\]

and stream specifications $\atrs$
\end{lemma}

Every pebbleflow rewrite sequence and every stream specification rewrite sequence 
can be lifted to enriched terms,
i.e.\ 
$\myall{\enrichment{s' \mredi s}{t}}{\myex{t'}{t \leadsto^* t' \wedge \enrichment{s'}{t'}}}$
where ${\leadsto} = \setunion{\sred}{\sired{\text{\eqref{rules_tracker_pf}}}}$ 
(or ${\leadsto} = \setunion{\sred}{\sired{\text{\eqref{rules_tracker_scs}}}}$).
Sometimes the tracker symbols are in the way of rule application,
but, because of the special structure of pebbleflow (resp.\ SCS) rules,
it is always possible to move them away using rules~\eqref{rules_tracker_pf} 
(resp.~\eqref{rules_tracker_scs}).

Now we sketch the proof for the inequalities ``$\le$'' (i) and ``$\ge$'' (ii).
\begin{description}
  \item[``$\le$'':]
We need to show 
  $\netprd{\trnslF{M}{\afam}} \le \terprd{\atrs}{M}$.
For the proof of (i) we map a given finite pebbleflow rewrite sequence on $\rattrnsl{M}{\mcl{F}}$ producing $n$ pebbles
to a $\atrs$~rewrite sequence on the infinite unfolding of $M$ 
that produces at least $n$ data elements.
(For (ii) the proof proceeds the other way arround.)

Thereby not every rewrite step will be mapped,
but we trace `essential' steps
and map them to a many-step (possibly infinite) reduction on `the other side'.
For this purpose we partition the pebbleflow term $\rattrnsl{M}{\mcl{F}}$
as well as infinite unfolding of $M$ into coloured contexts
(using the same colours for the net and the unfolding of $M$).
The partitioning is arranged in such a way that
contexts of the same colour have the same production function.
In the sequel we will refer to this property
as the \emph{context production property} (CPP).

Usual rewriting will then be allowed only within contexts.
Rewrite steps crossing the border between contexts are prohibited.
In order to circumvent the loss of possible reduction steps,
we introduce additional rules for exchanging pebbles between neighboring contexts
(recolouring of the pebbles).

The construction of the mapped reduction sequence
will be driven by the goal to uphold CPP.
It can be shown pebbleflow rewriting and \SCS\ rewriting
does not change the production function of a context.
Hence in order to maintain CPP
it is not necessary to map internal steps ocurring within contexts.
The `essential' steps are the exchange steps between the contexts.

We use $\ccontextPF{c}$
to refer to the context of colour $c$ in the pebbleflow net $\rattrnsl{M}{\mcl{F}}$.
Likewise we use $\ccontextM{c}$
for the context of colour $c$ in 
the infinite unfolding of the stream constant $M$.
Although there may be multiple, typically infinitely many, occurrences of a context with colour $c$ on the term level,
contexts with the same colour are syntactically equal.
Initially all contexts are of one of the following (simple) shapes:

\begin{figure}
\begin{center}
\renewcommand{\arraystretch}{1.2}
\begin{tabular}{|c|c|c|}
\hline
case & $\ccontextPF{c}$ & $\ccontextM{c}$\\
\hline
1 & $\netpeb{\Box_1}$ & $\strcns{d}{\Box_1}$\\
\hline
2 & $\netmeetn{r}{\netbox{\aiosq_{\msf{f},1}}{\Box_1},\ldots,\netbox{\aiosq_{\msf{f},r}}{\Box_r}}$ & $\msf{f}(\Box_1,\ldots,\Box_r)$\\
\hline
\end{tabular}
\caption{Initial context configurations}
\label{}
\end{center}
\end{figure}

Clearly $\snetprd_{\ccontextPF{c}} = \snetprd_{\ccontextM{c}}$ for case 1.
An important step in the proof is to show
that the translation of functions symbols into rational gates
preserves the quantitative production behaviour,
yielding $\snetprd_{\ccontextPF{c}} = \snetprd_{\ccontextM{c}}$ for case 2.

The proof now continues as follows.
Every pebble exchange step in the pebbleflow net is mapped
to an infinite rewrite sequence on the corresponding coloured \SCS\ term.
For the definition of these corresponding rewrite steps we
apply the assumption $\snetprd_{\ccontextPF{c}} = \snetprd_{\ccontextM{c}}$
guaranteeing that a data element can also be extracted from the context on \SCS\ term level,
and that after the step 
we have again $\snetprd_{\ccontextPF{c}} = \snetprd_{\ccontextM{c}}$ 
for all colours $c$.
In this way, we define a rewrite sequence $\rho_\atrs$ in $\atrs$
by induction on the number of rewrite steps of the given pebbleflow reduction $\rho$
such that $\rho_\atrs$ and $\rho$ produce the same amount of data elements and pebbles, respectively.
  \item[``$\ge$'':]
    Here we need to show
    $\netprd{\trnslF{M}{\afam}} \ge \terprd{\atrs}{M}$.
The proof of this inequality proceeds similarly to that of ``$\le$''.
Let $\rho_\atrs$ in $\atrs$ be a reduction sequence in $\atrs$.
In order to map $\rho_\atrs$ back to a pebbleflow rewrite sequence $\rho$
it is crucial that the contexts of one colour are changed synchronously and in the same way.
It cannot be assumed that $\rho_\atrs$ posesses this property.
For this purpose $\rho_\atrs$ is transformed into a sequence $\rho_\atrs'$ of complete developments.
This is always possible since complete developments in orthogonal TRSs are a cofinal rewrite strategy.
For mapping $\rho_\atrs'$ to a rewrite sequence $\rho$ with the same production on the pebbleflow net,
we can then proceed by induction on the length of $\rho_\atrs'$,
in converse direction to the argument for ``$\le$'', above.
\end{description}

\end{document}